\documentclass[letter,11pt]{article}
\usepackage[margin=25mm]{geometry}
\usepackage{setspace}
\usepackage{graphicx,subcaption,booktabs,multirow,url,tabu}
\usepackage{microtype}
\usepackage{amssymb}
\usepackage{enumitem}

\usepackage{amsfonts,amsmath,amsthm}
\newtheorem{theorem}{Theorem}[section]
\newtheorem{lemma}{Lemma}[section]
\newtheorem{proposition}{Proposition}[section]

\newtheorem{assumption}{Assumption}[section]

\newtheorem{definition}{Definition}[section]

\usepackage{bm,comment,makecell,microtype}
\newcommand{\bTheta}{\bm{\Theta}}
\newcommand{\btheta}{\bm{\theta}}
\newcommand{\bx}{\mathbf{x}}
\newcommand{\bX}{\mathbf{X}}

\newcommand{\cX}{\mathcal{X}}
\newcommand{\cO}{\mathcal{O}}

\newcommand{\E}{\mathrm{E}}
\newcommand{\Var}{\mathrm{V}}
\newcommand{\Cov}{\mathrm{Cov}}

\newcommand{\estmuMC}{\bar{\mu}}
\newcommand{\estmuIS}{\widehat{\mu}}
\newcommand{\estmuSN}{\widetilde{\mu}}

\newcommand{\tmustar}{\widetilde{\mu}^\star}

\newcommand{\ESS}{n^e}
\newcommand{\real}{\mathbb{R}}
\DeclareMathOperator{\st}{\mbox{subject to}}
\def\argmin{\mathop{\rm arg\,min}}%

\providecommand{\keywords}[1]
{
	\small
	\textbf{\textit{Keywords---}} #1
}

\title{Efficient Nested Simulation Experiment Design\\
	via the Likelihood Ratio Method}

\usepackage{authblk}

\author[1]{Mingbin Ben Feng}
\author[2]{Eunhye Song}

\affil[1]{Department of Statistics and Actuarial Science, University of Waterloo, Canada}
\affil[2]{The Harold and Inge Marcus Department of Industrial and Manufacturing Engineering, The Pennsylvania State University, USA}

\date{} % Comment this line to show today's date

\usepackage{natbib}

\begin{document}
\maketitle

\begin{abstract}
	In nested simulation literature, a common assumption is that the experimenter can choose the number of outer scenarios to sample. This paper considers the case when the experimenter is given a fixed set of outer scenarios from an external entity.
	We propose a nested simulation experiment design that pools inner replications from one scenario to estimate another scenario's conditional mean via the likelihood ratio method. Given the outer scenarios, we decide how many inner replications to run at each outer scenario as well as how to pool the inner replications by solving a bi-level optimization problem that minimizes the total simulation effort.
	We provide asymptotic analyses on the convergence rates of the performance measure estimators computed from the optimized experiment design.
	Under some assumptions, the optimized design achieves $\cO(\Gamma^{-1})$ mean squared error of the estimators given simulation budget $\Gamma$.
	Numerical experiments demonstrate that our design outperforms a state-of-the-art design that pools replications via regression.
\end{abstract} \hspace{10pt}


\section{Introduction}\label{sec:intro}
We consider a simulation experiment design problem known as \emph{nested simulation}~\citep{hong2017} whose objective is to estimate some functional of a conditional mean of a stochastic simulation output.
Specifically, we are interested in quantities of the form
\begin{equation*}
	\rho\left( \mu(\btheta) \right),\quad \mu(\btheta)=\E_{\btheta}[g(\bX)] \equiv \E\left[g(\bX)|\btheta\right],
\end{equation*}
where $\btheta$ is a simulation input vector referred to as the \emph{outer scenario}, $\bX$
is a random vector whose distribution is parameterized by $\btheta$, and $g(\bX)$ is the stochastic simulation output. Each simulation run that generates $g(\bX)$ given $\btheta$ is referred to as an \emph{inner replication}. Thus, $\mu(\btheta)$ represents the conditional mean of the simulation output given the outer scenario, $\btheta$.
Finally, $\rho$ is a functional that maps the distribution of $\mu(\btheta)$ to a real number.
Two classes of $\rho$ are considered in this study:
\begin{enumerate}[label = (\roman*), itemsep=0pt, topsep=0pt]
	\item $\rho\left( \mu(\btheta) \right) = \E[\zeta(\mu(\btheta))]$ for some real-valued function $\zeta$.
	\item $\rho\left( \mu(\btheta) \right) = \inf\left\{z\in\real: \mathbb{P}(\mu(\btheta) \leq z) \geq \alpha\right\}$ for some user-specified $0<\alpha<1$, i.e., the $\alpha$-quantile or the $\alpha$-level Value-at-Risk (VaR) of $\mu(\btheta)$.
\end{enumerate}
We refer to a generic $\rho(\mu(\btheta))$ as the \emph{nested statistic} in this paper.

Nested simulation is a common tool for quantifying uncertainty of a stochastic simulation model. There is rich literature on simulation model risk analysis~\citep{Barton2022} that adopts the prediction interval of $\mu(\btheta)$ (i.e., quantiles) to quantify the variability in the simulation output caused by uncertain input parameter $\btheta$. The enterprise risk management (ERM) literature~\citep{broadie2015,hong2017,dang_feng_hardy_2022,wang2022smooth} also applies nested simulation to manage risk exposures of portfolios of complex financial instruments.
In this case, $\mu(\btheta)$ is typically a financial portfolio's profit/loss at a future time which depends on underlying risk factors such as equity prices and interest rates represented by $\btheta$.
Popular choices for $\rho$ include exceedance probability, variance, quantile (Value-at-Risk), and Conditional Value-at-Risk (CVaR)~\citep{rockafellar2002conditional}.

Standard nested simulation (SNS) is a classical experiment design that samples $M$ outer scenarios, runs $N$ inner replications at each sampled outer scenario $\btheta$ to estimate $\mu(\btheta)$, and computes a nested statistic from them.
Estimating a nested statistic with small bias and variance requires both $M$ and $N$ to be large.
Thus, designing an efficient nested simulation experiment has attracted much attention from researchers and practitioners.

Our objective is to propose an experiment design from which various nested statistics can be computed with estimation errors comparable to SNS while spending significantly fewer replications. We propose an experiment design that has different numbers of inner replications at different outer scenario.
Furthermore, we focus on the case where the experimenter has no control over outer scenario generation and \emph{must work with a given set of outer scenarios}.
Indeed, such a case arises frequently in practice. For instance,~\cite{GMVCO} describe the nested simulation experiment conducted at General Motors (GM) to quantify uncertainty in their market simulator. In their context, $\btheta$ is a consumer utility parameter vector for vehicle features whose distribution is estimated from a conjoint analysis. Specifically, $\btheta$ has a few thousand dimensions, where each represents a survey respondent's preference for a vehicle feature (e.g., color-black). They impose a Bayesian prior on $\btheta$ and updates it to the posterior by incorporating the conjoint data assuming that the respondents make their choices following a discrete choice model. Since there is no conjugacy to exploit in this model, a Markov chain Monte Carlo algorithm is applied to generate an approximately independent and identically distributed (i.i.d.) sample from the posterior, which discards many realized $\btheta$s to avoid high positive correlations among them. This procedure takes significantly longer than inner replications (market simulation). Thus, it is more natural to consider the computational budget for generating $\btheta$s separately from the cost of inner replications.
Moreover, because GM runs several other market analyses using a common set of $\btheta$s across different business divisions, it is required to work with the same given set of $\btheta$s for uncertainty quantification for consistency.

In practical ERM applications, the underlying risk factors (outer scenarios) may be provided by regulators or by a separate modeling team in a firm, so the experimenter has little control over them.~\cite{risk2018sequential} point out that ``the outer
scenarios are usually treated by practitioners as a fixed object (namely coming from an exogenous economic scenario generator (ESG))...one cannot add (or subtract) further scenarios''~\cite{li2021nested} also mention that ``(in nested simulation) it is often necessary to run a pre-determined set of economic scenarios to project the institution's
cash flows''.
~\cite{ha2022least} mention that ``the same firm-wide scenarios are typically used for evaluating the company's asset and liability portfolios''.
Therefore, it has an important practical impact to study the design where the number of outer scenarios is fixed.

A long-studied experiment design question for nested simulation is how to allocate given simulation budget (the number of inner replications) $\Gamma$ to minimize estimation error.
For SNS, we have $\Gamma=NM$ and the question boils down to choosing $M$ and $N$ so that the nested statistic estimator's mean squared error (MSE) is minimized.~\cite{lee1998} shows that for an exceedance probability estimator, its bias and variance respectively diminish in $\cO(N^{-1})$ and $\cO(M^{-1})$.\footnote{See Section~\ref{sec:problem} for the definition of $\cO(\cdot)$ and other asymptotic notation used throughout the paper.}
This analysis leads to the asymptotically optimal allocation, $N = \cO(\Gamma^{1/3})$ and $M=\cO(\Gamma^{2/3})$.~\cite{gordy2010} establish the same result for a quantile estimator and analyze the bias of CVaR.
~\cite{sun2011} study estimating the variance of $\mu(\btheta)$ using analysis of variance and shows that as $M$ increases, the choice of $N$ minimizing the asymptotic variance of the variance estimator converges to a constant.

Recent studies depart from the $M$-outer-$N$-inner experiment design to enhance efficiency.~\cite{broadie2011} propose a sequential experiment design to estimate an exceedance probability, which selects outer scenarios adaptively and allocates inner replications to the scenarios more critical for reducing the MSE of the estimator;
they achieve asymptotic MSE of $\cO(\Gamma^{-4/5+\epsilon})$ for any positive $\epsilon$.~\citet{giles2019} extend their idea to propose a multi-level Monte Carlo estimator.
Another line of research considers \emph{pooling} inner replications from some or all outer scenarios to estimate $\mu(\btheta)$ for each outer scenario $\btheta$.
These approaches fit a metamodel for $\mu(\btheta)$ by running an initial experiment design, then use the fitted model to estimate a nested statistic with no additional inner replications.~\cite{liustaum2010} adopts a stochastic kriging model to estimate CVaR.~\cite{broadie2015} applies a regression model to estimate Type-(i) nested statistic.
With the optimal choice of $M=\Gamma$ and $N=1$ for the initial experiment, they show that the MSE of the nested statistic estimator converges in $\cO(\Gamma^{-1+\delta})$ for any $\delta>0$ until it reaches an asymptotic bias level.~\cite{hong2017} adopts nonparametric kernel regression and the $k$-nearest-neighbor estimator to pool inner replications from nearby outer scenarios to estimate each $\mu(\btheta)$.
For Type-(i) nested statistic, their optimal parameter choice leads to MSE convergence rate of $\cO(\Gamma^{-\min\{1,4/(d+2)\}})$, where $d$ is the dimension of $\btheta$.~\cite{wang2022smooth} considers kernel ridge regression, which results in the MSE convergence rate between $\cO(\Gamma^{-1})$ and $\cO(\Gamma^{-2/3})$, depending on the smoothness of $\mu(\cdot)$.

A common assumption in aforementioned studies is that given budget $\Gamma$, the experimenter can choose the values of $M$ and $N$.
Hence, it is sensible to seek for estimation-error-minimizing $M$ (and $N$)
as a function of $\Gamma$. However, the same assumption does not apply to when the $M$ outer scenarios are fixed.
Another common feature of the reviewed studies is that an equal number of inner replications are spent at all outer scenarios;~\cite{broadie2011} and~\citet{giles2019} are exceptions, however, they study sequential experiment designs whereas we propose a single-stage experiment design. By single stage, we mean that no initial experiment is required to fit a metamodel as in other work.

We identify two key experiment design decisions for our pooling scheme: (1) \emph{At which outer scenarios to run inner simulation and how many replications to run?} (2) \emph{How to pool these replications} when estimating the conditional mean at each outer scenario?
The former is referred to as the \emph{sampling decision} and the latter as the \emph{pooling decision}.

We propose a new single-stage experiment design that makes both optimal sampling and pooling decisions via the likelihood ratio (LR) method by solving a bi-level optimization problem.
We estimate each $\mu(\btheta)$ by taking a weighted average of the LR estimators defined between $\btheta$ and all $M$ outer scenarios. However, some pairwise LR estimators may have infinite variances.
We adopt an LR estimator variant that allows us to (approximately) assess its variance and hence evaluate the efficiency of pooling each pair of the outer scenarios \emph{prior to running any inner replications}.

Utilizing this feature, our bi-level optimization problem has a constraint that requires the pooled estimator of $\mu(\btheta)$ for each $\btheta$ to have variance (approximately) no larger than that of a standard Monte Carlo (MC) estimator with $N$ i.i.d. samples.
The upper-level problem makes the sampling decision that minimizes the total simulation budget.
The lower-level problem makes the pooling decision by finding the variance-minimizing weights for each $\btheta$ given any sampling decision.
Unlike other pooled nested simulation designs, ours may result in running no replication at some outer scenarios.
We reformulate the bi-level problem into a linear program (LP) with guaranteed feasibility.
The total simulation budget for our design is the optimal objective value of the LP. %, which tends to be \emph{much smaller} than $MN$.

While our method works for any $N$ that the experimenter provides, a question arises on how to choose $N$. When $M$ is fixed, even if $N$ is infinity, the nested statistic estimator has nonzero MSE.
Hence, we argue that it is sensible to increase $N$ just enough so that the MSE convergence rate is not slowed down by $N$.
The functional relationship between $N$ and $M$ is difficult to derive for finite $M$ in general.
However, by analyzing the convergence rate of the estimator as $M,N\to\infty$, we obtain a guidance on how to choose $N$ relative to $M$.
We show that the MSE convergence rate of the Type-(i) nested statistic estimator computed from our experiment design is $\mathcal{O}(N^{-1})+\mathcal{O}(M^{-1})$. Similarly, the Type-(ii) statistic (quantile) converges to the true quantile in $\cO_p(M^{-1/2})+\cO_p(N^{-1/2})$. Both results suggest $N=\Theta(M)$.

Caution must be taken when interpreting our asymptotic analysis results to avoid the misconception that it is less efficient than SNS, which has an MSE convergence rate of $\mathcal{O}(N^{-2})+\mathcal{O}(M^{-1})$.
Recall that in our design, $N$ is merely a parameter in the bi-level program.
This parameter ensures that the pooled LR estimator's approximate variance is no larger than the variance of the standard MC estimator with $N$ i.i.d. replications.
Our design does not actually run $N$ replications at the $M$ scenarios, and its simulation budget tends to be much less than $MN$.
Indeed, under some mild conditions, we show that the optimized simulation budget grows in $\cO(M)$ when $N=M$ is adopted. This result combined with the MSE convergence rate analyses suggests that the MSE of a type-(i) statistic estimated from our design converges in $\cO(\Gamma^{-1})$ given a simulation budget $\Gamma$.
%Contrasting to that, SNS runs $N$ inner replications at all $M$ outer scenarios, so its simulation budget is always $\Gamma = MN$. A fair comparison between these two designs is to examine their simulation budgets when $N$ is chosen to optimize the MSE convergence rate given $M$. Note that the asymptotic analysis of SNS favors $N=\sqrt{M}$ resulting in $\Gamma = M^{3/2}$. Although it is difficult to characterize the optimized $\Gamma$ as a function of $M$ and $N$ for our design in general, we show that for a special case, $\Gamma=\cO(M)$ when $N=M$ in Section~\ref{sec:opt.obj}. Our empirical analyses also demonstrate that optimized $\Gamma$ in our design is orders of magnitude smaller than $M^{3/2}$.

Lastly, to place our contribution in the literature, we provide a brief review on two closely related methods: importance sampling (IS) and the LR method.
These are mathematically identical but differ in their goals.
IS concerns finding a variance-minimizing sampling distribution for a single estimator \emph{prior to} running simulations.
While we also consider sampling decisions, our study differ from classical IS~\citep{hesterberg1988advances,owenbook,glassermanbook} because we consider multiple estimators at the same time.
The LR method aims to save computation by reusing outputs \emph{after} the simulations are run.
This method has been applied to improve efficiency of metamodeling~\citep{dong2018unbiased}, gradient estimation~\citep{l1990unified,l1993two,glasserman2014robust}, and sensitivity analysis and optimization~\citep{rubinstein1993discrete,kleijnen1996optimization,fu2016handbook,maggiar2018derivative}.
The LR method is also applied in the ``green simulation'' experiment design~\citep{feng2015green,feng2017green}, which recycles and reuses simulation outputs previously made to save computations and improve precision.
Similar to these approaches, our pooling decision is also guided by the LR method once the sampling decision is made.

In recent development of simulation model risk quantification literature,~\cite{zhouliu2019} apply the LR method to pool inner replications of the nested simulation with equal weights and demonstrate significant computational savings.~\cite{zhang2022sample} also propose an LR-based approach that pools inner replications simulated from one common sampling distribution.
Neither of these studies optimizes the sampling decision. Moreover, both designs have a risk of making some outer scenarios' conditional mean estimators have large or even unbounded variances.
Our optimal experiment design avoids such a case by assigning a small (or zero) pooling weight to an outer scenario pair, if their LR estimator is deemed to have a large variance.

Our preliminary work~\citep{feng2019efficient} explores application of various kinds of LR estimators to make the simulation risk analysis more efficient. However, we have not considered optimizing the sampling and pooling decisions in the earlier work. Moreover, the theoretical convergence analyses on different types of nested statistics are new and we focus exclusively on the self-normalized LR estimator in the current paper.

The remainder of the paper is organized as follows. Section~\ref{sec:problem} provides a mathematical framework for the nested simulation problem. Section~\ref{sec:estimator} defines our pooled LR estimator and discusses its properties. In Section~\ref{sec:design}, we propose the experiment design framework to optimize sampling and pooling decisions. Asymptotic properties of the nested simulation statistics computed from the experiment design are analyzed in Section~\ref{sec:asymptotics} followed by the analysis on the simulation budget growth rate in Section~\ref{sec:opt.obj}. Empirical evaluations and conclusions are provided in Sections~\ref{sec:expr} and~\ref{sec:conclusion}, respectively.

\section{Problem Statement}\label{sec:problem}
Consider a nested simulation experiment, where the outer scenarios are $M$ i.i.d.\ multidimensional random vectors, $\btheta_1,\ldots,\btheta_M$, whose support is $\bTheta \subseteq \real^p$.
With slight abuse of notation, we treat $\btheta_1,\ldots,\btheta_M$ as realized outer scenarios rather than random vectors until Section~\ref{sec:asymptotics}.
For any $\btheta\in\bTheta$, let $h(\bx;\btheta)$, $\mu(\btheta)$, and $\Var_{\btheta}[g(\bX)]$ represent the conditional density function, mean, and variance of the input vector $\bX$ given $\btheta$, respectively.
When no ambiguity arises, we also use the shorthand notation $\mu_i = \mu(\btheta_i)$ for convenience.
The following assumption facilitates our discussion.
\begin{assumption}\label{assm:MCbasic}
	For all $\btheta\in\bTheta$, $h(\bx;\btheta)$ is a well-defined probability distribution function and has a common support $\cX\subseteq \real^d$ for a fixed dimension $d$.
	Furthermore, the simulation output function $g$ does not depend on $\btheta$ and $\sup_{\btheta\in\bTheta}\Var_{\btheta}[g(\bX)]<\infty$.
\end{assumption}
The common support, $\cX$, ensures that the LR between any two outer scenarios is well-defined.
The common $g$ allows us to reuse $\bX$ {generated} from $\btheta_j\neq \btheta_i$ to estimate~$\mu_i$ via the LR method.
The fixed input dimension, $d$, is a limitation of our method.
We refer the readers to~\cite{feng2019efficient} on applying the LR method to improve nested simulation efficiency when the dimension of $\bX$ is random.
Lastly, $\sup_{\btheta\in\bTheta}\Var_{\btheta}[g(\bX)]<\infty$ ensures that the standard Monte Carlo (MC) estimator for $\mu(\btheta)$ for all $\btheta\in\bTheta$ has finite variance.

Recall that we consider two types of nested statistics.
For Type-(i) statistic, we consider three choices for $\zeta$: (1) The indicator function, $\zeta(\mu_i) = I(\mu_i \leq \xi)$ for some fixed $\xi\in\real$; the indicator function can be used to estimate the exceedance probability that a portfolio's loss beyond $\xi$, that is, $\Pr(\mu_i> \xi)$.
(2) The hockey stick function, $\zeta(\mu_i) = \max\{\mu_i-\xi,0\} =(\mu_i-\xi)I(\mu_i>\xi)$, which is commonly used in ERM applications to price derivatives or compound options~\citep{glassermanbook}.
(3) Smooth function $\zeta$ with a bounded second derivative. An example of such $\zeta$ is the squared loss function given target $\xi$, $\zeta(\mu_i) = (\mu_i-\xi)^2$.

In SNS with $N$ inner replications for each outer scenario, the MC estimator of $\mu_i$,
$\estmuMC_i \equiv \sum_{k=1}^{N} g(\bX_{k})/N$, is computed for each $i$, where $\bX_{k} \stackrel{i.i.d.}{\sim} h(\bx;\btheta_i)$ is the input vector generated within the $k$th inner replication made at $\btheta_i$.
Under Assumption~\ref{assm:MCbasic}, $\Var_{\btheta_i}[\estmuMC_i]=\Var_{\btheta_i}[g(\bX)]/N$ for any $\btheta_i\in\bTheta$. From $\estmuMC_1,\ldots,\estmuMC_M$, $\E[\zeta(\mu_i)]$ can be estimated by $\sum_{i=1}^M \zeta(\estmuMC_i)/M$.
Also, the $\alpha$-quantile of $\mu_i$ can be estimated by the empirical quantile $\estmuMC_{(\lceil M\alpha \rceil)}$, where $\estmuMC_{(i)}$ is the $i$th order statistic of $\estmuMC_1,\ldots,\estmuMC_M$.

Throughout the paper, for positive sequences $\{a_n\}$ and $\{b_n\}$, $a_n = \cO(b_n)$ means that there exists constant $\bar{c}>0$ such that $a_n \leq \bar{c}b_n$ for all $n\in\mathbb{Z}^+$ and $a_n = o(b_n)$ implies $a_n/b_n\to0$ as $n\to\infty$. Moreover, $a_n=\Theta(b_n)$ means that there exist constants $\bar{c}, \underline{c}>0$ such that $\underline{c}b_n \leq a_n\leq \bar{c}b_n$ for all $n\in\mathbb{Z}^+$. For sequence of random variables $\{X_n\}$, we say $X_n = \cO_p(b_n)$, if for any $\varepsilon>0$ there exist $\underline{c}>0$ and $\bar{n}>0$ such that $\Pr\{|X_n/b_n|>\underline{c}\}<\varepsilon$ for all $n> \bar{n}$.

\section{Conditional Mean Estimator via Likelihood Ratio Method}\label{sec:estimator}
In this section, we continue to treat $\btheta_1,\ldots,\btheta_M$ as realized outer scenarios.
We apply the LR method to pool inner replications from several scenarios to estimate the conditional mean at each scenario.
Consider two scenarios $\btheta_i$ and $\btheta_j$, where $\btheta_i$ is the target scenario whose conditional mean $\mu_i$ is desired and $\btheta_j$ is the sampling scenario at which inner replications are generated.
Then, from the change of measure, we have
\begin{equation}\label{eq:muthetaLR}
	\mu_i = \E_{\btheta_i}\left[g(\bX)\right] = \E_{\btheta_j}\left[g(\bX)\frac{h(\bX;\btheta_i)}{h(\bX;\btheta_j)}\right] = \E_{\btheta_j}\left[g(\bX) W_{ij}(\bX)\right],
\end{equation}
where $W_{ij}(\bX) \equiv \frac{h(\bX;\btheta_i)}{h(\bX;\btheta_j)}$ is the LR of the input vector, $\bX$, between $\btheta_i$ and $\btheta_j$.
Under Assumption~\ref{assm:MCbasic}, $W_{ij}(\bX)$ is well-defined and $\E_{\btheta_j}[W_{ij}(\bX)] = 1$ for any $\btheta_i,\btheta_j\in\bTheta$. From~\eqref{eq:muthetaLR}, one can define an unbiased LR estimator of $\mu_i$, $\widehat{\mu}_{ij} \equiv \sum_{k=1}^{N_j} g(\bX_k)W_{ij}(\bX_k)/N_j$, where $N_j$ is the number of inner replications made at $\btheta_j$. We refer to $\widehat{\mu}_{ij}$ as the nominal LR estimator.

A diagnostic measure for assessing efficiency of an LR estimator is the \emph{effective sample size} (ESS), which is defined as the number of replications required for the MC estimator to achieve the same variance as the LR estimator.
So, the larger the ESS, the more precise the LR estimator is.
Because $\Var[\estmuIS_{ij}] = \Var_{\btheta_j}[g(\bX)W_{ij}(\bX)]/N_j$, the ESS of the nominal LR estimator is
$\frac{\Var_{\btheta_i}[g(\bX)]}{\Var_{\btheta_j}[g(\bX)W_{ij}(\bX)]} N_j$.
Due to its dependence on $g(\bX)$, this ESS cannot be evaluated analytically in general and must be estimated via simulation.

In this study, we adopt the self-normalized LR estimator:
\begin{equation}\label{eq:muhatthetaSN}
	\estmuSN_{ij}= {N_j}^{-1}\sum\nolimits_{k=1}^{N_j} g(\bX_k)\widetilde{W}_{ij,k},\quad\bX_k \stackrel{i.i.d.}{\sim} h(\bx;\btheta_j),\,\forall k=1,\ldots,N_j,
\end{equation}
where $\widetilde{W}_{ij,k} = \frac{W_{ij}(\bX_k)}{\sum\nolimits_{\ell=1}^{N_j}W_{ij}(\bX_\ell)/N_j}$ is the self-normalized LR from the $k$th replication so that the sample average of $\widetilde{W}_{ij,1},\widetilde{W}_{ij,2},\ldots,\widetilde{W}_{ij,N_j}$ becomes $1$.
Lemma~\ref{lem:consistency} states asymptotic properties of $\estmuSN_{ij}$. Note that $\bX$ is dropped from $W_{ij}(\bX)$ henceforth for notational convenience.

\begin{lemma}\label{lem:consistency}
	Consider any given target and sampling scenarios $\btheta_i,\btheta_j\in\bTheta$. Under Assumption~\ref{assm:MCbasic}, (i) $\estmuSN_{ij}$ converges almost surely to $\mu_i$ as $N_j$ increases; and (ii)
	with additional regularity conditions in Assumption~\ref{assm:moment.conditions.for.LR} in~\ref{app:technicalassumptions} of the Online Supplement,
	\begin{equation}\label{eq:asymp.bias.var}
		\E_{\btheta_j}[\estmuSN_{ij}] - \mu_i =
		\tfrac{-\E_{\btheta_j}[W_{ij}^2(g(\bX)-\mu_i)]}{N_j} + o(N_j^{-1}); 	\Var_{\btheta_j}[\estmuSN_{ij}] = \tfrac{\E_{\btheta_j}[W_{ij}^2(g(\bX)-\mu_i)^2]}{N_j} + o(N_j^{-1}).
	\end{equation}
\end{lemma}

\noindent Part (i) of Lemma~\ref{lem:consistency} is proved in Theorem 9.2 of~\cite{owenbook}. Results similar to Part~(ii) can be found in~\cite{hesterberg1988advances} and~\cite{owenbook}.

As shown in Lemma~\ref{lem:consistency}, the self-normalized LR estimator is biased but consistent.
We adopt the self-normalized LR estimator in our study because it has a convenient ESS approximation.~\cite{kong1992} shows that the variance of $\estmuSN_{ij}$ can be approximated as
\begin{equation}
	\label{eq:ESS}
	\Var_{\btheta_j}[\estmuSN_{ij}] \approx \Var_{\btheta_i}[g(\bX)]\frac{ \E_{\btheta_j}[W_{ij}^2]}{N_j}.
\end{equation}
Then, the approximate ESS of $\estmuSN_{ij}$ is $\ESS_{ij} \equiv N_j/\E_{\btheta_j}[W_{ij}^2]$, which is free of $g(\bX)$ and can be computed analytically in some cases.
In Lemma~\ref{lem:exp.EW2} (Section~\ref{sec:opt.obj}), we derive a closed-form expression for $\E_{\btheta_j}[W_{ij}^2]$ when $h(\bx;\btheta)$ belongs to the exponential family.

In a more detailed derivation,~\cite{liu1996} shows the approximation error of~\eqref{eq:ESS} is
\begin{equation}\label{eq:ESS.liu}
	\Var_{\btheta_j}[\estmuSN_{ij}] = \left(\Var_{\btheta_i}[g(\bX)]{ \E_{\btheta_j}[W_{ij}^2]} + \E_{\btheta_i}[(W_{ij}-\E_{\btheta_i}[W_{ij}])(g(\bX)-\mu_i)^2]\right){N_j^{-1}} + o(N_j^{-1}),
\end{equation}
which holds under Assumption~\ref{assm:moment.conditions.for.LR} and comments that $\E_{\btheta_i}[(W_{ij}-\E_{\btheta_i}[W_{ij}])(g(\bX)-\mu_i)^2]N_j^{-1}$ omitted from~\eqref{eq:ESS} may be small when $g(\bX)$ is relatively flat, however, when it is large, the approximation error of~\eqref{eq:ESS} can be significant.
Nevertheless,~\cite{liu1996} recommends using~\eqref{eq:ESS} as a rule of thumb.
We adopt this recommendation and evaluate efficiency of $\widetilde{\mu}_{ij}$ using $\ESS_{ij}$ for each $(\btheta_i,\btheta_j)$ pair prior to running any inner simulations.
In Section~\ref{subsec:ERM}, we empirically demonstrate that Approximation~\eqref{eq:ESS} performs well with an ERM example.

\noindent \textbf{Remark 1}. Some studies~\citep{martino2017effective,elvira2018rethinking} define the ESS as the number of replications such that the LR estimator's MSE matches $\Var[\estmuMC_i]$. As seen in Lemma~\ref{lem:consistency}, the squared bias diminishes faster than the variance asymptotically, thus we consider matching the variances.

\section{Optimal Nested Simulation Experiment Design}\label{sec:design}
Suppose $M$ realized outer scenarios, $\btheta_1,\ldots,\btheta_M$, are given.
We estimate the conditional means at all $M$ scenarios by
pooling the self-normalized LR estimators.
We require that, for each $\btheta_i$, the pooled estimator's variance is no larger than $\Var[\bar{\mu}_i]/N$, which is the variance of the SNS estimator with $N$ inner replications in each scenario; we refer to this requirement as the \emph{precision requirement}.
Our goal is to minimize the total number of inner replications run at the $M$ scenarios, that is, the simulation budget.
In this section, we assume that the experimenter chooses an appropriate $N$.

Suppose $N_j$ i.i.d.\ inner replications are run at each $\btheta_j$. Some outer scenarios may have zero replications, i.e., $N_j = 0$.
Then, $\widetilde{\mu}_{ij}$ is well-defined for any $(\btheta_i,\btheta_j), 1\leq i,j\leq M$, if $N_j>0$.
For each $i$, consider the following pooled LR estimator of $\mu_i$:
\begin{equation}\label{eq:muhatthetaSNpooled}
	\estmuSN_i \equiv \sum\nolimits_{{j=1, N_j>0}}^M \gamma_{ij} \estmuSN_{ij}, \quad \sum\nolimits_{{j=1, N_j>0}}^M\gamma_{ij} = 1,
\end{equation}
where $\gamma_{ij}, j=1,\ldots,M$, are the pooling weights.
In SNS, $\estmuMC_i$ only uses the inner replications simulated at $\btheta_i$.
In contrast, the pooled LR estimator $\estmuSN_i$ pools all inner replications from all sampling scenarios with weights $\{\gamma_{ij}:j=1,\ldots,M\}$.
As all inner replications are run independently, the variance of the pooled LR estimator is given by
$\Var[\estmuSN_i] = \sum_{j=1,N_j>0}^M \gamma_{ij}^2 \Var_{\btheta_j}[\estmuSN_{ij}]$.

Recall that our precision requirement for each $\btheta_i$ is $\Var[\estmuSN_i]\leq \Var[\estmuMC_i]$.
From~\eqref{eq:ESS}, we have
\begin{equation}\label{eq:var.constraint}
	\sum\nolimits_{{j=1, N_j>0}}^M \gamma_{ij}^2 \Var_{\btheta_j}[\estmuSN_{ij}] =\Var[\estmuSN_i]\leq \Var[\estmuMC_i]= \frac{\Var_{\btheta_i}[g(\bX)]}{N} \stackrel{\eqref{eq:ESS}}{\Rightarrow}
	\sum\nolimits_{{j=1, N_j>0}}^M \frac{\gamma_{ij}^2 \E_{\btheta_j}[W_{ij}^2]}{N_j} \leq \frac{1}{N}.
\end{equation}
For given sampling decisions $\{N_j\}$, there may be infinitely many feasible weights $\{\gamma_{ij}\}$ that satisfy the constraints.
Among them, it is sensible to choose $\{\gamma_{ij}\}$ that minimizes $\Var[\estmuSN_i]$.
From this insight, we formulate the following bi-level optimization problem:
\begin{align}\label{prob:bi-level}
	& \min_{N_j \geq 0,\gamma_{ij} } &  & \sum\nolimits_{j=1}^M N_j                                                                                                                                                                                                                       \\
	& \st                            &  & \sum\nolimits_{{j=1, N_j>0}}^M \frac{\gamma_{ij}^2 \E_{\btheta_j}[W_{ij}^2]}{N_j} \leq \frac{1}{N}, \qquad\qquad\qquad \forall i=1,\ldots,M \nonumber                                                                                           \\
	&                                &  & \begin{aligned}\label{prob:bi-levelsub}
		\{\gamma_{ij}\} \in\argmin_{\gamma_{ij}} \left\{\sum_{{j=1, N_j>0}}^M \frac{\gamma_{ij}^2 \E_{\btheta_j}[W_{ij}^2]}{N_j}: \sum_{{j=1, N_j>0}}^M \gamma_{ij} = 1\right\},
	\end{aligned} \hspace{10pt} \forall i=1,\ldots,M
\end{align}
A brief overview of bi-level optimization is provided in~\ref{app:bi-levelopt}.
The upper-level problem~\eqref{prob:bi-level} makes the \emph{sampling decision} by finding the optimal $\{N_j\}$.
The sampling decision determines not only at which $\btheta_j$s to run replications at, but also how many replications to run at each.
The lower-level problem~\eqref{prob:bi-levelsub} defined for each target scenario $i$ makes the \emph{pooling decision} to find $\{\gamma_{ij}\}$ that minimizes the (approximate) variance of $\estmuSN_i$ given the sampling decision $\{N_j\}$.
We ignore integrality constraints for $\{N_j\}$ here.
In numerical experiments, we assign $\lceil N_j \rceil$ inner replications to $\btheta_j$.

Given the sampling decision $\{N_j\}$, for each $i$ the lower-level problem~\eqref{prob:bi-levelsub} is a simple quadratic program that can be solved analytically via the Karush-Kuhn-Tucker (KKT) conditions:
The Lagrangian function of the $i$th lower-level problem is $\mathcal{L}(\gamma_{ij},\lambda)=\sum_{j=1, N_j>0}^M \frac{\gamma_{ij}^2 \E_{\btheta_j}[W_{ij}^2]}{N_j}+\lambda(1-\sum_{j=1, N_j>0}^M \gamma_{ij})$.
The corresponding KKT condition is $\frac{2\gamma_{ij} \E_{\btheta_j}[W_{ij}^2]}{N_j}-\lambda=0$ for all $j=1,\ldots,M$ such that $N_j>0$, which implies $\gamma_{ij}= \frac{\lambda N_j}{2\E_{\btheta_j}[W_{ij}^2]}$.
Hence, the optimal pooling decision, $\gamma_{ij}^\star$, is proportional to $\frac{N_j}{\E_{\btheta_j}[W_{ij}^2]}$.
Considering $\sum_{j=1, N_j>0}^M \gamma_{ij}=1$, we have
$\gamma^\star_{ij} = \frac{N_j/\E_{\btheta_j}[W_{ij}^2]}{\sum_{k=1}^M N_k/\E_{\btheta_k}[W_{ik}^2]}$, for all $j=1,\ldots,M$ such that $N_j>0$.
Note that we drop the condition, $N_j>0$ since the same expression produces $\gamma_{ij}^\star = 0$ when $N_j = 0$, as desired.
Notice that $\gamma_{ij}^\star = 0$ when $\E_{\btheta_j}[W_{ij}^2]=\infty$ even if $N_j>0$.
This is the case when $\btheta_i$ does not pool from the inner replications at $\btheta_j$ because $\Var[\estmuSN_{ij}]$ is large.
Plugging in $\gamma^\star_{ij}$ into~\eqref{prob:bi-levelsub}, the $i$th lower-level problem's optimal objective value is
\begin{equation}\label{eq:lower.obj}
	\sum\nolimits_{j=1}^M \frac{(\gamma^\star_{ij})^2 \E_{\btheta_j}[W_{ij}^2]}{N_j} = \left(\sum\nolimits_{j=1}^M \frac{N_j}{\E_{\btheta_j}[W_{ij}^2]}\right)^{-1}.
\end{equation}
Plugging~\eqref{eq:lower.obj} into the first constraint of the bi-level program, we have:
\begin{align}\label{prob:LP}
	\min_{N_j\geq 0} \hspace{6pt}\sum\nolimits_{j=1}^M N_j \;\;\;\;
	\mbox{subject to} \hspace{6pt}\sum\nolimits_{j=1}^M \frac{N_j}{ \E_{\btheta_j}[W_{ij}^2]} \geq N , \qquad \forall i=1,\ldots,M.
\end{align}
Note that this linear program (LP) is always feasible because $\E_{\btheta_j}[W_{jj}^2] = 1$, so $\{N_j=N, \forall j=1,\ldots,M\}$ satisfies all constraints.
Also note that the constraints can be written as $\sum\nolimits_{j=1}^M \ESS_{ij} \geq N$ for all $ i=1,\ldots,M$.
Namely, these constraints require the total ESS of the pooled estimator to be at least $N$ for each target scenario.
Provided that $\E_{\btheta_j}[W_{ij}^2]$ can be computed a priori, e.g., exponential family,~\eqref{prob:LP} can be solved prior to running any inner replications.
Numerical studies in Section~\ref{sec:expr} show that the optimal objective value of~\eqref{prob:LP} is orders of magnitude smaller than $MN$, which demonstrates that our experiment design significantly reduces the simulation budget compared to SNS.

Let $\{c_j^\star\}$ be an optimal solution of~\eqref{prob:LP} when $N=1$.
Proposition~\ref{prop:linearity.of.Nj} shows that, for fixed outer scenarios, one can solve~\eqref{prob:LP} for $N=1$ and obtain the optimal solution for arbitrary $N$ by simply scaling the solution from the former by $N$.
Therefore, even if the target precision, $N$, is changed \emph{post hoc}, there is no need to resolve~\eqref{prob:LP}.
Such proportionality is also useful for showing asymptotic properties of the pooled estimator in Section~\ref{sec:asymptotics}.

\begin{proposition}\label{prop:linearity.of.Nj}
	Given $\btheta_1,\ldots,\btheta_M$, let $\{c_j\}$ and $\{c_j^\star\}$ be a feasible solution and an optimal solution of~\eqref{prob:LP}, respectively, for $N=1$. Then, $\{N_j=N\cdot c_j\}$ and $\{N_j^\star=N \cdot c_j^\star\}$ are a feasible solution and an optimal solution of~\eqref{prob:LP}, respectively, for arbitrary $N>0$.
\end{proposition}
\begin{proof}
	As $\{c_j\}$ is feasible {for}~\eqref{prob:LP} when $N=1$, multiplying both sides of the inequality constraint of~\eqref{prob:LP} by arbitrary $N$
	shows that $\{N_j= N\cdot c_j\}$ is a feasible solution of the revised problem.
	By means of contradiction, suppose that there exists a feasible solution $\{N_j'\}$ of the revised problem such that $\sum_{j=1}^{M}N_j'<\sum_{j=1}^{M}N_j^\star$.
	Dividing both sides of the constraints by $N$, it is clear that $\{c_j'=N_j'/N\}$ is a feasible solution of~\eqref{prob:LP} when $N=1$.
	However, by construction $\sum_{j=1}^{M}c_j'=(\sum_{j=1}^{M}N_j')/N<(\sum_{j=1}^{M}N_j^\star)/N=\sum_{j=1}^{M}c_j^\star$, which contradicts the premise that $\{c_j^\star\}$ is an optimal solution of~\eqref{prob:LP} when $N=1$.
\end{proof}

In light of Proposition~\ref{prop:linearity.of.Nj}, the optimal pooling weights corresponding to $\{N_j^\star\}$ satisfy
\begin{equation}\label{eq:lower.KKT.N0=1}
	\gamma^\star_{ij} = \frac{c_j^\star/\E_{\btheta_j}[W_{ij}^2]}{\sum_{k=1}^M c_k^\star/\E_{\btheta_k}[W_{ik}^2]},\hspace{14pt} \forall j=1,\ldots,M,
\end{equation}
which implies $\{\gamma_{ij}^\star\}$ does not depend on $N$.

We note that both the optimal pooling decision~\eqref{eq:lower.KKT.N0=1} and the objective value~\eqref{eq:lower.obj} have meaningful interpretations:
Recall that the ESS of $\estmuSN_i$ is approximated by $\ESS_{ij}=N_j/E_{\btheta_j}[W_{ij}^2]$.
For any outer scenario $\btheta_i$,~\eqref{eq:lower.KKT.N0=1} can be written as $\gamma^\star_{ij}=\ESS_{ij}/(\sum_{k=1}^M \ESS_{ik})$.
So, given any sampling decision $\{N_j\}$, the optimal way to pool the estimators $\estmuSN_{ij}$, $j=1,\ldots,M$, is to weight them proportionally to their ESS.
Moreover, plugging $\gamma^\star_{ij}=\ESS_{ij}/(\sum_{k=1}^M \ESS_{ik})$ into the first equality of~\eqref{eq:var.constraint}, we have $\Var[\estmuSN_i]\approx\Var_{\btheta_i}[g(\bX)]/(\sum_{j=1}^M \ESS_{ij})$.
This suggests that the ESS of the optimally pooled estimator $\estmuSN_i$ is equal to the total ESS of the estimators $\estmuSN_{ij}, j=1,\ldots,M$.

Once $\{N_j^\star\}$ and $\{\gamma_{ij}^\star\}$ are found, we run inner replications at $\{\btheta_j\}$ as prescribed by $\{N_j^\star\}$. For each $(i,j)$ pair such that $\gamma_{ij}^\star >0$, we compute $\estmuSN_{ij}$ defined in~\eqref{eq:muhatthetaSN}. The optimally pooled conditional mean estimators are computed as
$\estmuSN_i^\star = \sum\nolimits_{{j=1}}^M \gamma_{ij}^\star \estmuSN_{ij}$, for $i=1,\ldots,M$.
Then, the Type-(i) nested statistic can be estimated by $\widetilde{\zeta} = \sum_{i=1}^M\zeta(\estmuSN_i^\star)/M$.
The $\alpha$-quantile of $\mu_i$ is estimated by the empirical quantile, $\estmuSN_{(\lceil M\alpha\rceil)}^\star$, computed from $\estmuSN_1,\ldots,\estmuSN_M$.

\section{Asymptotic Analysis}\label{sec:asymptotics}
In this section, we treat $\btheta_1,\ldots,\btheta_M$ as random vectors, not realizations, and establish asymptotic properties of the estimated nested statistics, $\widetilde{\zeta}$ and $\estmuSN_{(\lceil M\alpha\rceil)}^\star$, computed from the optimized design. These analyses provide a guidance for selecting $N$ as a function of $M$.
Prior to presenting the main results, we first establish convergence results for $\widetilde{\mu}_i^*$.

Since feasibility of~\eqref{prob:LP} is always guaranteed, $\estmuSN_i^\star$ is well-defined for any $N$ and $M$.
The following theorem states a strong consistency result for $\estmuSN_i^\star$.

\begin{theorem}\label{thm:strong.consistency}
	Suppose Assumption~\ref{assm:MCbasic} holds. Given $\btheta_1,\ldots,\btheta_M$,
	$\estmuSN_i^\star \xrightarrow{a.s.} \mu_i \mbox{ as } N\to\infty$ for $i=1,\ldots,M$.
\end{theorem}
\begin{proof}
	Recall that from Proposition~\ref{prop:linearity.of.Nj}, $\{\gamma_{ij}^\star\}$ do not depend on $N$ and are constants once $\btheta_1,\ldots,\btheta_M$ are given.
	From the definition of $\estmuSN_i^\star$,
	\begin{equation*}
		\lim\limits_{N\to\infty} \estmuSN_{i}^\star = \lim\limits_{N\to\infty} \sum\nolimits_{j=1}^M \gamma_{ij}^\star \estmuSN_{ij} = \sum\nolimits_{j=1}^M \gamma_{ij}^\star \lim\limits_{N_j\to\infty} \estmuSN_{ij} \xrightarrow{a.s.} \sum\nolimits_{\substack{j=1}}^M \gamma_{ij}^\star \mu_i = \mu_i\sum\nolimits_{j=1}^M \gamma_{ij}^\star = \mu_i,
	\end{equation*}
	where the second equality holds because $N_j^\star \propto N$ from Proposition~\ref{prop:linearity.of.Nj} and $\gamma_{ij}^\star=0$ whenever $N_j^\star=0$.
	The almost sure convergence holds from Lemma~\ref{lem:consistency}.
\end{proof}

Next, we examine the MSE convergence rate of $\tmustar_i$.
Recall that Lemma~\ref{lem:consistency} shows both bias and variance of $\estmuSN_{ij}$ are $\cO(N_j^{-1})$ under Assumption~\ref{assm:moment.conditions.for.LR}. Because $N_j^\star\propto N$, to show $\mathrm{MSE}[\tmustar_i]=\cO(N^{-1})$, it suffices to have the result of Lemma~\ref{lem:consistency} hold for each $(i,j)$ pair such that $\gamma_{ij}^\star>0$. Assumption~\ref{assm:boundedmoments} below formally states the sufficient condition.

\begin{assumption}\label{assm:boundedmoments}
	For each $\btheta_i\in\bTheta$, let $\bTheta_i \equiv \{\btheta_j\in\bTheta|\E_{\btheta_j}[W_{ij}^2]<\infty\}$. Then, 
	\begin{enumerate}[label=(\roman*)]
		\item there exists $C>0$ such that for any $N_j >C, \btheta_i\in\bTheta$ and $\btheta_j \in \bTheta_i$, $N_j |\E_{\btheta_j}[\estmuSN_{ij}]-\mu_i| < |\E_{\btheta_j}[W_{ij}^2(g(\bX)-\mu_i)]| + 1$ and
		$N_j \Var_{\btheta_j}[\estmuSN_{ij}]< \E_{\btheta_j}[W_{ij}^2(g(\bX)-\mu_i)^2] +1$.
		\item `
		$\sup_{\btheta_i\in\bTheta}\sup_{\btheta_j\in\bTheta_i}\E_{\btheta_j}[W_{ij}^2(g(\bX)-\mu_i)]<\infty$ and $\sup_{\btheta_i\in\bTheta}\sup_{\btheta_j\in\bTheta_i}\left|\E_{\btheta_j}[W_{ij}^2(g(\bX)-\mu_i)^2]\right|<\infty$.
	\end{enumerate}
\end{assumption}

Assumption~\ref{assm:boundedmoments}.(i) ensures the conclusions of Lemma~\ref{lem:consistency} hold for any $\btheta_i,\btheta_j\in\bTheta$, while (ii) assures the moments in the constants to be uniformly bounded in $\bTheta$.

Additionally, we make a minor modification to the definition of $\{N_j^\star\}$:
\begin{align}\label{eq:modified.Njstar}
	N_j^\star =
	\begin{cases}
		\delta N,    & \mbox{if } 0< c_j^\star < \delta, \\
		c_j^\star N, & \mbox{otherwise},
	\end{cases}
\end{align}
where $\delta$ is a small positive constant. In words,~\eqref{eq:modified.Njstar} guarantees that if any replications are made at $\btheta_j$, then $N_j^\star$ is to be at least $\delta$ fraction of $N$.
The number of outer scenarios we sample at does not increase because $N_j^\star = 0$ if $c_j^\star = 0$. We emphasize that~\eqref{eq:modified.Njstar} has no practical impact as $\delta$ can be chosen to be arbitrarily small in the experiments.
The following theorem establishes that for any sample of $M$ outer scenarios, $\mathrm{MSE}[\tmustar_i] = \cO(N^{-1})$.

\begin{theorem}\label{thm:convergencerate}
	Suppose Assumptions~\ref{assm:MCbasic} and~\ref{assm:boundedmoments} hold. Then, for any finite $M$
	\begin{align*}
		& \sup\nolimits_{\{\btheta_1,\ldots,\btheta_M\}\subset\bTheta} \left|\E[\estmuSN^\star_i|\btheta_1,\ldots,\btheta_M] - \mu_i\right| = \cO(N^{-1}) \mbox{ and } \\ 	&\sup\nolimits_{\{\btheta_1,\ldots,\btheta_M\}\subset\bTheta}\Var[\estmuSN_i^\star|\btheta_1,\ldots,\btheta_M] = \cO(N^{-1})
	\end{align*}
	as $N\to\infty$. Moreover, the same statement holds when $M\to\infty$.
\end{theorem}

\begin{proof}
	
	By construction, $\estmuSN_i^\star$ only pools replications at $\btheta_j \in\bTheta_i$. For sufficiently large $N$,
	\begin{align*}
		& \left|\E[\estmuSN_i^\star|\btheta_1,\ldots,\btheta_M] - \mu_i\right| \leq \sum_{{j=1 }}^M \gamma_{ij}^\star\left|\E_{\btheta_j}[\estmuSN_{ij}] - \mu_i\right| \stackrel{(*)}{<} \sum_{{j=1, N^\star_j>0 }}^M \frac{\gamma_{ij}^\star}{N_j^\star}\left(|\E_{\btheta_j}[W_{ij}^2(g(\bX)-\mu_i)]| + 1\right) \\
		& \stackrel{(**)}{\leq} \sum\nolimits_{{j=1 }}^M \frac{\gamma_{ij}^\star}{\delta N}\sup_{\btheta_j\in \bTheta_i}\left(|\E_{\btheta_j}[W_{ij}^2(g(\bX)-\mu_i)]| + 1\right) \stackrel{(***)}{=} \frac{1}{\delta N}\sup_{\btheta_j\in \bTheta_i}\left(|\E_{\btheta_j}[W_{ij}^2(g(\bX)-\mu_i)]| + 1\right),
	\end{align*}
	where $(*)$, $(**)$, and $(***)$ follow from Assumption~\ref{assm:boundedmoments},~\eqref{eq:modified.Njstar}, and ${\sum_{{j=1, N^\star_j>0 }}^M \gamma_{ij}^\star=1}$ respectively. Because $\sup_{\btheta_i \in \bTheta}\sup_{\btheta_j\in \bTheta_i}\left(|\E_{\btheta_j}[W_{ij}^2(g(\bX)-\mu_i)]| + 1\right)$ is bounded from
	Assumption~\ref{assm:boundedmoments}, we conclude 
	\[	\sup_{\{\btheta_1,\ldots,\btheta_M\}\in\bTheta} \left|\E[\estmuSN^\star_i|\btheta_1,\ldots,\btheta_M] - \mu_i\right| = \cO(N^{-1})
	\] 
	for finite $M$ as well as when $M\to\infty$.
	For the variance, as all inner replications are independently simulated,
	\begin{align*}
		& \Var[\tmustar_i|\btheta_1,\ldots,\btheta_M] = \sum\nolimits_{j=1}^M (\gamma_{ij}^\star)^2 \Var_{\btheta_j}[\estmuSN_{ij}] \stackrel{(*)}{<} \sum\nolimits_{j=1,N^\star_j>0}^M \frac{(\gamma_{ij}^\star)^2}{N_j^\star} \left(\E_{\btheta_j}[W_{ij}^2(g(\bX)-\mu_i)^2] +1 \right) \\
		& \stackrel{(**)}{\leq} \sum\nolimits_{j=1}^M \frac{\gamma_{ij}^\star}{\delta N} \sup_{\btheta_j \in \bTheta_i} \left(\E_{\btheta_j}[W_{ij}^2(g(\bX)-\mu_i)^2] +1 \right) {=} \frac{1}{\delta N}\sup_{\btheta_j\in \bTheta_i}\left(\E_{\btheta_j}[W_{ij}^2(g(\bX)-\mu_i)^2] +1 \right)
	\end{align*}
	for sufficiently large $N$, where $(*)$ follows from Assumption~\ref{assm:boundedmoments} and~\eqref{eq:modified.Njstar} and $(**)$ holds since ${0 \leq \gamma_{ij}^\star \leq 1}$. Because $\sup_{\btheta_i \in \bTheta}\sup_{\btheta_j\in \bTheta_i} \left(\E_{\btheta_j}[W_{ij}^2(g(\bX)-\mu_i)^2] +1 \right)<\infty$ from Assumption~\ref{assm:boundedmoments}, we conclude
	\[
	\sup_{\{\btheta_1,\ldots,\btheta_M\}\in\bTheta}\Var[\tmustar_i|\btheta_1,\ldots,\btheta_M] = \cO(N^{-1})
	\]
	 for finite $M$ and when $M\to\infty$.
\end{proof}

In Sections~\ref{subsec:indicator}--\ref{subsec:quantile}, we analyze asymptotic properties of Type-(i) and Type-(ii) nested statistics.
For $\widetilde{\zeta} = \sum_{i=1}^M\zeta(\estmuSN_i^\star)/M$, we show that its bias and variance converge in $\cO(N^{-1})$ and $\cO(N^{-1})+\cO(M^{-1})$, respectively, when $\zeta$ is an indicator, hockey stick, or smooth function with bounded second derivative.
Sections~\ref{subsec:indicator}--\ref{subsec:smooth} present different assumptions and proofs for each choice of $\zeta$.
In Section~\ref{subsec:quantile}, we show that the empirical quantile $\estmuSN_{(\lceil M\alpha\rceil)}^\star$ converges to the true $\alpha$-quantile of $\mu(\btheta)$ in $\cO_p(M^{-1/2})+\cO_p(N^{-1/2})$.
%A consistent finding from the analyses is that \emph{choosing $N= \Theta(M)$ for our experiment design leads to the fastest convergence rate for the nested simulation statistics} as $M$ increases. To investigate how simulation budget increases in $M$ under this choice, we analyze the growth rate of the optimized simulation budget when $N=M$ is adopted in Section~\ref{subsec:opt.obj}. For a special case, we show that the simulation budget is $\cO(M^{1+\varepsilon})$ for any $\varepsilon>0$.

\subsection{Indicator function of the conditional mean}\label{subsec:indicator}
Suppose $\zeta(\mu_i) = I(\mu_i\leq \xi)$ for some $\xi \in\real$. Let $\Phi$ be the cumulative distribution function (cdf) of $\mu_i$. By definition, $\E[\zeta(\mu_i)]= \Phi(\xi)$. Thus, we denote the corresponding estimator $\widetilde{\zeta} \equiv M^{-1}\sum_{i=1}^M I(\tmustar_i\leq \xi)$ by $\Phi_{M,N}(\xi)$, where $\Phi_{M,N}(\cdot)$ is the empirical cdf (ecdf) constructed from $\tmustar_1,\ldots,\tmustar_M$.

For ease of exposition, let $\epsilon_i \equiv \sqrt{N}(\tmustar_i-\mu_i)$, which is the scaled estimation error of $\tmustar_i$ so that its limiting distribution is not degenerate as $N \to\infty$.
From Theorem~\ref{thm:convergencerate}, $\E[\epsilon_i|\btheta_i]=\E[\E[\epsilon_i|\btheta_1,\ldots,\btheta_M]|\btheta_i] = \cO(N^{-1/2})$ uniformly for all $\btheta_i\in\bTheta$. Similarly, $\Var[\epsilon_i|\btheta_i]=\cO(1)$ uniformly for all $\btheta_i\in\bTheta$. In the following, we denote the joint distribution of $\mu_i$ and $\epsilon_i$ by $f_i(\mu,\epsilon)$, where $\btheta_i$ is an arbitrary scenario in $\{\btheta_1,\ldots,\btheta_M\}$. Similarly, $f_{ij}(\mu_i,\mu_j,\epsilon_i,\epsilon_j)$ refers to the joint distribution of $\mu_i,\mu_j,\epsilon_i$, and $\epsilon_j$ for some arbitrary $\btheta_i$ and $\btheta_j$ among $\{\btheta_1,\ldots,\btheta_M\}$. We make Assumption~\ref{assm:indicator} below to facilitate Theorem~\ref{thm:indicator.mse} that stipulates the MSE convergence rate of $\Phi_{M,N}(\xi)$.
\begin{assumption}\label{assm:indicator}
	The cdf of $\mu_i$, $\Phi$, is absolutely continuous with continuous probability density function (pdf) $\phi$. For any $\btheta_i\in\{\btheta_1,\ldots,\btheta_M\}$, $f_{i}(\mu,\epsilon)$, is differentiable with respect to $\mu$ for each $M$ and $N$. Moreover, there exist $p_{s,M,N}(\epsilon), s=0,1$, such that $f_i(\mu,\epsilon)\leq p_{0,M,N}(\epsilon)$ and $\left|\frac{\partial f_i(\mu,\epsilon)}{\partial \mu}\right|\leq p_{1,M,N}(\epsilon)$
	for all $\mu$ and for each $M$ and $N$, and
	$\sup_M\sup_{N} \int_{-\infty}^\infty |\epsilon|^k p_{s,M,N}(\epsilon) d\epsilon <\infty$ for $s=0,1$ and $0\leq k \leq 2$.
	Also, for any $\btheta_i\neq \btheta_j$, $f_{ij}(\mu_i,\mu_j,\epsilon_i,\epsilon_j)$, is differentiable with respect to $\mu_i$ and $\mu_j$ for each $M$ and $N$.
	There exist $p_{s,M,N}(\epsilon_i,\epsilon_j), s=0,1$, such that $f_{i,j}(\mu_i,\mu_j,\epsilon_i,\epsilon_j)\leq p_{0,M,N}(\epsilon_i,\epsilon_j)$ and
	$\left|\frac{\partial f_{i,j}(\mu_i,\mu_j,\epsilon_i,\epsilon_j)}{\partial \mu}\right|\leq p_{1,M,N}(\epsilon_i,\epsilon_j)$ for all $\mu_i, \mu_j, i\neq j$ and for each $M$ and $N$, and
	$\sup_M\sup_{N} \int_{-\infty}^\infty \int_{-\infty}^\infty |\epsilon_i|^{k_i} |\epsilon_j|^{k_j} p_{s,M,N}(\epsilon_i,\epsilon_j) d\epsilon_i d\epsilon_j <\infty
	$ for $s=0,1$, and $0\leq k_i, k_j \leq 2$ such that $k_i + k_j \leq 3$.
\end{assumption}
%Although it is difficult to verify without knowing $\Phi$, Assumption~\ref{assm:indicator} is standard in the nested simulation literature~\citep{gordy2010,hong2017}.~\cite{gordy2010} point out that this assumption is expected to be true when the payoff $g(\cdot)$ and the LRs $\frac{h(\cdot;\theta_i)}{h(\cdot;\theta_j)}$ for all $i,j$ are sufficiently smooth.

%\noindent \textbf{Remark}:
Assumption~\ref{assm:indicator} is likely to hold when $M$ and $N$ are large and $\rho(\mu(\btheta)) = \E[\zeta(\mu(\btheta))]$ is well-estimated via SNS. Indeed, Assumption~\ref{assm:indicator} is identical to Assumption~1 in~\cite{gordy2010} made to show the MSE convergence rate for the SNS estimator of $\rho(\mu(\btheta))$ except that i) we define $\epsilon_i$ as the scaled estimation error of the pooled LR estimator, $\tilde{\mu}_i^\star$, not the standard MC estimator, $\bar{\mu}_i$; and ii) we impose the boundedness and moment conditions for the joint distribution of $(\mu_i,\mu_j,\epsilon_i,\epsilon_j)$, not just for that of $(\mu_i,\epsilon_i)$. The latter is because $\tilde{\mu}_i^\star$ pools simulation outputs from $\btheta_j\neq \btheta_i$. Recall that $\tilde{\mu}_i^\star$ is the weighted average of all $\tilde{\mu}_{ij}$'s such that $N_j^\star>0$ and $\gamma_{ij}^\star>0$, where each $\tilde{\mu}_{ij}$ is a sample average of $N_j^\star$ i.i.d. observations of $W_{ij}g(\bX)$. Therefore, for sufficiently large $N$, $\epsilon_i = \sqrt{N}(\tilde{\mu}_i^\star-\mu_i)$ behaves like a normal random variable as in SNS. Note that if $\tilde{\mu}_i^\star$ is replaced by a generic pooled LR estimator $\tilde{\mu}_i$ in~\eqref{eq:muhatthetaSNpooled} (e.g., constructed with suboptimal weights $\gamma_{ij}$), then $\epsilon_i$ does not converge to a normal random variable in general since some $\tilde{\mu}_{ij}$ may have unbounded variance. For optimized $\tilde{\mu}_i^\star$, the variance result in Theorem~\ref{thm:convergencerate} guarantees that a central limit theorem can be applied to $\epsilon_i$.
On the other hand, if SNS does not estimate $\rho(\mu(\btheta))$ well, %(e.g., Assumption~1 in~\cite{gordy2010} is violated),
then there is little hope that Assumption~\ref{assm:indicator} holds.

%For the ERM examples in Section~\ref{sec:expr}, lognormal density $h(\cdot;\btheta)$ is smooth and so is the LR, $\frac{h(\cdot;\theta_i)}{h(\cdot;\theta_j)}$, for all $i,j$.
%Also, $g(\cdot)$ is a sum of option payoffs, which is continuous and has finitely many points that are non-differentiable.

\begin{theorem}\label{thm:indicator.mse}
	Under Assumptions~\ref{assm:MCbasic}--\ref{assm:indicator}, $\E[\Phi_{M,N}(\xi)] - \Phi(\xi) = \cO(N^{-1})$ and $\Var[\Phi_{M,N}(\xi)] = \cO(N^{-1})+\cO(M^{-1})$.
\end{theorem}
\begin{proof}
	
	Let us define $\Phi_M$ as the ecdf constructed from $\mu_1,\ldots,\mu_M$ given the same outer scenarios as $\Phi_{M,N}$.
	Because $\Phi_M$ is an unbiased estimator of $\Phi$, $\E[\Phi_{M,N}(\xi)]-\Phi(\xi) = \E[\Phi_{M,N}(\xi)-\Phi_M(\xi)]$. From definitions,
	$
	\Phi_{M,N}(\xi)-\Phi_{M}(\xi) =M^{-1}\sum_{i=1}^M \left(I(\tmustar_i\leq \xi) - I(\mu_i\leq \xi) \right)
	$.
	For each $i$,
	\begin{align}
		\E[I(\tmustar_i\leq \xi) - I(\mu_i\leq \xi)] & = \int_{-\infty}^\infty \int_{-\infty}^{\xi-\frac{\epsilon}{\sqrt{N}}} f_{i}(\mu,\epsilon)d\mu d\epsilon - \int_{-\infty}^\infty \nonumber \int_{-\infty}^\xi f_{i}(\mu,\epsilon)d\mu d\epsilon \nonumber \\
		& = -\int_{-\infty}^\infty \int_{\xi-{\epsilon}/{\sqrt{N}}}^{\xi} f_{i}(\mu,\epsilon)d\mu d\epsilon.\label{eq:exp.ind.diff}
	\end{align}
	The first-order Taylor series expansion of $f_{i}(\mu,\epsilon)$ at $\mu\in {[\xi-\epsilon/\sqrt{N}, \xi]}$ is
	\begin{equation}\label{eq:marginal.taylor}
		f_{i}(\mu,\epsilon) = f_i(\xi,\epsilon) + \tfrac{\partial f_i(\Check{\mu},\epsilon)}{\partial \mu} (\mu-\xi),
	\end{equation}
	where $\Check{\mu}\in(\mu,\xi)$.
	From Assumption~\ref{assm:indicator}, for all $M$ and $N$,
	\begin{equation}
		\label{eq:integral.density}
		\frac{\epsilon}{\sqrt{N}}f_i(\xi,\epsilon) - \frac{\epsilon^2}{2N} p_{1,M,N}(\epsilon)
		\leq \int_{\xi-\epsilon/\sqrt{N}}^\xi f_{i}(\mu,\epsilon)d\mu \leq
		\frac{\epsilon}{\sqrt{N}}f_i(\xi,\epsilon) + \frac{\epsilon^2}{2N} p_{1,M,N}(\epsilon).
	\end{equation}
	Thus, from~\eqref{eq:exp.ind.diff}, integrating all three sides of~\eqref{eq:integral.density} with respect to $\epsilon \in (-\infty, \infty)$ gives upper and lower bounds to $\E[I(\estmuSN_i^\star\leq \xi) - I(\mu_i\leq \xi)]$.
	Note that $f_i(\xi,\epsilon) = f_i(\epsilon|\xi)\phi(\xi)$, where $ f_i(\epsilon|\xi)$ is the conditional pdf of $\epsilon_i$ given $\mu_i = \xi$. Thus, $\int_{-\infty}^\infty f_i(\xi,\epsilon)\frac{\epsilon}{\sqrt{N}}d\epsilon = \phi(\xi)\E\left[\left.\frac{\epsilon_i}{\sqrt{N}}\right|\mu_i=\xi\right] = \phi(\xi)\E[\tmustar_i-\mu_i|\mu_i=\xi]$, where $\E[\tmustar_i-\mu_i|\mu_i]=\cO(N^{-1})$ uniformly for all $\mu_i$ from Theorem~\ref{thm:convergencerate}. Also, Assumption~\ref{assm:indicator} guarantees that $\int_{-\infty}^\infty \epsilon^2 p_{1,M,N}(\epsilon) d\epsilon$ is bounded.
	Therefore, $\E[I(\tmustar_i\leq \xi) - I(\mu_i\leq \xi)] = \cO(N^{-1})$ for each $i$, which in turn implies $\E[\Phi_{M,N}(\xi)-\Phi_{M}(\xi)] = \cO(N^{-1})$.
	Next, noticing $\Phi_{M,N}(\xi) = \Phi_{M,N}(\xi) - \Phi_M(\xi)+ \Phi_M(\xi)$, $\Var[ \Phi_{M,N}(\xi)]$ can be written as
	\begin{equation}\label{eq:variance.decomp}
		\begin{aligned}
			\Var[ \Phi_{M,N}(\xi)]
			& = \Var[ \Phi_{M,N}(\xi) - \Phi_M(\xi)] + \Var[ \Phi_M(\xi)] +2\Cov[ \Phi_{M,N}(\xi) - \Phi_M(\xi), \Phi_M(\xi)] \\
			& \leq 2\Var[ \Phi_{M,N}(\xi) - \Phi_M(\xi)] + 2\Var[ \Phi_M(\xi)]
		\end{aligned}
	\end{equation}
	Clearly, $\Var[ \Phi_M(\xi)]=\cO(M^{-1})$.
	In the following, we show $\Var[\Phi_{M,N}(\xi) - \Phi_M(\xi)]=\cO(M^{-1}) + \cO(N^{-1})$ by first expanding it as
	\small
	\begin{equation}\label{eq:temp2}
		\frac{1}{M^2} \sum_{i=1}^M \Var[I(\tmustar_i\leq \xi)-I(\mu_i\leq \xi)]
		+ \frac{1}{M^2} \sum_{i=1}^M\sum_{\substack{j=1\\ j\neq i}}^M \Cov[ I(\tmustar_i\leq \xi)-I(\mu_i\leq \xi), I(\tmustar_j\leq \xi)-I(\mu_j\leq \xi)]
	\end{equation}
	\normalsize
	and showing that the leading order of~\eqref{eq:temp2} is $\cO(N^{-1})$.
	First, $\Var[I(\tmustar_i\leq \xi)-I(\mu_i\leq \xi)] = \E[(I(\tmustar_i\leq \xi)-I(\mu_i\leq \xi))^2] - \E[I(\tmustar_i\leq \xi)-I(\mu_i\leq \xi)]^2$, where the latter term is $\cO(N^{-2})$ from the bias derivation above. Because $(I(\tmustar_i\leq \xi)-I(\mu_i\leq \xi))^2=1$, if and only if, $\tmustar_i\leq \xi$ and $\mu_i> \xi$, or $\tmustar_i > \xi$ and $\mu_i \leq \xi$,
	\begin{align*}
		\E[(I(\tmustar_i\leq \xi)-I(\mu_i\leq \xi))^2] = \int_{-\infty}^0 \int_{\xi}^{\xi-\epsilon/\sqrt{N}} f_i(\mu,\epsilon) d\mu d\epsilon + \int_0^{\infty} \int_{\xi-\epsilon/\sqrt{N}}^\xi f_i(\mu,\epsilon) d\mu d\epsilon.
	\end{align*}
	Following the same logic in~\eqref{eq:integral.density}, both inner integrals are bounded from above by $\frac{\epsilon}{\sqrt{N}}f_i(\xi,\epsilon) + \frac{\epsilon^2}{2N} p_{1,M,N}(\epsilon)$. Hence,
	\begin{align*}
		\E[(I(\tmustar_i\leq \xi)-I(\mu_i\leq \xi))^2] \leq \int_{-\infty}^\infty \left\{\frac{\epsilon}{\sqrt{N}}f_i(\xi,\epsilon) + \frac{\epsilon^2}{2N} p_{1,M,N}(\epsilon)\right\} d\epsilon = \cO(N^{-1}).
	\end{align*}
	This concludes that the first summation of~\eqref{eq:temp2} is $\cO((NM)^{-1})$.
	
	The covariance term of~\eqref{eq:temp2} can be rewritten as
	\small
	\begin{equation}\label{eq:cov.indicator}
		\E\left[ \left(I(\tmustar_i\leq \xi)-I(\mu_i\leq \xi)\right)\right.
		\left.\left(I(\tmustar_j\leq \xi)-I(\mu_j\leq \xi)\right)\right]
		- \E\left[I(\tmustar_i\leq \xi)-I(\mu_i\leq \xi)\right]\E\big[I(\tmustar_j\leq \xi)-I(\mu_j\leq \xi)\big].
	\end{equation}
	\normalsize
	The first expectation of~\eqref{eq:cov.indicator} is equal to
	\begin{align*}
		& \Pr\{\tmustar_i \leq \xi, \tmustar_j\leq \xi\} - \Pr\{\mu_i\leq \xi,\tmustar_j\leq \xi\} + \Pr\{ \mu_i\leq\xi,\mu_j\leq \xi\} - \Pr\{\tmustar_i \leq \xi, \mu_j\leq \xi\}                                     \\
		& = \int_{-\infty}^\infty \int_{-\infty}^\infty \int_{\xi-{\epsilon_i}/{\sqrt{N}}}^{\xi}\int_{\xi-{\epsilon_j}/{\sqrt{N}}}^\xi f_{ij}(\mu_i,\mu_j,\epsilon_i,\epsilon_j) d\mu_j d\mu_i d\epsilon_j d\epsilon_i.
	\end{align*}
	Applying the first-order Taylor series expansion to $f_{ij}(\mu_i,\mu_j,\epsilon_i,\epsilon_j)$ gives
	\small
	\begin{align}\label{eq:jointpdf.taylor}
		& f_{ij}(\mu_i,\mu_j,\epsilon_i,\epsilon_j) = f_{ij}(\xi,\xi,\epsilon_i,\epsilon_j) +\frac{\partial f_{ij}(\Check{\mu}_i,\Check{\mu}_j,\epsilon_i,\epsilon_j)}{\partial \mu_i} (\mu_i-\xi) + \frac{\partial f_{ij}(\Check{\mu}_i,\Check{\mu}_j,\epsilon_i,\epsilon_j)}{\partial \mu_j} (\mu_j-\xi)
	\end{align}
	\normalsize
	for some $\Check{\mu}_i \in (\mu_i,\xi)$ and $\Check{\mu}_j \in (\mu_j,\xi)$.
	Under Assumption~\ref{assm:indicator}, the integral of~\eqref{eq:jointpdf.taylor} with respect to $\mu_i\in[\xi-\epsilon_i/\sqrt{N},\xi]$ and $\mu_j\in[\xi-\epsilon_j/\sqrt{N},\xi]$ is lower/upper-bounded by
	\begin{equation}\label{eq:integrated.jointpdf.taylor} \frac{\epsilon_i\epsilon_j}{N}f_{ij}(\xi,\xi,\epsilon_i,\epsilon_j) \mp \frac{|\epsilon_i^2\epsilon_j+\epsilon_j^2\epsilon_i|}{2N^{3/2}} p_{1,M,N}(\epsilon_i,\epsilon_j),
	\end{equation}
	Integrating~\eqref{eq:integrated.jointpdf.taylor} once again with respect to $\epsilon_i, \epsilon_j \in(-\infty,\infty)$,
	we have 
	\[
	\E[ \left(I(\tmustar_i\leq \xi)-I(\mu_i\leq \xi)\right)\left(I(\tmustar_j\leq \xi)-I(\mu_j\leq \xi)\right)] = \cO(N^{-1}).
	\]
	Since it is already shown that $\E\big[I(\tmustar_i\leq \xi)-I(\mu_i\leq \xi)\big]=\cO(N^{-1})$, we conclude $\Cov[ I(\tmustar_i\leq \xi)-I(\mu_i\leq \xi), I(\tmustar_j\leq \xi)-I(\mu_j\leq \xi)]=\cO(N^{-1})$ from~\eqref{eq:cov.indicator}, which in turn implies $\Var[ \Phi_{M,N}(\xi) - \Phi_M(\xi)] = \cO(M^{-1}) + \cO(N^{-1})$ from~\eqref{eq:variance.decomp}.
\end{proof}

\subsection{Hockey stick function of the conditional mean}\label{subsec:hockeystick}
Next, we analyze the case when $\zeta$ is a hockey stick function, i.e., $\zeta(\mu_i)=\max\{\mu-\xi,0\}=(\mu-\xi)I(\mu_i>\xi)$ for some $\xi\in\real$, which requires the following additional moment conditions.

\begin{assumption}\label{assm:hockeystick}
	For $p_{1,N,M}(\epsilon)$ in Assumption~\ref{assm:indicator}, $
	\sup_M\sup_{N} \int_{-\infty}^\infty |\epsilon|^3 p_{1,M,N}(\epsilon) d\epsilon <\infty$.
	Similarly, for $p_{s,N,M}(\epsilon_i,\epsilon_j)$ defined in Assumption~\ref{assm:indicator},\
	\begin{align*}
		\sup\nolimits_M\sup\nolimits_{N} \int_{-\infty}^\infty \int_{-\infty}^\infty |\epsilon_i|^{k_i} |\epsilon_j|^{k_j} p_{s,M,N}(\epsilon_i,\epsilon_j) d\epsilon_i d\epsilon_j <\infty
	\end{align*}
	for $s=0,1$, $0\leq k_i, k_j \leq 3$ and $k_i + k_j \leq 5$. Also, for each $M$ and $N$, there exist functions $q_{s,M,N}(\mu_i,\epsilon_i,\epsilon_j), s=0,1$, such that $f_{ij}(\mu_i,\mu_j,\epsilon_i,\epsilon_j) \leq q_{0,M,N}(\mu_i,\epsilon_i,\epsilon_j)$ and $\left|\frac{\partial f_{ij}(\mu_i,\mu_j,\epsilon_i,\epsilon_j)}{\partial \mu_j}\right| \leq q_{1,M,N}(\mu_i,\epsilon_i,\epsilon_j)$ for all $\mu_i, \mu_j, \epsilon_i$, and $\epsilon_j$.
	Lastly, for $s=0,1$,
	\begin{align*}
		\sup_{M}\sup_{N} \int_{-\infty}^\infty \int_{-\infty}^\infty \int_{\xi}^\infty |\mu_i|^{k_m}|\epsilon_i|^{k_i}|\epsilon_j|^{k_j}q_{s,M,N}(\mu_i,\epsilon_i,\epsilon_j) d\mu_i d\epsilon_i d\epsilon_j
		< \infty
	\end{align*}
	for any $\xi\in\real, 0 \leq k_m \leq 1$, and $0 \leq k_i, k_j \leq 3$ {such that} $k_i+k_j \leq 4$.
\end{assumption}

Similar to Assumption~\ref{assm:indicator}, Assumption~\ref{assm:hockeystick} is an extension of Assumption~1 in~\citet{gordy2010} to analyze the hockey stick function, $\zeta$, and is likely to hold when SNS estimates $\E[\zeta(\mu(\theta))]$ well.
%expected to be true in an ERM example where there are a few options in the portfolio with sufficiently smooth payoffs~\citep{gordy2010}.
Theorem~\ref{thm:hockey.stick} below shows that the hockey stick function results in the same bias and variance convergence rates as the those of the indicator function. Due to the space limit, we present the proof of Theorem~\ref{thm:hockey.stick} in~\ref{app:proof.hockey}.

\begin{theorem}\label{thm:hockey.stick}
	Suppose $\zeta(\mu_i) = (\mu_i-\xi) I(\mu_i> \xi)$ for some constant $\xi\in\real$ and $\widetilde{\zeta} = \sum_{i=1}^M\zeta(\tmustar_i)/M$.
	Under Assumptions~\ref{assm:MCbasic}--\ref{assm:hockeystick}, $\E[\widetilde{\zeta}-\zeta(\mu_i)] = \cO(N^{-1})$ and $\Var[\widetilde{\zeta}] = \cO(M^{-1})+\cO(N^{-1})$.
\end{theorem}

\subsection{Smooth function of the conditional mean}\label{subsec:smooth}
To analyze the case when $\zeta$ is a smooth function of $\mu_i$, we make  the following assumption.

\begin{assumption}\label{assm:continuous.risk.measure}
	The continuous function, $\zeta:\real\to\real$, is twice differentiable everywhere with bounded second derivative $\zeta^{\prime\prime}$. Also, $\E[(\zeta(\mu_i))^2], \E[(\zeta^\prime(\mu_i)\epsilon_i)^2]$ and $\E[\epsilon_i^4]$ are bounded.
\end{assumption}

Similar to the indicator and the hockey stick functions, the following theorem shows that the MSE convergence rate of $\widetilde{\zeta}$ is $\cO(M^{-1})+\cO(N^{-1})$ for $\zeta$ that satisfies Assumption~\ref{assm:continuous.risk.measure}.

\begin{theorem}\label{thm:smooth.func}
	Suppose $\zeta$ satisfies Assumption~\ref{assm:continuous.risk.measure} and $\widetilde{\zeta} = \sum_{i=1}^M \zeta(\estmuSN_i^\star)/M$. Under Assumptions~\ref{assm:MCbasic} and~\ref{assm:boundedmoments}, $\E[\widetilde{\zeta}-\zeta(\mu_i)] = \cO(N^{-1})$ and $\Var[\widetilde{\zeta}] = \cO(M^{-1})+\cO(N^{-1})$.
\end{theorem}

The proof of Theorem~\ref{thm:smooth.func} can be found in~\ref{app:proof.smooth}.

\subsection{Quantile of the conditional mean}\label{subsec:quantile}
Let $\widetilde{q}_\alpha = \tmustar_{(\lceil M\alpha\rceil)}$.
In this section, we prove weak consistency of $\widetilde{q}_\alpha$ as $M$ and $N$ increase.
Recall that in SNS, $q_\alpha$ is estimated by the $\lceil M\alpha\rceil$th order statistic of $M$ independent conditional mean estimators.
In our design, $\widetilde{q}_\alpha$ is an order statistic of the correlated estimators, $\tmustar_1,\ldots,\tmustar_M$.
Consistency of an empirical quantile estimator constructed from dependent outputs has been studied~\citep{sen1972,heidelberger1984} under the assumption that the output sequence has a strong mixing property, which ensures that pairwise correlation between distant outputs in the sequence vanishes.
Our pooled LR estimators do not have this property.
Instead, their pairwise correlation decreases as $N$ increases.

To show weak consistency of $\widetilde{q}_\alpha$, we need the following intermediate result. Lemma~\ref{lem:uniform.ecdf} states that $\Phi_{M,N}(\cdot)$ is uniformly weakly consistent to $\Phi(\cdot)$; see~\ref{app:lemmaproof} for its proof.
Theorem~\ref{thm:quantile} is the main result of this section.

\begin{lemma}\label{lem:uniform.ecdf}
	Under Assumptions~\ref{assm:MCbasic}--\ref{assm:indicator}, $\sup_{\xi\in\real}|\Phi_{M,N}(\xi)-\Phi(\xi)| = \cO_p(M^{-1/2})+\cO_p(N^{-1/2})$.
\end{lemma}

\begin{theorem}\label{thm:quantile}
	Suppose Assumptions~\ref{assm:MCbasic}--\ref{assm:indicator} hold and $\phi(q_\alpha)>0$ for given $0<\alpha<1$. Then, $|\widetilde{q}_\alpha - q_\alpha|= \cO_p(M^{-1/2})+\cO_p(N^{-1/2})$.
\end{theorem}
\begin{proof}
	For each $M$, $|\Phi_{M,N}(\widetilde{q}_\alpha) - \alpha| \leq 1/M$.
	Also, Lemma~\ref{lem:uniform.ecdf} implies $|\Phi_{M,N}(\widetilde{q}_\alpha)-\Phi(\widetilde{q}_\alpha)|\leq \sup_{\xi\in\real}|\Phi_{M,N}(\xi)-\Phi(\xi)| = \cO_p(M^{-1/2})+\cO_p(N^{-1/2})$. Therefore, $|\Phi(\widetilde{q}_\alpha)-\alpha| = \cO_p(M^{-1/2})+\cO_p(N^{-1/2})$.
	Then, for sufficiently large $M$ and $N$, there exists $U^\star \in (\Phi(\widetilde{q}_\alpha), \alpha)$ such that $\phi(\Phi^{-1}(U^\star)) >0$ and
	$\Phi^{-1}(\Phi(\widetilde{q}_\alpha)) = \Phi^{-1}(\alpha) + \frac{1}{\phi(\Phi^{-1}(U^\star))}(\Phi(\widetilde{q}_\alpha)-\alpha)$
	with probability arbitrarily close to $1$.
	Because $\phi$ is bounded in a neighborhood of $q_\alpha$, $|\Phi^{-1}(\Phi(\widetilde{q}_\alpha)) - \Phi^{-1}(\alpha)|=|\widetilde{q}_\alpha - q_\alpha| = \cO_p(M^{-1/2}) + \cO_p(N^{-1/2})$.
\end{proof}

\section{Growth Rate of Simulation Budget in the Optimal Design}\label{sec:opt.obj}
Section~\ref{sec:asymptotics} suggests that selecting $N=\Theta(M)$ leads to the most efficient experiment design.
In general, it is difficult to derive an analytical expression %of optimal objective function value of~\eqref{prob:LP}, i.e.,
for the minimized simulation budget as a function of $M$ and $N$.
In this section, we study how quickly the optimal budget grows in $M$ asymptotically when we set $N=M$ and the input distribution of the inner simulation belongs to an exponential family.
\begin{definition}\label{def:exp.fam}
	A probability distribution function, $h(\bx;\btheta)$, is said to belong to an \emph{exponential family} in the canonical form, if it can be written as
	\begin{equation}\label{eq:expon.fam}
		h(\bx;\btheta) = B(\bx)\exp(\btheta^\top T(\bx) - A(\btheta)),
	\end{equation}
	where $B(\bx)$, $T(\bx)$, $\btheta$, and $A(\btheta) \equiv \ln \left(\int_{\cX} B(\bx) \exp(\btheta^\top T(\bx))d\bx\right)$ are referred to as the {base measure}, {sufficient statistic}, {natural parameter}, and {log-partition function}, respectively.
	The {natural parameter space} is defined as $\bar{\bTheta} \equiv\left\lbrace \btheta:\int_{\cX} B(\bx) \exp(\btheta^\top T(\bx))d\bx<\infty\right\rbrace$.
	An exponential family is called \emph{regular} if $\bar{\bTheta}$ is a non-empty open set in $\real^p$.
\end{definition}
\begin{lemma}\label{lem:exp.EW2}
	Suppose $h(\bx;\btheta)$ belongs to an exponential family with natural parameter space~$\bar{\bTheta}$.
	Then, for any $\btheta_i,\btheta_j\in\bar{\bTheta}$, $\E_{\btheta_j}[W_{ij}^2] = \exp(A(\btheta_j) -2A(\btheta_i)+ A(2\btheta_i-\btheta_j))$ if $2\btheta_i-\btheta_j \in \bar{\bTheta}$ or $\infty$ otherwise.
	% \begin{equation}\label{eq:exp.EW2}
		% 	\E_{\btheta_j}[W_{ij}^2] = \begin{cases}
			% 		\exp(A(\btheta_j) -2A(\btheta_i)+ A(2\btheta_i-\btheta_j)), & \mbox{ if } 2\btheta_i-\btheta_j \in \bar{\bTheta},     \\
			% 		\infty,                                                     & \mbox{ if } 2\btheta_i-\btheta_j \not\in \bar{\bTheta}.
			% 	\end{cases}
		% \end{equation}
\end{lemma}
\begin{proof}
	Substituting~\eqref{eq:expon.fam} into the definition of likelihood ratio $W_{ij}$ we have \small
	\begin{align*}
		\E_{\btheta_j}[W_{ij}^2] & = \int_\cX \left(\frac{h(\bx;\btheta_i)}{h(\bx;\btheta_j)}\right)^2h(\bx;\btheta_j)d\bx = \exp(A(\btheta_j) -2A(\btheta_i))\int_\cX B(\bx)\exp((2\btheta_i-\btheta_j)^\top T(\bx))d\bx      \\
		& = \exp(A(\btheta_j) -2A(\btheta_i)+A(2\btheta_i-\btheta_j))\int_\cX \underbrace{B(x)\exp((2\btheta_i-\btheta_j)^\top T(\bx) - A(2\btheta_i-\btheta_j))}_{=h(\bx;2\btheta_i-\btheta_j)}d\bx.
	\end{align*} \normalsize
	Note that $\int_\cX h(\bx;2\btheta_i-\btheta_j) d\bx$ equals one, if $2\btheta_i-\btheta_j \in \bar{\bTheta}$, and diverges, otherwise.
\end{proof}

Lemma~\ref{lem:exp.EW2} implies that $(\E_{\btheta_j}[W_{ij}^2])^{-1}=0$, if $2\btheta_i-\btheta_j \not\in \bar{\bTheta}$.
In this case, the optimal design does not pool replications from $h(\bx;\btheta_j)$ to estimate $\mu_i$ as $\ESS_{ij} = 0$.
\begin{assumption}\label{assm:exp.fam}
	The following holds for $\bTheta$ and $h(\mathbf{x};\btheta)$:
	\begin{enumerate}[label=(\roman*), itemsep=0pt, topsep=0pt]
		\item $h(\mathbf{x};\btheta)$ belongs to a regular exponential family with natural parameter space $\bar{\bTheta}$.\label{assm:exp.fam1}
		\item $\bTheta$ is a closed and bounded (or, equivalently, compact) subset of $\bar{\bTheta}$.\label{assm:exp.fam2}
	\end{enumerate}
\end{assumption}
Assumption~\ref{assm:exp.fam}.\ref{assm:exp.fam1} includes many common choices of input models such as exponential, normal, and log-normal distributions. %input distribution for risk management applications.
%Requiring $\bar{\bTheta}$ to be an non-empty open set (so $h(\bx;\btheta)$ is regular) is a reasonable assumption because otherwise it is meaningless to consider an input distribution without any valid parameter.
Assumption~\ref{assm:exp.fam}.\ref{assm:exp.fam2} is not restrictive as the bounds for $\bTheta$ can be chosen arbitrarily large. Under Assumption~\ref{assm:exp.fam}, we first show the following lemma before stating the main theorem of this section.

\begin{lemma} \label{lem:subcover}
	Suppose Assumption~\ref{assm:exp.fam} holds. Then, there exists finite set $\bTheta_{M^*}=\{\btheta_1,\ldots,\btheta_{M^*}\}$ such that for each $\theta_i\in\bTheta$, $n_{ij}^e>1/2$ for some $\btheta_j\in\bTheta_{M^*}$.
\end{lemma}
\begin{proof}
	From~\eqref{eq:expon.fam}, the moment generating function of the sufficient statistic, $T(\bx)$, of $h(\bx;\btheta)$ can be derived as $M_T(\bm{u})=\E[e^{\bm{u}^\top T(\bx)}]=e^{A(\btheta + \bm{u}) - A(\btheta)}$ for $\bm{u} \in \real^p$. Because $\bar{\bTheta}$ is a non-empty open set, $M_T(\mathbf{u})$ is finite for any $\btheta\in\bar{\bTheta}$ and sufficiently small $\bm{u}$. This implies that $\E_{\btheta}[T(\bx)]< \infty$ for any $\btheta\in \bar{\bTheta}$. From Definition~\ref{def:exp.fam}, we have \small
	\[
	\nabla_{\btheta} A(\btheta) = \frac{\nabla_{\btheta}\left(\int_{\cX} B(\bx) \exp(\btheta^\top T(\bx))d\bx\right)}{\int_{\cX} B(\bx) \exp(\btheta^\top T(\bx))d\bx} = \int_{\cX} T(\bx)B(\bx) \exp(\btheta^\top T(\bx)-A(\btheta))d\bx = \E_{\btheta}[T(\bx)]<\infty.
	\] \normalsize
	Therefore, $A(\btheta)$ is continuous in $\btheta\in\bar\bTheta$ as its gradient exists everywhere in $\bar\bTheta$.

	Consider any $\btheta_i,\btheta_j \in {\bTheta}$. By Lemma~\ref{lem:exp.EW2} we have $\E_{\btheta_j}[W_{jj}^2] = 1$ and that $\E_{\btheta_j}[W_{ij}^2]$ is continuous in $\btheta_i$ because $A(\cdot)$ is continuous.
	Then, there exists open neighborhood $\mathcal{N}(\btheta_j)$ of $\btheta_j$ such that $\E_{\btheta_j}[W_{ij}^2] < 2$ for all $\btheta_i\in \mathcal{N}(\btheta_j)$.
	Equivalently, $\ESS_{ij} = 1/\E_{\btheta_j}[W_{ij}^2] > 1/2$ for all $\btheta_i\in \mathcal{N}(\btheta_j)$.
	We can construct $\mathcal{N}(\btheta_j)$ for each $\btheta_j\in {\bTheta}$ and the collection of them is an open cover of $\bTheta$. Then, there exists finite subcover $\bTheta_{M^*}=\{\btheta_1,\ldots,\btheta_{M^*}\}$ of the open cover by the compactness of $\bTheta$. Consequently, for any $\theta_i\in\bTheta,$ there exists $\theta_j\in\bTheta_{M^*}$ such that $\btheta_i\in\mathcal{N}(\btheta_j)$ and thus, the conclusion follows.
\end{proof}

Lemma~\ref{lem:subcover} implies that by simulating $2M$ inner replications at all outer scenarios in $\bTheta_{M^*}$, we can guarantee that all $\btheta_i\in\bTheta$ have ESS no less than $M$. Building upon this insight, Theorem~\ref{thm:linear.growth.obj.val} shows the main result of this section.

\begin{theorem}\label{thm:linear.growth.obj.val}
	Suppose Assumption~\ref{assm:exp.fam} holds. Furthermore, for $M>M^*$, suppose we choose the first $M^*$ outer scenarios to be those in $\bTheta_{M^*}$ and let $N=M$. Let $\Gamma_M$ be the optimal objective function value of~\eqref{prob:LP} with $M=N$.
	Then, %there exists a sequence of experiment designs with increasing number, $M$, of scenarios in $\bTheta$ such that
	$\Gamma_{M} = \cO(M)$.
\end{theorem}
\begin{proof}
	As discussed, by Lemma~\ref{lem:subcover}, $N_{j}=2M$, $j=1,\ldots,M^*$, is a feasible solution to~\eqref{prob:LP} with the resulting objective function value of
	$2MM^*$.
	Therefore, $\Gamma_M \leq 2MM^*$ for all $M >M^*$, which implies $\Gamma_M = \cO(M)$.
\end{proof}

Recall that the MSE convergence rates of the nested statistic estimators studied in Sections~\ref{subsec:indicator}--\ref{subsec:smooth} are $\cO(M^{-1})$. Therefore, Theorem~\ref{thm:linear.growth.obj.val} stipulates that the MSEs converge in $\cO(\Gamma^{-1})$ given simulation budget $\Gamma$ when Assumption~\ref{assm:exp.fam} holds.

The effect of fixing the first $M^*$ outer scenarios as described in Theorem~\ref{thm:linear.growth.obj.val} is negligible as $M$ increases. In Section~\ref{subsec:IU.inventory}, we empirically observe that the minimized simulation budget indeed grows linearly in $M$ when $N=M$ even if all outer scenarios are chosen at random.

%As discussed in Section~\ref{sec:intro}, the existing nested simulation experiment designs with the fastest proven MSE convergence rates for the estimated nested statistics are those that pool outer scenarios via parametric regression~\citep{broadie2015} and kernel ridge regression~\citep{wang2022smooth}. In both cases, the convergence rates are strictly .
%Recall that for the standard nested simulation design, the optimal choice for $N$ is $\cO(M^{1/2})$, which leads to the total simulation budget of $\cO(M^{3/2})$.
%Theorem~\ref{thm:linear.growth.obj.val} demonstrates the efficiency of our optimal design.

In the proof for Theorem~\ref{thm:linear.growth.obj.val}, we postulate the existence of $\bTheta_{M^*}$.
In~\ref{app:examples}, we provide two examples to illustrate how $\bTheta_{M^*}$ can be constructed.

\section{Numerical Studies}\label{sec:expr}
In this section, we present three numerical examples to demonstrate the performance of the proposed optimal design.
The following experiment designs are compared:
\begin{itemize}[leftmargin=*, itemsep=0pt, topsep=0pt]
	\item \emph{Optimal design}: Given the outer scenarios, the optimal sampling and pooling decisions are obtained by solving~\eqref{prob:LP} with $N=M$. Let $\Gamma$ denote the minimized simulation budget in our optimal design.
	We emphasize again that $N$ is only a benchmark number of inner replications for setting the ESS constraints in~\eqref{prob:LP}. The optimal design does not run $N$ inner replications in each scenario.
	
	\item \emph{SNS}$^{+}$: Standard nested simulation design with the same $M$ outer scenarios as the optimal design and  {$N=M$, which results in the total budget of $MN=M^2$.
		{SNS}$^{+}$ sets a benchmark for the precision requirement for the optimal design as the latter shows the same MSE convergence rate as the standard nested simulation when $N$ is set to $M$.}
	
	\item \emph{SNS}: Standard nested simulation design that consumes the same simulation budget $\Gamma$ as the optimal design.
	Following the recommendations in~\cite{gordy2010}, we set $M = \lceil\Gamma^{2/3}\rceil$ and $N = \lceil\Gamma^{1/3}\rceil$, respectively.
	The SNS design's simulation budget may be slightly higher than the optimal design's due to the ceiling function.
	
	\item \emph{Regression}: Pooling based on a linear regression model; as recommended by~\cite{broadie2015}, we sample $\Gamma$ initial design points and run one replication at each to fit the model. We chose different basis functions to suit each example. The fitted model is evaluated at the same $M$ outer scenarios used in the optimal design to compute the performance measures.
\end{itemize}

The first two examples are ERM problems whose objectives are to evaluate the risk of the future portfolio values due to the price fluctuation of the underlying assets.
The first example has two options written on one underlying asset while the second example considers a portfolio with 300 options written on 10 assets.
%For ease of exposition, the first example considers a straddle option portfolio, i.e., one vanilla European call and one vanilla European put, written on one underlying asset.
%This simple example is easy to visualize and allows us to demonstrate the optimal design in detail.
%The second example considers a portfolio with 300 options, including exotic options such as Asian and barrier options.
%This realistic example further showcases the computational feasibility and efficiency of the optimal design compared to others.
The third example demonstrates Bayesian input model uncertainty quantification applied to a multi-product newsvendor problem.
The experiment designs are evaluated by the coverage probabilities and widths of the credible intervals (CrIs) of the expected profit.

The expression for $\mu(\btheta)$ is known in all three examples, which facilitates {performance evaluation} of the experiment designs in comparison. All source codes are available at \url{https://github.com/INFORMSJoC/2022.0392/}~\citep{GitHub}.

\subsection{Single-Asset Enterprise Risk Management (ERM) Problem}\label{subsec:ERM}
In this section, we consider a straddle option portfolio that consists of a call option and a put option with the same underlying stock, strike, and maturity.
Let the underlying stock price at $t=0$ be $S_0=\$100$.
We assume that the stock is non-dividend-paying and follows the Black-Scholes model with $\eta=5\%$ annualized expected rate of return and $\sigma=30\%$ annualized volatility.
The annualized risk-free rate is $r=2\%$.
The common maturity of both options is $T=2$ years and the common strike price is $K=\$110$.

We are interested in the value of the portfolio in $\tau=1/4$ year from now.
The future portfolio value can be evaluated via nested simulation by first simulating the stock price at $\tau$, $S_\tau$, then computing the expected payoff of the stock at the maturity given $S_\tau$ by running inner replications. Here, the outer scenario is one-dimensional $\theta = S_\tau$, and $\mu(\theta)$ corresponds to the conditional expected payoff given $S_\tau$.
From the Black-Scholes model, the outer scenarios are simulated under the real-world measure, i.e., $S_\tau|S_0 =S_0 e^{Z_\tau}$, where the rate of return $Z_\tau$, is distributed as $\mathcal{N}((\eta-\frac{1}{2}\sigma^2)\tau, \sigma^2 \tau)$.
Thus, $\theta = S_\tau$ has a log-normal distribution whose density function is shown in Figure~\ref{subfig:OutDist}.
In this example, we choose $M$ equally-spaced quantiles of the log-normal distribution as the outer scenarios.

Given any scenario $\theta = S_\tau$, the input random variable for the inner replication is the stock price at maturity, $X = S_T$.
From the Black-Scholes model, the inner simulation is conducted under the risk-neutral measure, i.e., $S_T|S_\tau=S_\tau e^{Z_T}$, where $Z_T\sim\mathcal{N}((r-\frac{1}{2}\sigma^2)(T-\tau),\sigma^2(T-\tau))$.
The simulation model $g$ computes the discounted payoff of the straddle option from $X$; $g(X)=e^{-r(T-t)}[(K-X)_++(X-K)_+]$, where $(x)_+=\max\{x,0\}$. Thus, $\mu(\theta) = \E_{\theta}[g(X)] = \E[g(X)|S_\tau]$.
The analytical expression for $\mu(\theta)$ can be derived from the Black-Scholes model without simulation.
Figure~\ref{subfig:ConditionalMean} depicts $\mu(\theta)$, which shows that the portfolio value is high when $S_\tau$ takes extreme values.
We consider a financial institution that offers the straddle strategies to investors, i.e., a short position.
So the company suffers large losses when $\mu(\theta)$ is large, or when $S_\tau$ takes extreme values.
For the regression approach, the weighted Laguerre polynomials up to order $2$ are adopted as the basis functions, which is a common choice in pricing American options~\citep{longstaff2001valuing}.

\begin{figure}
	\centering
	\begin{subfigure}[b]{0.45\textwidth}
		\includegraphics[width=\textwidth]{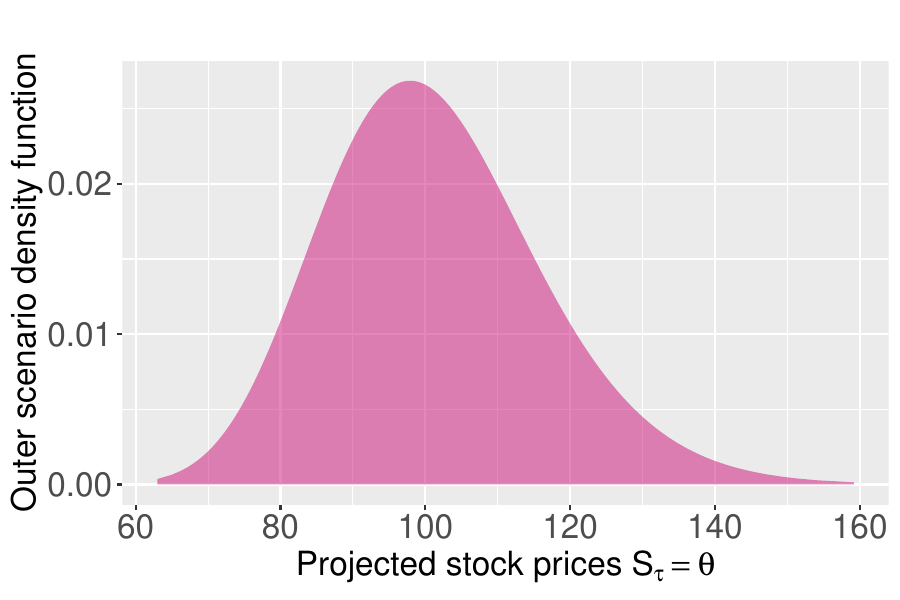}
		\caption{pdf of outer scenarios $S_\tau = \theta$.}\label{subfig:OutDist}
	\end{subfigure}
	\hfill
	\begin{subfigure}[b]{0.45\textwidth}
		\includegraphics[width=\textwidth]{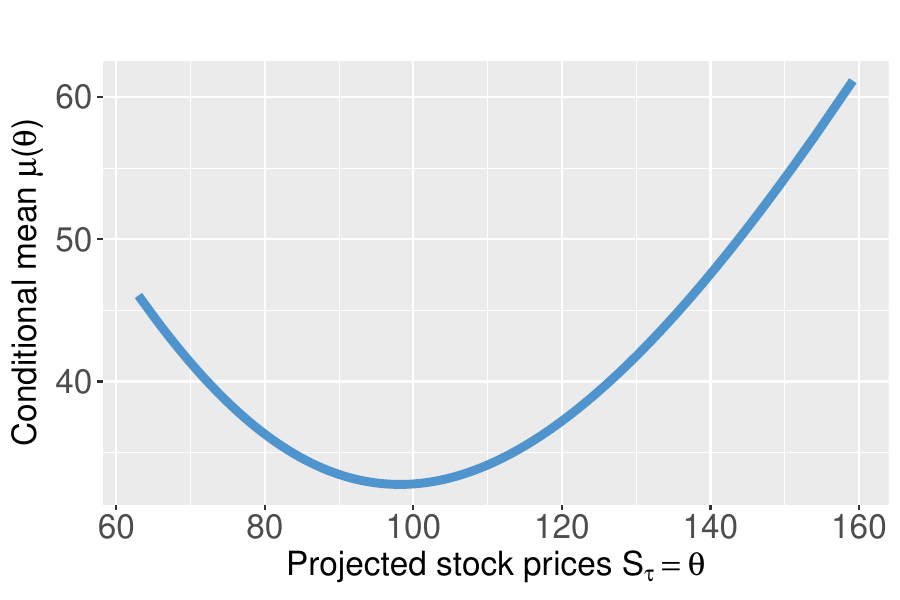}
		\caption{Conditional expectation $\mu(\theta)$.}\label{subfig:ConditionalMean}
	\end{subfigure}
	\caption{Problem settings 1D ERM example: (a) shows the pdf of $\theta$ and (b) plots $\mu(\theta)$ against $\theta$.}
\end{figure}

Firstly, we consider a fixed number of outer scenarios, i.e., $M=1{,}000$ equally-spaced quantiles from the outer scenario's log-normal distribution, and compare the estimation accuracies for different nested simulation designs.
We set $N = M = 1{,}000$ when solving~\eqref{prob:LP} for the optimal design.
The resulting simulation budget in the optimal design is $\Gamma = 2{,}148$ replications.
The optimal sampling decision, $\{c_j^\star\}$, allocates $29\%$, $21\%$, $21\%$, and $29\%$ of the simulation budget $\Gamma$ to only 4 scenarios, i.e., $\theta=70.63, 71.01, 141.18$, and $141.94$. These points are near the two tail ends of distribution of $\theta$, which can be explained by the ESS formula.
In this example, the inner simulation random variable $X=S_T|S_\tau$ for any outer scenario $\theta=S_\tau$ follows a log-normal distribution with the common variance $\sigma^2(T-\tau)$ with mean $m_\theta=\ln \theta + (r-\frac{1}{2}\sigma^2)(T-\tau)$.
Thus, for any $\theta_i$ and $\theta_j$, the LR is $W_{ij}(X)=\exp\left(\frac{(\ln X - m_{\theta_j})^2-(\ln X - m_{\theta_i})^2}{2\sigma^2(T-\tau)}\right)$ and its second moment is $\E_{\theta_j}[W_{ij}^2] = \exp\left(\frac{(m_{\theta_i}-m_{\theta_j})^2}{\sigma^2}\right)=\exp\left(\frac{(\ln \theta_i-\ln \theta_j)^2}{\sigma^2}\right)$.
Consequently, the ESS of using one sample from $\theta_j$ to estimate the conditional mean for scenario $\theta_i$ is inversely proportional to $\exp((\ln\theta_i-\ln \theta_j)^2)$, which indicates that the ESS falls off quickly when $\theta_i\neq\theta_j$. Therefore, the $\theta$s on the tails benefit the most by pooling from nearby $\theta$s, whereas the $\theta$s in the middle can achieve the desired ESS by pooling from both tails.

For each of the $1{,}000$ $\theta$'s, we run $10{,}000$ macro replications for all three simulation designs to estimate the corresponding $\mu(\theta)$.
For each simulation design, the $2.5\%$ and $97.5\%$ quantiles of the estimated $\mu(\theta)$ from these $10{,}000$ macro runs at all $\theta$'s are connected to form a 95\% confidence band.
Three resulting confidence bands, one for each simulation design, are depicted in Figure~\ref{fig:95ConfBands}.
Note that SNS is omitted from Figure~\ref{fig:95ConfBands} as its confidence band is too wide to be compared {in} the same plot.
Figure~\ref{subfig:95ConfBandsZoomed} is a zoomed in version of Figure~\ref{subfig:95ConfBandsAll} near the mode of the outer scenario distribution.
The black curve in these figures shows the true option portfolio value $\mu(\theta)$ calculated from closed-form option pricing formulas in the Black-Scholes model.
We see in Figure~\ref{subfig:95ConfBandsAll} that the optimal design's confidence band is almost indistinguishable from that of the SNS$^+$ design.
This demonstrates that the precision requirement~\eqref{eq:var.constraint} in our optimization formulation is effective despite approximation~\eqref{eq:ESS}.
The small approximation error~\eqref{eq:ESS.liu} shows up in the zoomed Figure~\ref{subfig:95ConfBandsZoomed}, where we see the optimal design's confidence band is slightly wider than the SNS$^+$ design's.
We note that the optimal design achieves similar precision as the SNS$^+$ design with significant smaller simulation budget: The optimal design's simulation budget $2$,$148$ is 465 times smaller than that of the SNS$^+$ design ($MN = 10^6$).
We also see in Figure~\ref{subfig:95ConfBandsAll} that the regression-based design's confidence band is significantly wider than the optimal design's and the SNS$^+$ design's when $\theta=S_\tau$ takes extreme values.
In general, users do not know the relationship between $\mu(\theta)$ and $\theta$ prior to simulation. So, particularly low precision in any region could significantly impact tail risk measure estimation thus is undesirable.

\begin{figure}[h!]
	\centering
	\begin{subfigure}[b]{0.45\textwidth}
		\includegraphics[width=\textwidth]{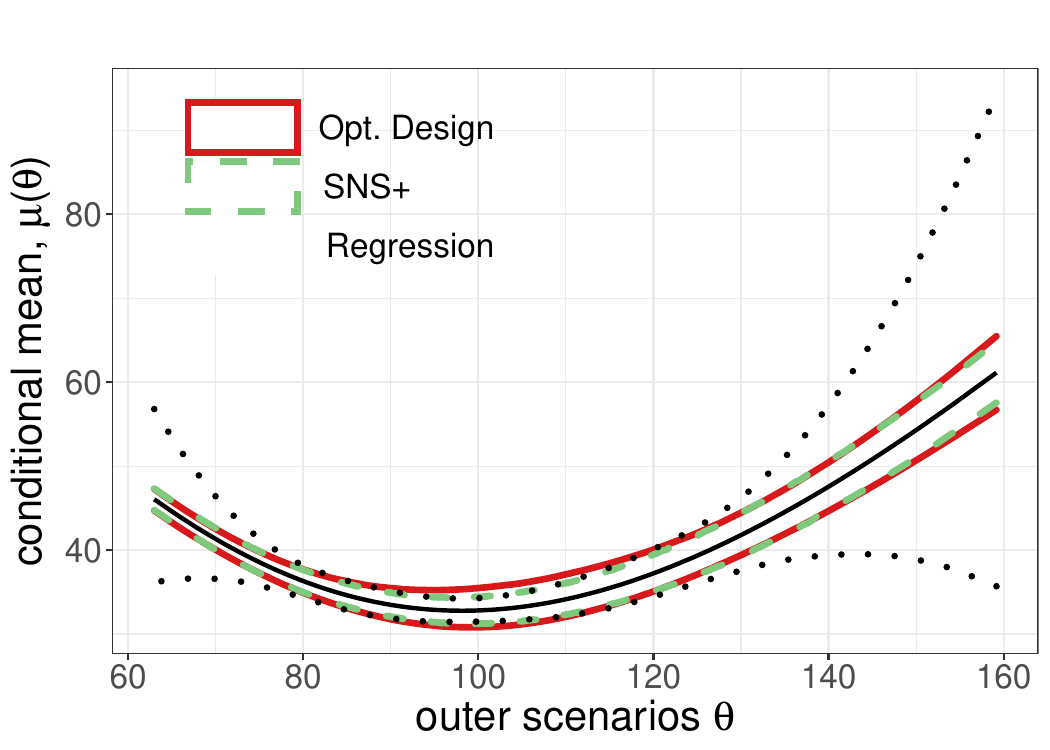}
		\caption{\small}\label{subfig:95ConfBandsAll}
	\end{subfigure}
	\hfill
	\begin{subfigure}[b]{0.45\textwidth}
		\includegraphics[width=\textwidth]{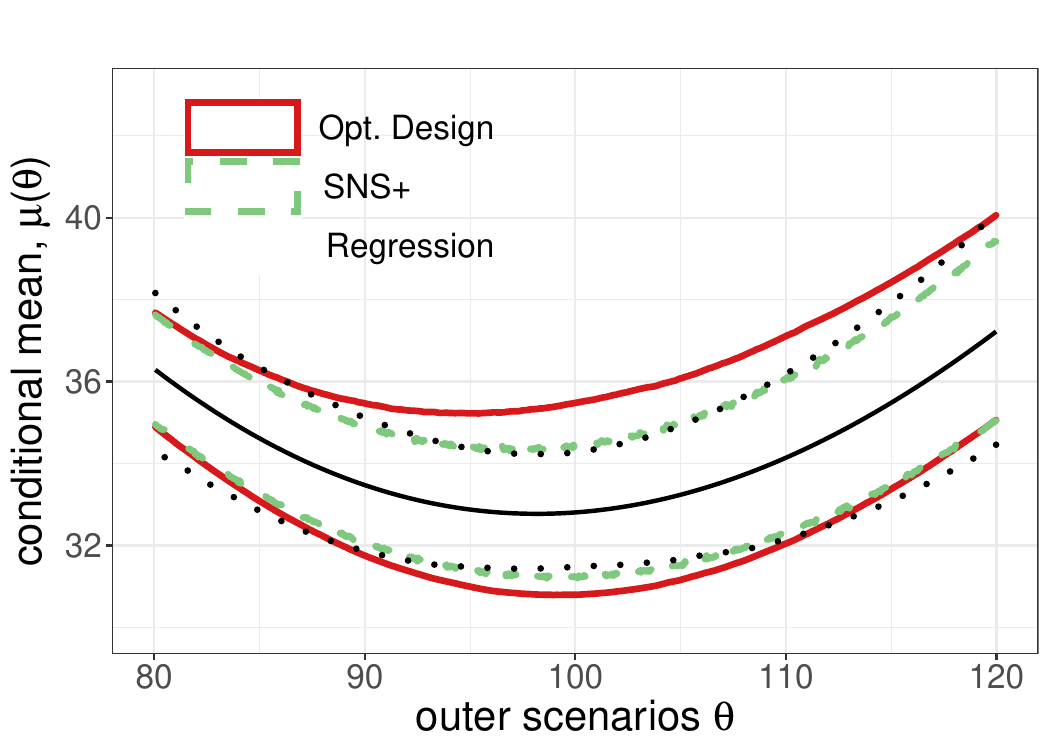}
		\caption{\small}\label{subfig:95ConfBandsZoomed}
	\end{subfigure}
	\caption{(a) shows the 95\% confidence bands for different experiment designs based on 10{,}000 macro runs; (b) is a zoomed-in version of (a) around the median of $\theta$.}\label{fig:95ConfBands}
\end{figure}

Next, we consider different numbers of outer scenarios $M$ and compare the MSEs of the portfolio risk measures computed from all four designs.
We set confidence level $\alpha = 0.99$ to emulate tail risk estimation.
The four nested statistics we consider are: (i) the $\alpha$-quantile of $\mu(\theta)$, (ii) indicator risk measure $\E[I(\mu(\theta)>49)]$, (iii) hockey stick risk measure $\E[(\mu(\theta)-49)_+]$, and (iv) squared tail risk measure $\E[(\mu(\theta)-49)_+^2]$.
The latter three measures emulate probability of large losses, expected excess loss, and expected squared excess loss (similar to semivariance, but measures variability in the tail), respectively. All are common in practical ERM problems.
These statistics cannot be calculated analytically, thus were estimated via MC simulation by sampling $10^8$ $\theta$s and computing $\mu(\theta)$ at each $\theta$ from its analytical expression, which are then used to compute the four nested statistics of interest. Based on these $10^8$ conditional means, the $99\%$-quantile is $\mu_\alpha\approx48.916$, which is why $49$ is chosen as the threshold in the other nested statistics.

\begin{table}
	\centering
	\caption{Single-asset ERM example: MSEs of the nested simulation statistics computed from $1{,}000$ macro runs of the four experiment designs. All methods use the same simulation budget except for SNS$^+$.}\label{tab:ERM}
	
	\resizebox{0.9\textwidth}{!}{%
		\begin{tabu}{c|cccc|cccc}\toprule
			\multicolumn{1}{c|}{\multirow{2}[0]{*}{$M$}} & \multicolumn{4}{c|}{Quantile}                             & \multicolumn{4}{c}{Indicator Function $\zeta(\cdot)$}                                                                          \\\cline{2-9}
			& Opt. Design                                               & SNS                                                   & SNS$^+$  & Regression & Opt. Design & SNS      & SNS$^+$  & Regression \\\cline{1-9}
			128                                          & 23                                                        & 259                                                   & 6.52     & 102        & 8.38E-04    & 8.77E-03 & 4.43E-05 & 5.66E-04   \\
			% 256	&	11.4	&	239	&	2.84	&	61.7	&	1.23E-04	&	6.73E-03	&	1.40E-05	&	2.44E-04	\\
			512                                          & 5.45                                                      & 146                                                   & 1.31     & 37.6       & 1.21E-04    & 3.58E-03 & 4.84E-06 & 1.27E-04   \\
			1024                                         & 2.52                                                      & 109                                                   & 0.503    & 20         & 2.54E-05    & 2.06E-03 & 1.71E-06 & 6.72E-05   \\
			2048                                         & 1.16                                                      & 81.2                                                  & 0.195    & 10.2       & 9.01E-06    & 1.29E-03 & 6.38E-07 & 3.64E-05   \\
			4096                                         & 0.517                                                     & 52.9                                                  & 0.059    & 4.93       & 2.40E-06    & 6.66E-04 & 2.46E-07 & 1.88E-05   \\
			\midrule
			\multicolumn{1}{c|}{\multirow{2}[0]{*}{$M$}} & \multicolumn{4}{c|}{Hockey Stick Function $\zeta(\cdot)$} & \multicolumn{4}{c}{Square Function $\zeta(\cdot)$}                                                                             \\\cline{2-9} & Opt. Design & SNS & SNS$^+$ & Regression & Opt. Design & SNS & SNS$^+$ & Regression\\\cline{1-9}
			128                                          & 5.34E-02                                                  & 8.89E-01                                              & 1.38E-03 & 9.01E-02   & 22          & 541      & 0.229    & 126        \\
			% 256	&	2.05E-03	&	5.65E-01	&	6.75E-04	&	3.43E-02	&	0.228	&	272	&	0.169	&	32.4	\\
			512                                          & 2.29E-03                                                  & 2.28E-01                                              & 2.80E-04 & 1.53E-02   & 0.185       & 81       & 0.108    & 10.6       \\
			1024                                         & 4.81E-04                                                  & 1.09E-01                                              & 1.06E-04 & 6.64E-03   & 0.0875      & 30.3     & 0.0585   & 3.29       \\
			2048                                         & 2.26E-04                                                  & 5.97E-02                                              & 3.87E-05 & 3.15E-03   & 0.0456      & 14.1     & 0.0281   & 1.25       \\
			4096                                         & 8.51E-05                                                  & 2.73E-02                                              & 1.41E-05 & 1.58E-03   & 0.0216      & 5.5      & 0.0125   & 0.543      \\
			\bottomrule
		\end{tabu}
	}
\end{table}

Table~\ref{tab:ERM} presents the MSE of the nested statistics computed from the four experiment designs by running  {$10{,}000$} macro runs.  {Excluding when $M=128$,} the optimal design's MSEs for all nested statistics are lower than those of SNS and the regression approach. Moreover, the optimal design's MSEs are within the same orders of magnitude as those of SNS+ even though the latter requires a much larger simulation budget;  {the optimal design's simulation budget is about $247$ and $1{,}760$ times smaller than that of SNS$^+$ when $M=512$ and $M=4{,}096$, respectively.}
Further, for $M\geq 512,$ the MSE of each risk measure computed from the optimal design also shrinks by approximately a half when $M$ increases by a factor of two, which is consistent with our asymptotic results in Section~\ref{sec:asymptotics}.
% {Moreover, in the first row of each quadrant of Table 2 (e.g., $M=128$), we clearly see that the MSEs for the optimal design do not follow the $\cO(M^{-1})$ convergence rate. We also conducted an experiment with $M=128$ to illustrate that our asymptotic results may not apply when $M$ is small. Though one might expect from our asymptotic results that the MSE would roughly increase by 4 times when the simulation budget decreases from $M=512$ to $M=128$. But we see in the first two rows in Table~\ref{tab:ERM} that the MSEs increases by about 4, 7, 23, and 119 times for quantile, indicator function, hockey stick function, square function, respectively. This observation shows that our asymptotic results only apply for sufficiently large $M$.}

{However, when $M=128$, we observe small-sample performance degradation of the optimal design; the MSE of the probability (indicator $\zeta$) is larger for the optimal design than the regression method. Also, when $M$ shrinks to $128$ from $512$, the MSE of the estimated nested statistics increase by more than $4$ times; notably, $119$ times in the square function case. Since the performance guarantee for the optimal design (and all other benchmarks) is shown asymptotically, such degradation for small $M$ is not surprising. Nevertheless, because the optimal design is formulated to minimize the simulation budget, when $M$ is small and thus the support of the outer scenario is not covered well by the size-$M$ sample, then the optimal design may not be as effective as when $M$ is large. }
%This result alerts that to fully reap the benefit of the optimal design's efficiency, one must choose large enough $M$. }

\subsection{Multi-Assets Enterprise Risk Management (ERM) Problem}\label{subsec:MultiAssetERM}
In this section, we consider a more realistic ERM problem for an option portfolio consists of 10 underlying assets and 300 options.
The underlying assets are 10 non-dividend-paying stocks that follow the multidimensional Black-Scholes model and their respective expected annualized returns are $5\%, 6\%, \ldots, 14\%$, and the respective expected annualized volatilities are $30\%, 32\%,\ldots,48\%$.
The correlations among stock prices are randomly generated using the R library, \texttt{randcorr}.
All stocks have an initial price $S_0=\$100$ at time~$0$ and the risk-free rate is 2\%.
The following 30 options are written on each of the 10 stocks:
\begin{itemize}[leftmargin=*, itemsep=0pt, topsep=0pt]
	\item 10 European options: 5 puts with strike prices $\{85, 90, 95, 100, 105\}$ and 5 calls with strike prices $\{110, 115, 120, 125, 130\}$.
	
	\item 10 fixed-strike geometric Asian options: 5 puts with strike prices $\{85, 90, 95, 100, 105\}$ and 5 calls with strike prices $\{110, 115, 120, 125, 130\}$.
	
	\item 10 down-and-out put barrier options with strike prices $\{85,90,\ldots,130\}$ and lower barriers $\{75,80,\ldots,120\}$. The same type of barrier options are studied in~\cite{broadie2015} to showcase the performances of the regression-based experiment design.
\end{itemize}
All options have maturity $T=2$; all time units are in years.
Stock prices are simulated at discrete time steps, i.e., $h=1/52$ (approx. weekly), and are used to calculate the average prices in Asian options' payoffs.
Also, conditioning on the stock prices at two consecutive times, the minimum price between the two times is simulated via the Brownian bridge~\citep[Chapter 6.4]{glassermanbook}.
Nested simulation is applied to estimate the portfolio's values at time~$\tau=4/52$ (approx. one month) under different scenarios.
The Asian and barrier options are partial-time options, i.e., the Asian options' average prices are calculated between $\tau$ and $T$ and the barrier options can be knocked out only between $\tau$ and $T$.
These settings are similar to those in the numerical examples in~\cite{broadie2015}.

The portfolio value can be decomposed into 10 groups of 3 options written on the same stock.
Then, one can estimate the value of each group using only the simulated prices of the corresponding underlying stock.
This decomposition scheme is inspired by Section~5 in~\cite{hong2017} and is implemented for all experiment designs we compare.

Compared to the single-asset example in Section~\ref{subsec:ERM}, this multi-asset example has more stocks, options, and time steps in both outer- and inner-level simulations.
The proposed optimal nested simulation design can be adapted to these additional complexities:
\begin{enumerate} [leftmargin=*, itemsep=0pt, topsep=0pt]
	\item Exploiting the decomposition of the portfolio, the optimal design can be applied for each underlying stock separately, allowing different simulation budget for each.
	\item As shown in~\ref{app:simplifiedLR}, the LRs are calculated using the one-step transition densities $S_{\tau+h}|S_\tau$ for different scenarios $S_\tau$.
	The ESS is also calculated from them.
	\item The asymptotic properties shown in Section~\ref{sec:asymptotics} are not affected by the portfolio decomposition.
	That is, if they apply to each group of options written on the same stock, they apply to the portfolio value.
\end{enumerate}

We run $K=1{,}000$ macro runs to assess the performances of different experiment designs.
For SNS$^+$, we simulate $M=1{,}000$ outer scenarios and $N=1{,}000$ inner replications at each scenario.
Unlike in Section~\ref{subsec:ERM}, the outer scenarios in different macro runs vary, but the same scenarios are used in the optimal design and SNS$^+$ design in each macro run.
For SNS, we set $M$ and $N$ to be $\lceil\Gamma^{2/3}\rceil$ and $\lceil\Gamma^{1/3}\rceil$, respectively, where $\Gamma$ is the average optimal simulation budget for the 10 stocks.
For the regression-based design, for each stock, we set the number of outer scenarios to be the average of the optimal design's simulation budgets for the 10 stocks, then simulate one inner replication per scenario.
We fit the portfolio value to the weighted Laguerre polynomials up to degree 3 of the 10 underlying stocks' prices, which has 31 explanatory variables including the intercept.
The fitted regression model is used to estimate $\mu(\theta)$ for the same $1{,}000$ outer scenarios as those in the optimal design and SNS$^+$ design.
\begin{table}[tbp]
	\centering
	\caption{Multi-asset ERM example: Average number of outer scenarios, simulation budgets, and CPU time (in seconds) per macro run. The third row shows the number of outer scenarios at which inner replications are made. Speedup factors relative to SNS$^+$ are shown in parentheses.}\label{tab:ERM2Runtime}
	\begin{tabular}{rcccc}\toprule
		& SNS$^+$ & Optimal Design            & SNS                       & Regression                \\\cmidrule{1-5}
		% && Opt. & Sim. & Total &&\\
		% & SNS & Opt: Sim \& LR & Opt: Total & Reg 1 & Reg 2 & Reg 3\\\cmidrule{2-7}
		% \# Inner Sim. & $10^3$ & 375 & 21 &1 & 1\\
		Simulation Budget  & $10^6$  & 8200.8 (122x$\downarrow$) & 8525.7 (117x$\downarrow$) & 8200.8 (122x$\downarrow$) \\
		CPU time (seconds) & 997     & 25.5 (39.16x$\downarrow$) & 28.2 (35.42x$\downarrow$) & 296 (3.37x$\downarrow$)   \\
		\# Outer Scenarios & $10^3$  & 21.9 (46x$\downarrow$)    & 407.3 (2.5x$\downarrow$)  & 8200.8 (8.2x$\uparrow$)   \\
		\bottomrule
	\end{tabular}
\end{table}

Experiments were run on Intel Xeon Gold 6334 8-core 3.6 GHz (Ice Lake) CPUs. Table~\ref{tab:ERM2Runtime} shows the average numbers of outer scenarios, simulation budgets, CPU time, and the corresponding speedup relative to SNS$^+$.
Notice that the optimal design's simulation budget is over 120 times smaller than SNS$^+$'s.
The SNS has slightly larger simulation budget due to rounding.
The CPU time for the optimal design is nearly $1/40$ of the SNS's, which includes the time for solving the LP, LR calculation, and simulation.
For this example, the optimal design's CPU time is evenly split between solving the LP (13.5 seconds), and simulation and LR calculation (12 seconds).
The SNS's CPU time is exclusively for simulation and more than double the simulation and LR calculation time for the optimal design.

To compare the estimation error of all experiment designs, we compute the average MSE (AMSE) of the conditional means in addition to the MSEs of the nested statistics.
\begin{equation*}
	\mbox{AMSE}(\mu) \equiv \frac{1}{K} \sum\nolimits_{k=1}^{K} \frac{1}{M}\sum\nolimits_{m=1}^{M} (\hat{\mu}_{m,k} - \mu_{m,k})^2,
\end{equation*}
where $\hat{\mu}_{m,k}$ is the estimated conditional mean in the $m$th scenario of the $k$th macro run and $\mu_{m,k}$ is the corresponding true conditional mean.

\begin{table}[tbp]
	\centering
	\caption{Multi-asset ERM example: MSEs of nested simulation statistics computed from 1,000 macro runs of the four experiment designs. Numbers in parentheses are the MSE inflation factors compared to the SNS$^+$ design.}\label{tab:ERM2MSEs}
	\begin{tabular}{lccccc}\toprule
		& SNS$^+$  & Opt. Design      & SNS             & Regression       \\\cmidrule{1-5}
		AMSE$(\mu)$                     & 763      & 740 (0.97x)      & 36628 (48x)     & 3052 (4.1x)      \\
		MSE.Quantile                    & 622      & 764 (1.2x)       & 100637 (162x)   & 11298	(18x)      \\
		MSE.Indicator $\zeta(\cdot)$    & 7.23E-06 & 8.51E-06 (1.18x) & 9.22E-03(1275x) & 1.44E-04	(20x)   \\
		MSE.Hockey Stick $\zeta(\cdot)$ & 4.81E-02 & 5.36E-02	(1.12x) & 144 (3004x)     & 2.4 (50x)        \\
		MSE.Square $\zeta(\cdot)$       & 1933     & 1993	(1.03x)     & 7845487	(4060x) & 231{,}413 (120x) \\
		\bottomrule
	\end{tabular}
\end{table}

Table~\ref{tab:ERM2MSEs} summarizes the AMSE and MSEs for different experiment designs.
Notice that the optimal design's AMSE and MSEs closely match those of SNS$^+$. This demonstrates that the ESS constraints in~\eqref{eq:var.constraint} work as intended.
Considering the computational cost comparison in Table~\ref{tab:ERM2Runtime}, the optimal design is extremely efficient.
While SNS has a similar CPU time as the optimal design, its AMSE and MSEs are orders of magnitudes larger than those of the optimal design.
The regression has higher AMSE and MSEs than both SNS$^+$ and the optimal design; although the AMSE is only mildly (4.1x) higher, the MSEs of the nested statistics are orders of magnitude higher.
This implies that the regression-based method has reasonable precision within the central region of the outer scenarios but performs poorly at outer scenarios with extreme conditional means.
This observation is consistent with that made in Section~\ref{subsec:ERM}.

In summary, the multi-asset ERM example demonstrates the applicability of the optimal design in practice and that it outperforms not only SNS but also the regression approach.

\subsection{Input Uncertainty Quantification for Multi-Product Newsvendor Problem}\label{subsec:IU.inventory}
In this section, we consider a single-stage newsvendor problem with ten products. We assume that the $\ell$th product's demand, $X_\ell$, follows a Poisson distribution. All product demands are independent. For the $\ell$th product, let $c_\ell$ and $p_\ell$ be the unit cost and sale price, respectively. The stocking policy, $\{k_1,\ldots,k_{10}\}$, is a vector representing the stocking levels of the ten products.
We chose $p_\ell = 7 + 3\ell$, $c_\ell = 2$, and $k_\ell = 9+\ell$ for all $\ell$ for the experiment.
Given these inputs, the simulator computes the total profit, $g(\bm{X})=\sum_{\ell=1}^{10}\left\{p_\ell \min(X_\ell,k_\ell)-c_\ell k_\ell\right\}$.

Unknown to us, the mean demand of the $\ell$th product is $\vartheta_\ell^c = 5+\ell$. We have {$50+5\ell$} i.i.d.\ realizations from Poisson($\vartheta_\ell^c)$ to estimate $\vartheta_\ell^c$ for the simulation study.
Taking the Bayesian view, we model the unknown parameter $\vartheta_\ell$ as a random variable with a prior distribution and update it with the observations from Poisson($\vartheta_\ell^c)$. To exploit conjugacy, the Gamma prior with rate $0.001$ and shape $0.001$ is adopted for each $\vartheta_\ell$. Then, the posterior distribution of $\vartheta_\ell$ is still Gamma with rate {$0.001 + 50+5\ell$} and shape $0.001$ plus the sum of observed demands of the $\ell$th product.
Let $\btheta = \{\vartheta_1,\ldots,\vartheta_{10}\}$ be a parameter vector sampled from the joint posterior distribution. The expected profit given $\btheta$, $\mu(\btheta)=\E_{\btheta}[g(\bm{X})]$, is a random variable whose distribution is induced by the posterior of $\btheta$.

To measure uncertainty in the model, we construct a $1-\alpha$ credible interval (CrI) for $\mu(\btheta)$ via nested simulation. % in~\cite{xie2014}.
The analytical expression for $\mu(\btheta)$ can be derived easily using the Poisson distribution function. Thus, a CrI can be constructed by sampling $\btheta_1,\ldots,\btheta_M$ from the posterior of $\btheta$ and computing the empirical $\alpha/2$ and $1-\alpha/2$ quantiles from $\mu(\btheta_1),\mu(\btheta_2),\ldots,\mu(\btheta_M)$; this interval is referred to as the \emph{oracle} CrI in the following and used as a benchmark to compare the performances of the algorithms.

Table~\ref{tab:CrI} compares the CrIs constructed by the four experiment designs as well as the oracle CrI from $1{,}000$ macro-runs.
Three different target coverage probabilities, $1-\alpha = 0.9, 0.95$, and $0.99$, are tested.
For each macro-run, a new set of real-world demands are sampled from the true demand distributions and the joint posterior of $\btheta$ is updated conditional on the data.
The oracle CrI is constructed from $M=1{,}000$ $\btheta$s sampled from its posterior. The optimal design and the regression use the same $1{,}000$ $\btheta$s as outer scenarios to construct CrIs. The average of the simulation budget used by the optimal design across $1{,}000$ macro-runs is $1{,}471$ (with standard error $1.1$), which is significantly less than $MN=10^6$ for SNS$^+$. For the regression, polynomial basis functions up to order $2$ were {adopted} without cross-terms reflecting that all product demands are independent.

The empirical coverage probabilities and the widths of CrIs in Table~\ref{tab:CrI} are averaged over $1{,}000$ macro runs.
%\bencomment{Eunhye: The following sentence is not clear. Specifically, the ``independently from the experiment design'' is confusing. I suggest rephrasing it.}
{The former is computed via Monte Carlo simulation: a million $\btheta$s were drawn from the posterior distribution of $\btheta$ and $\mu(\btheta)$s were computed using the analytical expression from which the coverage probability is estimated.
	For all algorithms, the same Monte Carlo sample was used.}

\begin{table}[h!]
	\caption{The estimated coverage probabilities and the widths of CrIs constructed by the oracle, optimal design, SNS, SNS$^+$, and regression from $1{,}000$ macro-runs. All methods use the same simulation budget except for SNS$^+$. The standard errors are in parentheses.}\label{tab:CrI}
	\centering
	\resizebox{0.9\textwidth}{!}{%
		\begin{tabular}{@{}lrrrcrrr@{}}\toprule
			& \multicolumn{3}{c}{Empirical coverage} & \phantom{abc} & \multicolumn{3}{c}{Width}                                                    \\\cmidrule{2-4} \cmidrule{6-8}
			Target $1-\alpha$ & $0.9$                                  & $0.95$        & $0.99$                    &  & $0.9$         & $0.95$        & $0.99$        \\\midrule
			Oracle            & 0.898(3E-04)                           & 0.948(2E-04)  & 0.988(1E-04)              &  & 81.20(2E-03)  & 96.55(3E-03)  & 125.71(5E-03) \\
			Opt. Design       & 0.886(7E-04)                           & 0.940(5E-04)  & 0.985(3E-04)              &  & 81.01(4E-03)  & 96.37(5E-03)  & 125.40(7E-03) \\
			SNS               & 1.000(2E-06)                           & 1.000(3E-07)  & 1.000(1E-08)              &  & 235.38(2E-02) & 277.18(2E-02) & 348.41(3E-02) \\
			SNS$^+$           & 0.912(3E-04)                           & 0.958(2E-04)  & 0.991(8E-05)              &  & 84.80(2E-03)  & 100.87(3E-03) & 131.37(5E-03) \\
			Regression        & 0.972(8E-04)                           & 0.991(4E-04)  & 0.999(1E-04)              &  & 120.28(2E-02) & 146.35(2E-02) & 200.55(3E-02) \\
			\bottomrule
		\end{tabular}
	}
\end{table}

Table~\ref{tab:CrI} shows that the CrIs constructed by the oracle and the optimal design are very close in both coverage and width across for all choices for $\alpha$, although the latter shows a slight undercoverage compared to the former. The undercoverage is caused by that the optimal design interpolates the simulation outputs run at the sampling outer scenarios, however, it is alleviated as $M$ grows.
SNS clearly exhibits overcoverage and wide CrIs. This is because the small inner sample size, $N$, makes $\Var[\estmuMC_i]$ large, which inflates the CrI. Notice that SNS$^+$ still overcovers and slightly wider CrIs than the oracle indicating that the inflation of CrI from MC error of $\estmuMC_i$ persists even with $N=1{,}000$. The optimal design and SNS$^+$ show comparable performances across all $\alpha$s, however, the former costs only $1/670$ of the simulation replications of the latter on average. The regression method shows clear overcoverage across all $\alpha$s compared to the optimal design. In particular, the difference between the CrI widths from the two methods is wider for smaller $\alpha$, which coincides with the observations from the ERM examples; the regression tends to work poorly at predicting $\mu(\btheta)$ for extreme $\btheta$s. On the other hand, the optimal nested simulation design does not suffer from this by allocating more replications to the extreme outer scenarios so that they achieve the same target ESS.
\begin{figure}[h!]
	\centering
	\begin{subfigure}[b]{0.45\textwidth}
		\includegraphics[width=\textwidth]{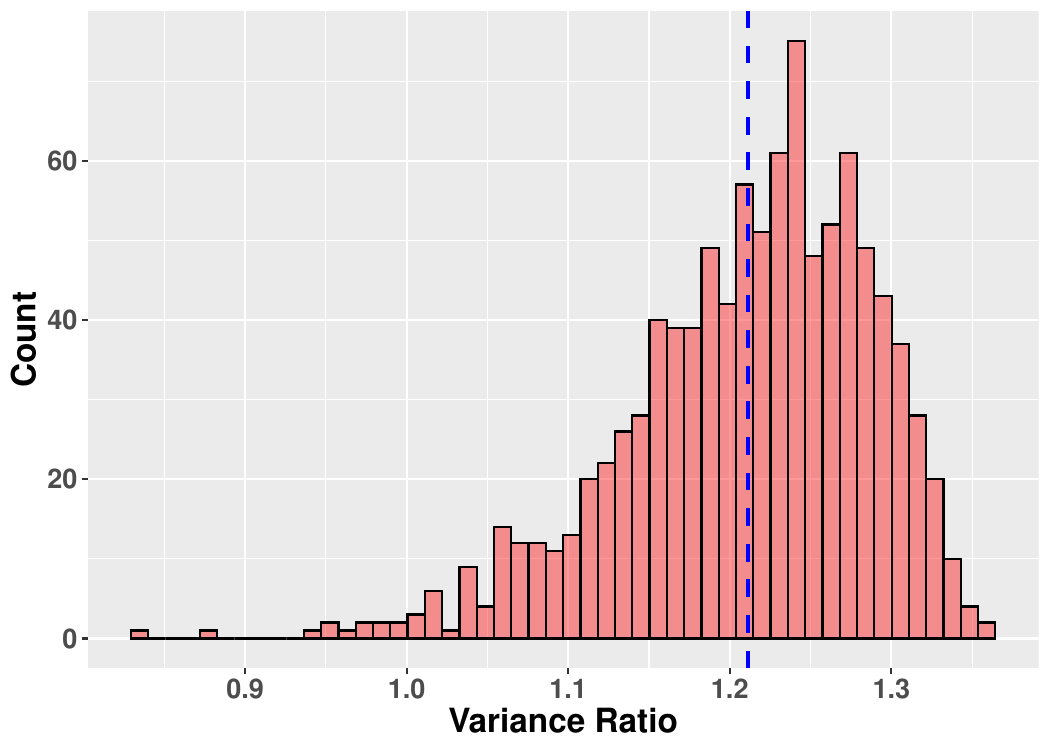}
		\caption{\small}\label{subfig:VarRatio}
	\end{subfigure}
	\hfill
	\begin{subfigure}[b]{0.45\textwidth}
		\includegraphics[width=\textwidth]{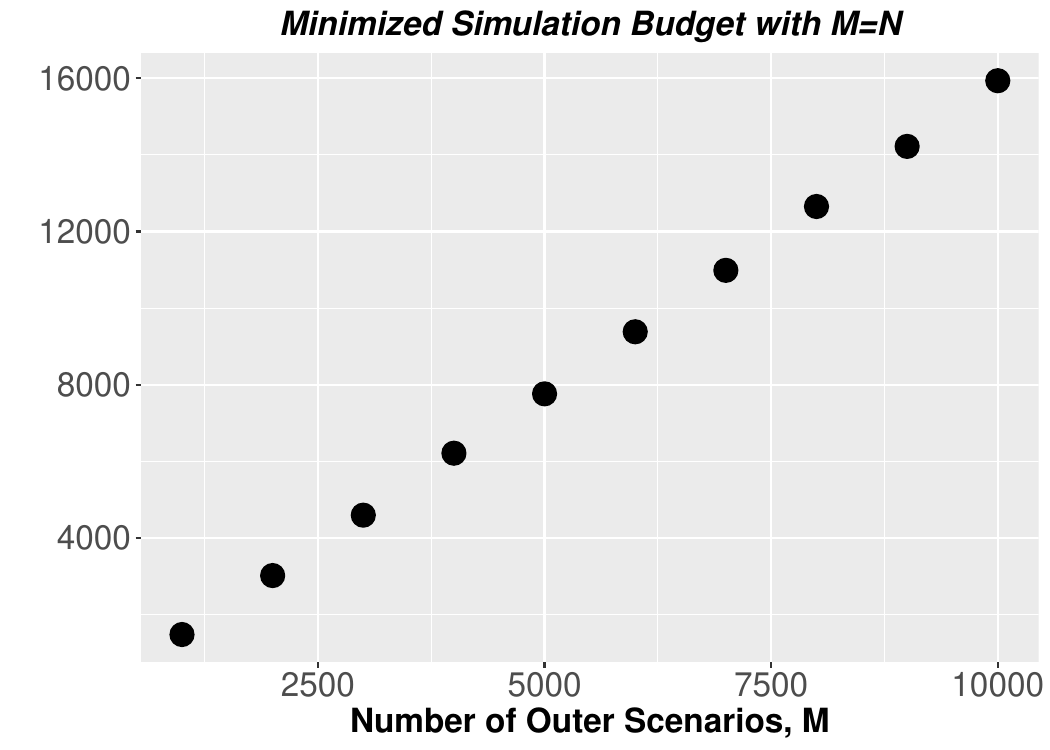}
		\caption{\small}\label{subfig:SimBudget}
	\end{subfigure}
	\caption{(a) is a histogram of ratios between the variances of the $\mu(\theta)$ estimates by the optimal design and by SNS.
		(b) plots the number of outer scenarios $M$ against minimized simulation budget, assuming a constant target inner simulation $N=1000$.}
\end{figure}

We also examine how well our optimal design achieves the target variance of the conditional means, as required by the constraint~\eqref{eq:var.constraint}, which illustrates effectiveness of the ESS expression in approximating the variance.
Figure~\ref{subfig:VarRatio} shows the histogram of the ratios between the variance of the MC estimator, $\Var[\bar{\mu}(\btheta_i)]$, and the variance of the pooled self-normalized LR estimator, $\Var[\widetilde{\mu}(\btheta_i)]$, of the conditional means for a fixed set of 1{,}000 scenarios.
Each MC estimator has $N=1{,}000$ inner replications and the optimal design has the same target $N=1{,}000$. The variance estimates are obtained from $1{,}000$ macro runs of both designs given the same outer scenarios.
Figure~\ref{subfig:VarRatio} shows that the variance ratios are close to 1, with average 1.21 and maximum 1.38.
This shows the ESS expression slightly underestimates the variance of the self-normalized LR estimator, which is a price we pay to approximate the variances without having to simulate the model at all.

Lastly, Figure~\ref{subfig:SimBudget} shows the growth of optimal simulation budget as $M$ increases. In all designs, $N=M$ is chosen for $M$ ranging from $1{,}000$ to $10{,}000$ and the plotted budgets are averaged over $100$ macro runs.
Figure~\ref{subfig:SimBudget} depicts that the minimized simulation budget increases approximately linearly in $M$.
In comparison, the standard nested simulation's budget would be $MN$, so it would increase with $M^2$ when $N=M$.
This comparison once again demonstrates efficiency of our experiment design.

\section{Conclusions}\label{sec:conclusion}
We study an efficient nested simulation experiment design when the set of outer scenarios is given to the experimenter. Using the LR method, our design makes both sampling and pooling decision for the nested simulation. We minimize the total simulation budget while ensuring the (approximate) variance of the conditional mean estimator at each outer scenario is of the same level as the MC estimator's from the SNS design with user-provided $N$.
Our asymptotic analysis shows that the estimation errors of the studied nested statistics achieve the fastest convergence rate when $M=\Theta(N)$.
{When the inner simulation's input distribution belongs in an exponential family, we show that the simulation budget of the optimal design grows in $O(M)$ by setting $N=M$.
	%Numerical studies show that our optimal design offers similar precision as SNS with $M$ outer scenarios and $N$ inner replications as intended, but the minimized simulation budget is hundreds to thousands of times smaller than $MN$.
	Compared to a state-of-the-art regression approach with the same simulation budget, the empirical analyses show that our optimal design delivers significantly smaller estimation error.}

In this paper, we focused on the setting where the vector of all inputs generated within each inner replication has a known and fixed dimension across outer scenarios.
In future research, we will extend the proposed scheme to a more general class of stochastic simulation models.
One example is when the dimension is unknown and random, which is a typical characteristic of a queuing-type simulation model.
For these cases, computing the second moment of the LR without running the simulation model is a challenge, since the joint density of the inputs is unknown a priori. A two-stage, or sequential experiment design may be considered.

\section*{Acknowledgments}
This work was funded by the United States National Science Foundation Grant~CMMI-2045400 and by the Discovery Grants from the Natural Sciences and Engineering Research Council of Canada (NSERC) Grant RGPIN--2018--03755. We thank the anonymous reviewers and editors for thorough and constructive reviews.

%%%%%%%%%%%%%%%%%%%%
\bibliographystyle{chicago}
\bibliography{IUQ_GS}

\begin{thebibliography}{}

\bibitem[\protect\citeauthoryear{Barton, Lam, and Song}{Barton
  et~al.}{2022}]{Barton2022}
Barton, R.~R., H.~Lam, and E.~Song (2022).
\newblock Input uncertainty in stochastic simulation.
\newblock In S.~Salhi and J.~Boylan (Eds.), {\em The Palgrave Handbook of
  Operations Research}, pp.\  573--620. Cham, Switzerland: Springer
  International Publishing.

\bibitem[\protect\citeauthoryear{Ben-Ayed and Blair}{Ben-Ayed and
  Blair}{1990}]{ben1990computational}
Ben-Ayed, O. and C.~E. Blair (1990).
\newblock Computational difficulties of bilevel linear programming.
\newblock {\em Operations Research\/}~{\em 38\/}(3), 556--560.

\bibitem[\protect\citeauthoryear{Broadie, Du, and Moallemi}{Broadie
  et~al.}{2011}]{broadie2011}
Broadie, M., Y.~Du, and C.~C. Moallemi (2011).
\newblock Efficient risk estimation via nested sequential simulation.
\newblock {\em Management Science\/}~{\em 57\/}(6), 1172--1194.

\bibitem[\protect\citeauthoryear{Broadie, Du, and Moallemi}{Broadie
  et~al.}{2015}]{broadie2015}
Broadie, M., Y.~Du, and C.~C. Moallemi (2015).
\newblock Risk estimation via regression.
\newblock {\em Operations Research\/}~{\em 63\/}(5), 1077--1097.

\bibitem[\protect\citeauthoryear{Dang, Feng, and Hardy}{Dang
  et~al.}{2022}]{dang_feng_hardy_2022}
Dang, O., M.~Feng, and M.~R. Hardy (2022).
\newblock Dynamic importance allocated nested simulation for variable annuity
  risk measurement.
\newblock {\em Annals of Actuarial Science\/}, 1–30.

\bibitem[\protect\citeauthoryear{Deb and Sinha}{Deb and
  Sinha}{2009}]{deb2009solving}
Deb, K. and A.~Sinha (2009).
\newblock Solving bilevel multi-objective optimization problems using
  evolutionary algorithms.
\newblock In {\em International conference on evolutionary multi-criterion
  optimization}, pp.\  110--124. Springer.

\bibitem[\protect\citeauthoryear{Dempe}{Dempe}{2002}]{dempe2002foundations}
Dempe, S. (2002).
\newblock {\em Foundations of bilevel programming}.
\newblock Springer Science \& Business Media.

\bibitem[\protect\citeauthoryear{Dong, Feng, and Nelson}{Dong
  et~al.}{2018}]{dong2018unbiased}
Dong, J., M.~Feng, and B.~L. Nelson (2018).
\newblock Unbiased metamodeling via likelihood ratios.
\newblock In {\em Proceedings of the 2018 Winter Simulation Conference}, pp.\
  1778--1789. IEEE.

\bibitem[\protect\citeauthoryear{Elvira, Martino, and Robert}{Elvira
  et~al.}{2018}]{elvira2018rethinking}
Elvira, V., L.~Martino, and C.~P. Robert (2018).
\newblock Rethinking the effective sample size.
\newblock {\em arXiv preprint\/}.
\newblock https://arxiv.org/abs/1809.04129, Last accessed on 6/20/2020.

\bibitem[\protect\citeauthoryear{Feng and Song}{Feng and Song}{2024}]{GitHub}
Feng, B.~M. and E.~Song (2024).
\newblock Efficient nested simulation experiment design via the likelihood
  ratio method.
\newblock Available for download at https://github.com/INFORMSJoC/2022.0392/.

\bibitem[\protect\citeauthoryear{Feng and Song}{Feng and
  Song}{2019}]{feng2019efficient}
Feng, M. and E.~Song (2019).
\newblock Efficient input uncertainty quantification via green simulation using
  sample path likelihood ratios.
\newblock In {\em Proceedings of the 2019 Winter Simulation Conference}, pp.\
  3693--3704. IEEE.

\bibitem[\protect\citeauthoryear{Feng and Staum}{Feng and
  Staum}{2015}]{feng2015green}
Feng, M. and J.~Staum (2015).
\newblock Green simulation designs for repeated experiments.
\newblock In {\em Proceedings of the 2015 Winter Simulation Conference}, pp.\
  403--413. IEEE.

\bibitem[\protect\citeauthoryear{Feng and Staum}{Feng and
  Staum}{2017}]{feng2017green}
Feng, M. and J.~Staum (2017).
\newblock Green simulation: Reusing the output of repeated experiments.
\newblock {\em ACM Transactions on Modeling and Computer Simulation
  (TOMACS)\/}~{\em 27\/}(4), 1--28.

\bibitem[\protect\citeauthoryear{Fu}{Fu}{2015}]{fu2016handbook}
Fu, M.~C. (2015).
\newblock {\em Handbook of Simulation Optimization}.
\newblock New York, New York: Springer.

\bibitem[\protect\citeauthoryear{Giles and Haji-Ali}{Giles and
  Haji-Ali}{2019}]{giles2019}
Giles, M.~B. and A.-L. Haji-Ali (2019).
\newblock Multilevel nested simulation for efficient risk estimation.
\newblock {\em SIAM/ASA Journal on Uncertainty Quantification\/}~{\em 7\/}(2),
  497--525.

\bibitem[\protect\citeauthoryear{Glasserman}{Glasserman}{2003}]{glassermanbook}
Glasserman, P. (2003).
\newblock {\em Monte Carlo Methods in Financial Engineering}.
\newblock Springer.

\bibitem[\protect\citeauthoryear{Glasserman and Xu}{Glasserman and
  Xu}{2014}]{glasserman2014robust}
Glasserman, P. and X.~Xu (2014).
\newblock Robust risk measurement and model risk.
\newblock {\em Quantitative Finance\/}~{\em 14\/}(1), 29--58.

\bibitem[\protect\citeauthoryear{Gordy and Juneja}{Gordy and
  Juneja}{2010}]{gordy2010}
Gordy, M.~B. and S.~Juneja (2010).
\newblock Nested simulation in portfolio risk measurement.
\newblock {\em Management Science\/}~{\em 56\/}(10), 1833--1848.

\bibitem[\protect\citeauthoryear{Ha and Bauer}{Ha and
  Bauer}{2022}]{ha2022least}
Ha, H. and D.~Bauer (2022).
\newblock A least-squares monte carlo approach to the estimation of enterprise
  risk.
\newblock {\em Finance and Stochastics\/}~{\em 26\/}(3), 417--459.

\bibitem[\protect\citeauthoryear{Heidelberger and Lewis}{Heidelberger and
  Lewis}{1984}]{heidelberger1984}
Heidelberger, P. and P.~A.~W. Lewis (1984).
\newblock Quantile estimation in dependent sequences.
\newblock {\em Operations Research\/}~{\em 32\/}(1), 185--209.

\bibitem[\protect\citeauthoryear{Hesterberg}{Hesterberg}{1988}]{hesterberg1988advances}
Hesterberg, T.~C. (1988).
\newblock {\em Advances in Importance Sampling}.
\newblock Ph.\ D. thesis, Stanford University.

\bibitem[\protect\citeauthoryear{Hobbs, Metzler, and Pang}{Hobbs
  et~al.}{2000}]{hobbs2000strategic}
Hobbs, B.~F., C.~B. Metzler, and J.-S. Pang (2000).
\newblock Strategic gaming analysis for electric power systems: An mpec
  approach.
\newblock {\em IEEE transactions on power systems\/}~{\em 15\/}(2), 638--645.

\bibitem[\protect\citeauthoryear{Hong, Juneja, and Liu}{Hong
  et~al.}{2017}]{hong2017}
Hong, L.~J., S.~Juneja, and G.~Liu (2017).
\newblock Kernel smoothing for nested estimation with application to portfolio
  risk measurement.
\newblock {\em Operations Research\/}~{\em 65\/}(3), 657--673.

\bibitem[\protect\citeauthoryear{Huang and Li}{Huang and
  Li}{2001}]{huang2001co}
Huang, Z. and S.~X. Li (2001).
\newblock Co-op advertising models in manufacturer--retailer supply chains: A
  game theory approach.
\newblock {\em European journal of operational research\/}~{\em 135\/}(3),
  527--544.

\bibitem[\protect\citeauthoryear{Kleijnen and Rubinstein}{Kleijnen and
  Rubinstein}{1996}]{kleijnen1996optimization}
Kleijnen, J.~P. and R.~Y. Rubinstein (1996).
\newblock Optimization and sensitivity analysis of computer simulation models
  by the score function method.
\newblock {\em European Journal of Operational Research\/}~{\em 88\/}(3),
  413--427.

\bibitem[\protect\citeauthoryear{Kong}{Kong}{1992}]{kong1992}
Kong, A. (1992).
\newblock A note on importance sampling using standardized weights. university
  of chicago.
\newblock Technical Report 348, Department of Statistics, University of
  Chicago.

\bibitem[\protect\citeauthoryear{L'Ecuyer}{L'Ecuyer}{1990}]{l1990unified}
L'Ecuyer, P. (1990).
\newblock A unified view of the {IPA}, {SF}, and {LR} gradient estimation
  techniques.
\newblock {\em Management Science\/}~{\em 36\/}(11), 1364--1383.

\bibitem[\protect\citeauthoryear{L'Ecuyer}{L'Ecuyer}{1993}]{l1993two}
L'Ecuyer, P. (1993).
\newblock Two approaches for estimating the gradient in functional form.
\newblock In {\em Proceedings of the 1993 Winter Simulation Conference}, pp.\
  338--346. IEEE.

\bibitem[\protect\citeauthoryear{Lee}{Lee}{1998}]{lee1998}
Lee, S. (1998).
\newblock {\em Monte Carlo Computation of Conditional Expectation Quantiles}.
\newblock Ph.\ D. thesis, Stanford University Department of
  Engineering-Economic Systems and Operations Research.

\bibitem[\protect\citeauthoryear{Li and Feng}{Li and Feng}{2021}]{li2021nested}
Li, P. and R.~Feng (2021).
\newblock Nested monte carlo simulation in financial reporting: a review and a
  new hybrid approach.
\newblock {\em Scandinavian Actuarial Journal\/}~{\em 2021\/}(9), 744--778.

\bibitem[\protect\citeauthoryear{Liu}{Liu}{1996}]{liu1996}
Liu, J.~S. (1996, Jun).
\newblock Metropolized independent sampling with comparisons to rejection
  sampling and importance sampling.
\newblock {\em Statistics and Computing\/}~{\em 6\/}(2), 113--119.

\bibitem[\protect\citeauthoryear{Liu and Staum}{Liu and
  Staum}{2010}]{liustaum2010}
Liu, M. and J.~Staum (2010).
\newblock Stochastic kriging for efficient nested simulation of expected
  shortfall.
\newblock {\em Journal of Risk\/}~{\em 12\/}(3), 3--27.

\bibitem[\protect\citeauthoryear{Longstaff and Schwartz}{Longstaff and
  Schwartz}{2001}]{longstaff2001valuing}
Longstaff, F.~A. and E.~S. Schwartz (2001).
\newblock Valuing american options by simulation: a simple least-squares
  approach.
\newblock {\em The review of financial studies\/}~{\em 14\/}(1), 113--147.

\bibitem[\protect\citeauthoryear{Maggiar, Waechter, Dolinskaya, and
  Staum}{Maggiar et~al.}{2018}]{maggiar2018derivative}
Maggiar, A., A.~Waechter, I.~S. Dolinskaya, and J.~Staum (2018).
\newblock A derivative-free trust-region algorithm for the optimization of
  functions smoothed via {G}aussian convolution using adaptive multiple
  importance sampling.
\newblock {\em SIAM Journal on Optimization\/}~{\em 28\/}(2), 1478--1507.

\bibitem[\protect\citeauthoryear{Martino, Elvira, and Louzada}{Martino
  et~al.}{2017}]{martino2017effective}
Martino, L., V.~Elvira, and F.~Louzada (2017).
\newblock Effective sample size for importance sampling based on discrepancy
  measures.
\newblock {\em Signal Processing\/}~{\em 131}, 386--401.

\bibitem[\protect\citeauthoryear{Owen}{Owen}{2013}]{owenbook}
Owen, A.~B. (2013).
\newblock {\em Monte Carlo theory, methods and examples}.
\newblock https://statweb.stanford.edu/~owen/mc/, Last accessed on 6/20/2020.

\bibitem[\protect\citeauthoryear{Risk and Ludkovski}{Risk and
  Ludkovski}{2018}]{risk2018sequential}
Risk, J. and M.~Ludkovski (2018).
\newblock Sequential design and spatial modeling for portfolio tail risk
  measurement.
\newblock {\em SIAM Journal on Financial Mathematics\/}~{\em 9\/}(4),
  1137--1174.

\bibitem[\protect\citeauthoryear{Rockafellar and Uryasev}{Rockafellar and
  Uryasev}{2002}]{rockafellar2002conditional}
Rockafellar, R.~T. and S.~Uryasev (2002).
\newblock Conditional value-at-risk for general loss distributions.
\newblock {\em Journal of banking \& finance\/}~{\em 26\/}(7), 1443--1471.

\bibitem[\protect\citeauthoryear{Rubinstein and Shapiro}{Rubinstein and
  Shapiro}{1993}]{rubinstein1993discrete}
Rubinstein, R.~Y. and A.~Shapiro (1993).
\newblock {\em Discrete event systems: Sensitivity analysis and stochastic
  optimization by the score function method}.
\newblock John Wiley \& Sons Inc.

\bibitem[\protect\citeauthoryear{Sen}{Sen}{1972}]{sen1972}
Sen, P.~K. (1972).
\newblock On the bahadur representation of sample quantiles for sequences of
  $\phi$-mixing random variables.
\newblock {\em Journal of Multivariate Analysis\/}~{\em 2\/}(1), 77 -- 95.

\bibitem[\protect\citeauthoryear{Song, Wu-Smith, and Nelson}{Song
  et~al.}{2020}]{GMVCO}
Song, E., P.~Wu-Smith, and B.~L. Nelson (2020).
\newblock Uncertainty quantification in vehicle content optimization for
  general motors.
\newblock {\em INFORMS Journal on Applied Analytics\/}~{\em 50\/}(4), 225--238.

\bibitem[\protect\citeauthoryear{Sun, Apley, and Staum}{Sun
  et~al.}{2011}]{sun2011}
Sun, Y., D.~W. Apley, and J.~Staum (2011).
\newblock Efficient nested simulation for estimating the variance of a
  conditional expectation.
\newblock {\em Operations Research\/}~{\em 59\/}(4), 998--1007.

\bibitem[\protect\citeauthoryear{von Stackelberg}{von
  Stackelberg}{1934}]{von1934marktform}
von Stackelberg, H. (1934).
\newblock {\em Marktform Und Gleichgewicht}.
\newblock Julius Springer.

\bibitem[\protect\citeauthoryear{Wang, Wang, and Zhang}{Wang
  et~al.}{2022}]{wang2022smooth}
Wang, W., Y.~Wang, and X.~Zhang (2022).
\newblock Smooth nested simulation: Bridging cubic and square root convergence
  rates in high dimensions.
\newblock {\em arXiv preprint arXiv:2201.02958\/}.

\bibitem[\protect\citeauthoryear{Yang and Huang}{Yang and
  Huang}{2005}]{yang2005mathematical}
Yang, H. and H.-J. Huang (2005).
\newblock {\em Mathematical and economic theory of road pricing}.
\newblock Emerald Group Publishing Limited.

\bibitem[\protect\citeauthoryear{Zhang, Feng, Liu, and Wang}{Zhang
  et~al.}{2022}]{zhang2022sample}
Zhang, K., B.~M. Feng, G.~Liu, and S.~Wang (2022).
\newblock Sample recycling for nested simulation with application in portfolio
  risk measurement.
\newblock {\em arXiv preprint arXiv:2203.15929\/}.

\bibitem[\protect\citeauthoryear{Zhou and Liu}{Zhou and
  Liu}{2019}]{zhouliu2019}
Zhou, E. and T.~Liu (2019).
\newblock Online quantification of input model uncertainty by two-layer
  importance sampling.
\newblock https://arXiv:1912.11172, Last accessed on 6/20/2020.

\end{thebibliography}
%%%%%%%%%%%%%%%%%%%%

\newpage
\appendix

\counterwithin{assumption}{section}

\section{Assumptions for Part (ii) of Lemma~\ref{lem:consistency}}\label{app:technicalassumptions}
Let $\overline{W_{ij}g(\bX)}\equiv\sum_{k=1}^{N_j} g(\bX_k){W}_{ij,k}/N_j$ and $\overline{W_{ij}}\equiv\sum_{k=1}^{N_j}{W}_{ij,k}/N_j$.

\begin{assumption}\label{assm:moment.conditions.for.LR}
	Given $\btheta_i,\btheta_j\in\bTheta$,
	\begin{enumerate}[label = (\roman*), itemsep=0pt, topsep=0pt]
		\item ${\E_{\btheta_j}[W_{ij}^4]<\infty}$ and $\E_{\btheta_j}[(W_{ij}g(\bX)-\mu_i)^4]<\infty$
		\item $\E_{\btheta_j}\left[\sup_{N_j}\sup_{B\in(\overline{W_{ij}}, 1)} \frac{N_j^2(\overline{W_{ij}}-1)^2(\overline{W_{ij}g(\bX)}-\mu_i)^2}{B^{4}}\right]<\infty$, and\\
		$\E_{\btheta_j}\left[\sup_{N_j}\sup_{A\in(\overline{W_{ij}g(\bX)},\mu_i), B\in(\overline{W_{ij}}, 1)} \frac{N_j^2A^2(\overline{W_{ij}}-1)^4}{B^{6}} \right]<\infty$
		\item  {$\E_{\btheta_j}\left[\sup_{N_j}\sup_{B,\tilde{B}\in(\overline{W_{ij}}, 1)} \frac{N_j^{3/2}(B-1)(\overline{W_{ij}}-1)(\overline{W_{ij}g(\bX)}-\mu_i)}{\tilde{B}^{3}}\right]<\infty$,\\
			$\E_{\btheta_j}\left[\sup_{N_j}\sup_{A\in(\overline{W_{ij}g(\bX)},\mu_i), B,\tilde{B}\in(\overline{W_{ij}}, 1)} \frac{N_j^{3/2}(A-\mu_i) (\overline{W_{ij}}-1)^2}{\tilde{B}^3} \right]<\infty$, and\\
			$\E_{\btheta_j}\left[\sup_{N_j}\sup_{A,\tilde{A}\in(\overline{W_{ij}g(\bX)},\mu_i), B,\tilde{B}\in(\overline{W_{ij}}, 1)} \frac{N_j^{3/2}\tilde{A}(B-1) (\overline{W_{ij}}-1)^2}{\tilde{B}^{4}} \right]<\infty$.}
	\end{enumerate}
\end{assumption}

Part (ii) of the moment conditions in Assumption~\ref{assm:moment.conditions.for.LR} may appear strong, but in practice it is likely to hold if Part (i) holds because $\sup_{A\in(\overline{W_{ij}g(\bX)},\mu_i)}A\stackrel{a.s.}{\to}\mu_i$ by Part (i) of Lemma~\ref{lem:consistency} and ${\sup_{B\in(\overline{W_{ij}}, 1)}B\stackrel{a.s.}{\to} 1}$ by the strong law of large numbers.  {Moreover, $\E_{\btheta_j}[\overline{W_{ij}}] = 1$ and $\E_{\btheta_j}[\overline{W_{ij}g(\bX)}] = \mu_i$, and the scaling factor, $N_j^2$, is justified by that $\overline{W_{ij}}$ and $\overline{W_{ij}g(\bX)}$ are sample average of $N_j$ i.i.d.~observations given $\btheta_j$ and $\btheta_i$. The $N_j^{3/2}$ scaling factor in Part (iii) is justified similarly.}

\section{Proof of Part (ii) of Lemma~\ref{lem:consistency}}\label{app:proof}
\begin{proof}
	By definition, $\estmuSN_{ij}$ is the ratio between $\overline{W_{ij}g(\bX)}$ and $\overline{W_{ij}}$. Applying the two-dimensional Taylor series expansion to $\estmuSN_{ij}$ at $\left(\E_{\btheta_j}[W_{ij}g(\bX)], \E_{\btheta_j}[W_{ij}]\right)^\top$,
	\begin{align*}
		\estmuSN_{ij} = &
		\frac{\E_{\btheta_j}[W_{ij} g(\bX)]}{\E_{\btheta_j}[W_{ij}]} +
		\begin{pmatrix}
			\tfrac{1}{\E_{\btheta_j}[W_{ij}]} \\
			-\tfrac{\E_{\btheta_j}[W_{ij} g(\bX)]}{(\E_{\btheta_j}[W_{ij}])^2}
		\end{pmatrix}^\top
		\begin{pmatrix}
			\overline{W_{ij}g(\bX)} -\E_{\btheta_j}[W_{ij} g(\bX)] \\
			\overline{W_{ij}} - \E_{\btheta_j}[W_{ij}]
		\end{pmatrix} \\
		& + \frac{1}{2}
		\begin{pmatrix}
			\overline{W_{ij}g(\bX)} -\E_{\btheta_j}[W_{ij} g(\bX)] \\
			\overline{W_{ij}} - \E_{\btheta_j}[W_{ij}]
		\end{pmatrix}^\top
		\begin{pmatrix}
			0      & -1/B^2 \\
			-1/B^2 & 2A/B^3
		\end{pmatrix}
		\begin{pmatrix}
			\overline{W_{ij}g(\bX)} -\E_{\btheta_j}[W_{ij} g(\bX)] \\
			\overline{W_{ij}} - \E_{\btheta_j}[W_{ij}]
		\end{pmatrix},
	\end{align*}
	where $A$ and $B$ are in between $\overline{W_{ij}g(\bX)} $ and $\E_{\btheta_j}[W_{ij} g(\bX)]$, and $\overline{W_{ij}}$ and $ \E_{\btheta_j}[W_{ij}]$, respectively. Because $\E_{\btheta_j}[W_{ij} g(\bX)] = \mu_i$ and $\E_{\btheta_j}[W_{ij}] = 1$, the expansion can be rewritten as
	\begin{equation}\label{eq:estmuSN_taylor}
		\estmuSN_{ij} = \overline{W_{ij}g(\bX)} - \mu_i(\overline{W_{ij}}-1) - \frac{1}{B^2}(\overline{W_{ij}}-1)(\overline{W_{ij}g(\bX)}-\mu_i) + \frac{A}{B^3}(\overline{W_{ij}}-1)^2.
	\end{equation}
	We first show that the variance of $\estmuSN_{ij}$ has the stated expression. From Assumption~\ref{assm:moment.conditions.for.LR}, the second moment of~\eqref{eq:estmuSN_taylor} is bounded. Then, by the dominated convergence theorem,~\eqref{eq:estmuSN_taylor} converges in mean square to
	\begin{equation}\label{eq:lim.in.meansq}
		\overline{W_{ij}g(\bX)} - \mu_i(\overline{W_{ij}}-1) - (\overline{W_{ij}}-1)(\overline{W_{ij}g(\bX)}-\mu_i) + \mu_i(\overline{W_{ij}}-1)^2.
	\end{equation}
	After some tedious algebra, the variance of~\eqref{eq:lim.in.meansq} can be shown to have the following form:
	\begin{align}
		{\E_{\btheta_j}[W_{ij}^2(g(\bX)-\mu_i)^2]}{N_j}^{-1} + \mathcal{R}_1 N_j^{-2} + \mathcal{R}_2 N_j^{-3},
	\end{align}
	where $\mathcal{R}_1$ and $\mathcal{R}_2$ are functions of moments of $W_{ij}$ and $g(\bX)$ bounded under Assumption~\ref{assm:moment.conditions.for.LR}. Thus, $\Var_{\btheta_j}[\estmuSN_{ij}] = \E_{\btheta_j}[W_{ij}^2(g(\bX)-\mu_i)^2]N_j^{-1} + o(N_j^{-1})$.
	
	{Next, we show the bias result. Subtracting~\eqref{eq:lim.in.meansq} from $\tilde{\mu}_{ij}$, we have
		\begin{align}
			& \left(1- \tfrac{1}{B^2}\right)(\overline{W_{ij}}-1)(\overline{W_{ij}g(\bX)}-\mu_i) + \left(\tfrac{A}{B^3}-\mu_i\right)(\overline{W_{ij}}-1)^2 \nonumber \\
			= & \tfrac{2}{\tilde{B}^3}(B-1)(\overline{W_{ij}}-1)(\overline{W_{ij}g(\bX)}-\mu_i)
			+\tfrac{1}{\tilde{B}^3} (A-\mu_i) (\overline{W_{ij}}-1)^2
			- \tfrac{3\tilde{A}}{\tilde{B}^4} (B-1)(\overline{W_{ij}}-1)^2,\label{eq:Taylor.for.bias}
		\end{align}
		where the equality follows from the two-dimensional Taylor series expansion with respect to $(A,B)$, and $\tilde{A} \in (A,\mu_i)$ and $\tilde{B} \in (B,1)$.
		Therefore, from Assumption~\ref{assm:moment.conditions.for.LR}, we have $\E[\eqref{eq:Taylor.for.bias}] =\cO(N_j^{-3/2})=o(N_j^{-1})$.
		Because} $\E_{\btheta_j}[(\overline{W_{ij}}-1)(\overline{W_{ij}g(\bX)}-\mu_i)]$ and $\E_{\btheta_j}[(\overline{W_{ij}}-1)^2]$ are the covariance between two sample means and variance of a sample mean, respectively, the expectation of~\eqref{eq:lim.in.meansq} becomes
	\begin{align*}
		& \mu_i - \E_{\btheta_j}[({W_{ij}}-1)({W_{ij}g(\bX)}-\mu_i)]/N_j + \mu_i\E_{\btheta_j}[({W_{ij}}-1)^2]/N_j                                            \\
		= & \mu_i - (\E_{\btheta_j}[{W_{ij}}^2g(\bX)] - \mu_i)/N_j + \mu_i(\E_{\btheta_j}[W_{ij}^2]-1)/N_j = \mu_i - \E_{\btheta_j}[W_{ij}^2(g(\bX)-\mu_i)]/N_j
	\end{align*}
	{Combining above with $\E[\eqref{eq:Taylor.for.bias}]$}, we get $\E_{\btheta_j}[\estmuSN_{ij}] - \mu_i =  - \E_{\btheta_j}[W_{ij}^2(g(\bX)-\mu_i)] N_j^{-1} + o(N_j^{-1})$.
\end{proof}

\section{Proof of Lemma~\ref{lem:uniform.ecdf}}\label{app:lemmaproof}
\begin{proof}
	From Theorem~\ref{thm:indicator.mse} and Chebyshev's inequality, $|\Phi_{M,N}(\xi)-\Phi(\xi)| = \cO_p(M^{-1/2})+\cO_p(N^{-1/2})  = \cO_p(\min(M,N)^{-1/2})$ for any $\xi\in\real$. To show the convergence rate holds uniformly, we proceed with an argument similar to the Glivenko-Cantelli theorem.
	Let $J$ be an arbitrary positive integer and $-\infty = \xi_0 <\xi_1<\cdots<\xi_J=\infty$ such that $\Phi(\xi_j) - \Phi(\xi_{j-1}) = 1/J$ for all $j=1,\ldots,J$.
	Then, there exists $j\in\{2,\ldots,J\}$ such that $\xi \in [\xi_{j-1},\xi_j]$.
	Note that
	$\Phi_{M,N}(\xi)-\Phi(\xi) \leq \Phi_{M,N}(\xi_j) - \Phi(\xi_{j-1}) = \Phi_{M,N}(\xi_j)-\Phi(\xi_j) + 1/J$, and
	$\Phi_{M,N}(\xi)-\Phi(\xi) \geq \Phi_{M,N}(\xi_{j-1}) - \Phi(\xi_{j}) = \Phi_{M,N}(\xi_{j-1})-\Phi(\xi_{j-1}) - 1/J$.
	Thus, $|\Phi_{M,N}(\xi)-\Phi(\xi)|\leq \max\{|\Phi_{M,N}(\xi_j)-\Phi(\xi_j)|,|\Phi_{M,N}(\xi_{j-1})-\Phi(\xi_{j-1})|\}+{1}/{J}$ and
	\begin{equation}\label{eq:sup.bound.ecdf}
		\sup\nolimits_{\xi\in\real}|\Phi_{M,N}(\xi)-\Phi(\xi)| \leq \max\nolimits_{1\leq j \leq J}\{|\Phi_{M,N}(\xi_j)-\Phi(\xi_j)|\}+1/J.
	\end{equation}
	{Let us fix $\varepsilon >0$ and $J>0$ to be constants. From the definition of $\cO_p$ in Section~\ref{sec:problem}, for each $j$, there exist $c_j>1/J$ and $L_j>0$ such that for all $M,N>L_j$,
		$\Pr\{|\Phi_{M,N}(\xi_j)-\Phi(\xi_j)|>(c_j-1/J){\min(M,N)^{-1/2}}\}<\varepsilon/J$, which in turn implies $\Pr\{|\Phi_{M,N}(\xi_j)-\Phi(\xi_j)|>c_j{\min(M,N)^{-1/2}}-1/J\}<\varepsilon/J$. Consequently, for all $M,N> \max_{1\leq j \leq J} L_j$, we have $\Pr\{\max_{1\leq j \leq J}|\Phi_{M,N}(\xi_j)-\Phi(\xi_j)| > (\max_{1\leq j \leq J}c_j){\min(M,N)^{-1/2}}-1/J\}< \varepsilon$. Combining this with~\eqref{eq:sup.bound.ecdf}, we have
		$\Pr\{\sup\nolimits_{\xi\in\real}|\Phi_{M,N}(\xi)-\Phi(\xi)| >(\max_{1\leq j \leq J}c_j){\min(M,N)^{-1/2}}\} < \varepsilon$ for all $M,N> \max_{1\leq j \leq J} L_j$, which implies $\sup\nolimits_{\xi\in\real}|\Phi_{M,N}(\xi)-\Phi(\xi)|=\cO_p(M^{-1/2})+\cO_p(N^{-1/2})$.} 
\end{proof}

\section{A terse introduction to bi-level optimization}\label{app:bi-levelopt}
A bi-level optimization problem contains another optimization problem within its constraints.
This hierarchical structure arises in various applications, such as transportation (the toll-setting problem~\citep{yang2005mathematical}), economics (e.g., the Stackelberg game~\citep{von1934marktform}), supply chain management~\citep{huang2001co}, energy market design~\citep{hobbs2000strategic} etc., where decisions made at one level affect the feasibility or the optimality of the decisions at another level.

A generic mathematical formulation for a bi-level optimization problem is:
\begin{align*}
	& \min_{x \in \mathcal{X}, y\in \mathcal{Y}} &  & F(x,y(x))                                                                                              \\
	& \st                                        &  & G_i(x,y) \leq 0, \quad i = 1,\ldots, I,                                                                \\
	&                                            &  & \begin{aligned}
		y \in\argmin_{y \in \mathcal{Y}(x)} \{ f(x, y): g_j(x,y) \leq 0, \quad \forall j=1,\ldots,J\},
	\end{aligned}
\end{align*}
where
\begin{itemize}[itemsep=0pt, topsep=0pt]
	\item $F: \mathcal{X}\times \mathcal{Y} \to \real$ is the objective function of the upper-level (master's) problem.
	
	\item $f: \mathcal{X}\times \mathcal{Y} \to \real$ is the objective function of the lower-level (follower's) problem.
	
	\item $x$ and $y$ represent the upper- and lower-level decision variables, respectively.
	
	\item $\mathcal{X}$ and $\mathcal{Y}$ represent the upper- and lower-level feasible solution spaces, respectively. The latter may depend on the upper-level decision, $x$: $\mathcal{Y} = \mathcal{Y}(x)$.
	
	\item $G_i: \mathcal{X}\times \mathcal{Y} \to \real$ and $g_j: \mathcal{X}\times \mathcal{Y} \to \real$ represent the upper- and lower-level constraints, respectively.
\end{itemize}

A distinctive feature of a bi-level optimization problem is that the lower-level problem, i.e., $\min_{y \in \mathcal{Y}(x)} \{ f(x, y): g_j(x,y) \leq 0, \quad \forall j=1,\ldots,J\}$ depends on the upper-level decision, $x$, but is solved only with respect to the lower-level decision variable, $y$.
In general, a bi-level optimization problem can have multiple lower-level problems.
For example, the bi-level optimization problem studied in this paper has $M$ lower-level problems, one for each target distribution $i$; see~\eqref{prob:bi-level} and~\eqref{prob:bi-levelsub}.

In general, solving bi-level optimization problems is more challenging than solving single-level optimization problems. In many cases, standard optimization techniques need to be modified or combined with other methods to handle the nested structure.
Several algorithms and methodologies have been proposed to solve bi-level optimization problems, including the Karush-Kuhn-Tucker (KKT) conditions approach~\citep{dempe2002foundations}, evolutionary algorithms~\citep{deb2009solving}, decomposition methods~\citep{ben1990computational}, etc.
Sometimes one can take advantage of some problem structures to simplify the bi-level optimization problems.
For example, the lower-level problems~\eqref{prob:bi-levelsub} have closed form solutions, thus the bi-level optimization problem~\eqref{prob:bi-level} is simplified as an LP~\eqref{prob:LP}.

\section{Proof of Theorem~\ref{thm:hockey.stick}} \label{app:proof.hockey}
\begin{proof}
	From the definition of $\zeta$, we have
	\begin{align}
		\E[\widetilde{\zeta}-\zeta(\mu_i)] & = \E[(\tmustar_i-\xi) I(\tmustar_i>\xi) - (\mu_i-\xi) I(\mu_i>\xi)] \nonumber                                                                                                                                                                                 \\
		& = \int_{-\infty}^\infty \int_{\xi-{\epsilon}/{\sqrt{N}}}^\infty \left(\mu+\tfrac{\epsilon}{\sqrt{N}}-\xi\right)f_i(\mu,\epsilon)d\mu d\epsilon -\int_{-\infty}^\infty \int_{\xi}^\infty (\mu-\xi) f_i(\mu,\epsilon)d\mu d\epsilon.\label{eq:hockeystick.bias}
	\end{align}
	Note that
	\begin{equation}\label{eq:sim.error.integration}
		\int_{-\infty}^\infty \int_{\xi}^\infty \tfrac{\epsilon}{\sqrt{N}} f_i(\mu,\epsilon)d\mu d\epsilon = \int_{\xi}^\infty \phi(\mu)\E\left[\left.\tfrac{\epsilon}{\sqrt{N}}\right|\mu_i = \mu\right] d\mu = \cO(N^{-1}),
	\end{equation}
	where the last equality holds because $\E\left[\left.\frac{\epsilon}{\sqrt{N}}\right|\mu_i = \mu\right] = \cO(N^{-1})$ uniformly for all $\mu$ as shown in the proof of Theorem~\ref{thm:indicator.mse}.
	Adding and subtracting~\eqref{eq:sim.error.integration} from both sides of~\eqref{eq:hockeystick.bias}, we have
	\[
	eqref{eq:hockeystick.bias}
	=\int_{-\infty}^\infty \int_{\xi-\frac{\epsilon}{\sqrt{N}}}^\xi \left(\mu+\frac{\epsilon}{\sqrt{N}}-\xi\right)f_i(\mu,\epsilon)d\mu d\epsilon
	+ \cO(N^{-1}).
	\]
	From the Taylor series expansion of $f_i(\mu,\epsilon)$ in~\eqref{eq:marginal.taylor},
	\begin{align}\label{eq:temp1}
		\int_{\xi-\frac{\epsilon}{\sqrt{N}}}^\xi \left(\mu+\tfrac{\epsilon}{\sqrt{N}}-\xi\right)f_i(\mu,\epsilon)d\mu = \int_{\xi-\frac{\epsilon}{\sqrt{N}}}^\xi \left(\mu+\tfrac{\epsilon}{\sqrt{N}}-\xi\right)
		\left\{f_i(\xi,\epsilon) + \tfrac{\partial f_i({\Check{\mu}},\epsilon)}{\partial \mu} (\mu-\xi) \right\} d\mu.
	\end{align}
	Note that $\int_{\xi-\frac{\epsilon}{\sqrt{N}}}^\xi \left(\mu+\tfrac{\epsilon}{\sqrt{N}}-\xi\right)f_i(\xi,\epsilon) d\mu = \left[ \frac{1}{2}(\mu-\xi)^2 + \frac{\epsilon}{\sqrt{N}}\mu\right]_{\xi-\frac{\epsilon}{\sqrt{N}}}^\xi f_i(\xi,\epsilon)= \frac{\epsilon^2}{2N}f_i(\xi,\epsilon)$ and
	\[
	\int_{\xi-\frac{\epsilon}{\sqrt{N}}}^\xi \left(\mu+\tfrac{\epsilon}{\sqrt{N}}-\xi\right)(\mu-\xi)\tfrac{\partial f_i({\Check{\mu}},\epsilon)}{\partial \mu} d\mu\leq \left|\left[\frac{1}{3}(\mu-\xi)^3 +\frac{\epsilon}{2\sqrt{N}}(\mu-\xi)^2\right]_{\xi-\frac{\epsilon}{\sqrt{N}}}^\xi\right|p_{1,M,N}(\epsilon)\leq \frac{|\epsilon^3|}{N^{3/2}}p_{1,M,N}(\epsilon).
	\]
	The lower bound can be obtained similarly. Therefore,~\eqref{eq:temp1} is lower/upper-bounded by
	$\frac{\epsilon^2}{2N} f_i(\xi,\epsilon)$ $ \mp \frac{|\epsilon^3|}{N^{3/2}}p_{1,M,N}(\epsilon)$.
	Integrating these bounds once again with respect to $\epsilon \in (-\infty, \infty)$, we have
	\begin{equation}\label{eq:expected.difference.hockey}
		\E[(\tmustar_i-\xi) I(\tmustar_i>\xi) - (\mu_i-\xi) I(\mu_i>\xi)] = \cO(N^{-1}).
	\end{equation}
	The variance of $\widetilde{\zeta}$ can be expanded as
	\begin{equation}\label{eq:var.hockey}
		\frac{1}{M^2}\sum_{i=1}^M \Var[(\tmustar_i-\xi) I(\tmustar_i>\xi)]
		+ \frac{1}{M^2} \sum_{i=1}^M \sum_{j=1, j \neq i}^M \Cov[(\tmustar_i-\xi) I(\tmustar_i>\xi),(\tmustar_j-\xi) I(\tmustar_j>\xi)].
	\end{equation}
	One can see that $ \Var[(\tmustar_i-\xi) I(\tmustar_i>\xi)]=\cO(1)$ because
	\begin{align*}
		\Var[(\tilde{\mu}_i^*-\xi)I(\tilde{\mu}_i^*>\xi)] & \leq
		\E[(\tilde{\mu}_i^*-\mu_i+\mu_i-\xi)^2]                                                                                                                                          \\
		& = \E[\E[(\tilde{\mu}_i^*-\mu_i)^2|\theta_i]] + 2(\E[{\mu}_i]-\xi)\E[\E[\tilde{\mu}_i^*-\mu_i|\theta_i]] + \E[(\mu_i-\xi)^2],
	\end{align*}
	where the first two terms are $\cO(N^{-1})$ from Theorem~\ref{thm:convergencerate} and the last term is a constant.
	Then, the first term of~\eqref{eq:var.hockey} is $\cO(M^{-1})$. Because $\mu_i$ and $\mu_j$ for arbitrary $i\neq j$ are independent, $\Cov[(\mu_i-\xi) I(\mu_i>\xi),(\mu_j-\xi) I(\mu_j>\xi)] = 0$. Thus, the covariance term in~\eqref{eq:var.hockey} is equal to
	\small
	\begin{align}
		& \Cov[(\tmustar_i-\xi) I(\tmustar_i>\xi),(\tmustar_j-\xi) I(\tmustar_j>\xi)] - \Cov[(\mu_i-\xi) I(\mu_i>\xi),(\mu_j-\xi) I(\mu_j>\xi)] \nonumber                     \\
		& = \E[(\tmustar_i-\xi)(\tmustar_j-\xi) I(\tmustar_i>\xi,\tmustar_j>\xi)]- \E[(\mu_i-\xi)(\mu_j-\xi) I(\mu_i>\xi,\mu_j>\xi)]\label{eq:cov.part1}                      \\
		& \;\;\;\; + \E[(\mu_i-\xi)I(\mu_i>\xi)]\E[(\mu_j-\xi) I(\mu_j>\xi)]- \E[(\tmustar_i-\xi)I(\tmustar_i>\xi)]\E[(\tmustar_j-\xi) I(\tmustar_j>\xi)]\label{eq:cov.part2}
	\end{align}
	\normalsize
	From~\eqref{eq:expected.difference.hockey},~\eqref{eq:cov.part2}$=\cO(N^{-1})$. We rewrite~\eqref{eq:cov.part1} as
	\small
	\begin{align}
		& \int_{-\infty}^\infty \int_{-\infty}^\infty \int_{\xi-\frac{\epsilon_i}{\sqrt{N}}}^\infty \int_{\xi-\frac{\epsilon_j}{\sqrt{N}}}^\infty \left(\mu_i+\tfrac{\epsilon_i}{\sqrt{N}}-\xi\right)\left(\mu_j+\tfrac{\epsilon_j}{\sqrt{N}}-\xi\right) f_{ij}(\mu_i,\mu_j,\epsilon_i,\epsilon_j)d\mu_j d\mu_i d\epsilon_j d\epsilon_i \nonumber                                            \\
		& - \int_{-\infty}^\infty \int_{-\infty}^\infty \int_{\xi}^\infty \int_{\xi}^\infty \left(\mu_i-\xi\right)\left(\mu_j-\xi\right) f_{ij}(\mu_i,\mu_j,\epsilon_i,\epsilon_j)d\mu_j d\mu_i d\epsilon_j d\epsilon_i \nonumber                                                                                                                                                            \\
		& = \int_{-\infty}^\infty \int_{-\infty}^\infty \int_{\xi-\frac{\epsilon_i}{\sqrt{N}}}^\infty \int_{\xi-\frac{\epsilon_j}{\sqrt{N}}}^\infty \left((\mu_i-\xi)(\mu_j-\xi) + \tfrac{\epsilon_i}{\sqrt{N}}(\mu_j-\xi) + \tfrac{\epsilon_j}{\sqrt{N}}(\mu_i-\xi) + \tfrac{\epsilon_i\epsilon_j}{N}\right) f_{ij}(\mu_i,\mu_j,\epsilon_i,\epsilon_j)d\mu_j d\mu_i d\epsilon_j d\epsilon_i
		\nonumber                                                                                                                                                                                                                                                                                                                                                                             \\
		& - \int_{-\infty}^\infty \int_{-\infty}^\infty \int_{\xi}^\infty \int_{\xi}^\infty \left(\mu_i-\xi\right)\left(\mu_j-\xi\right) f_{ij}(\mu_i,\mu_j,\epsilon_i,\epsilon_j)d\mu_j d\mu_i d\epsilon_j d\epsilon_i \nonumber                                                                                                                                                            \\
		& =\int_{-\infty}^\infty \int_{-\infty}^\infty \int_{\xi-\frac{\epsilon_i}{\sqrt{N}}}^\xi \int_{\xi-\frac{\epsilon_j}{\sqrt{N}}}^\xi \left(\mu_i-\xi\right)\left(\mu_j-\xi\right) f_{ij}(\mu_i,\mu_j,\epsilon_i,\epsilon_j)d\mu_j d\mu_i d\epsilon_j d\epsilon_i\label{eq:cov.part3}                                                                                                 \\
		& {
			+ 2\int_{-\infty}^\infty \int_{-\infty}^\infty
			\int_\xi^\infty \int_{\xi-\frac{\epsilon_i}{\sqrt{N}}}^\xi (\mu_i-\xi)(\mu_j-\xi)f_{ij}(\mu_i,\mu_j,\epsilon_i,\epsilon_j)d\mu_i d\mu_j d\epsilon_i d\epsilon_j}
		\label{eq:cov.part3-2}                                                                                                                                                                                                                                                                                                                                                                \\
		& + \int_{-\infty}^\infty \int_{-\infty}^\infty \int_{\xi-\frac{\epsilon_i}{\sqrt{N}}}^\infty \int_{\xi-\frac{\epsilon_j}{\sqrt{N}}}^\infty \left\{(\mu_i-\xi)\tfrac{\epsilon_j}{\sqrt{N}}+(\mu_j-\xi)\tfrac{\epsilon_i}{\sqrt{N}}+\tfrac{\epsilon_i\epsilon_j}{N}\right\} f_{ij}(\mu_i,\mu_j,\epsilon_i,\epsilon_j)d\mu_j d\mu_i d\epsilon_j d\epsilon_i.\label{eq:cov.part4}
	\end{align}
	\normalsize
	Using the Taylor expansion in~\eqref{eq:jointpdf.taylor}, the two inner integrals of~\eqref{eq:cov.part3} can be bounded from above and below by
	$\frac{\epsilon_i^2\epsilon_j^2}{N^2} f_{ij}(\xi,\xi,\epsilon_i,\epsilon_j) + \frac{|\epsilon_i^3\epsilon_j^2 \pm \epsilon_i^2\epsilon_j^3|}{N^{5/2}}p_{1,M,N}(\epsilon_i,\epsilon_j)
	$, which yields $\cO(N^{-2})$ when integrated with respect to $\epsilon_i$ and $\epsilon_j$.
	Next, we show that~\eqref{eq:cov.part3-2}$=\cO(N^{-1})$. The first-order Taylor series expansion of $f_{ij}(\mu_i,\mu_j,\epsilon_i,\epsilon_j)$ with respect to $\mu_i\in[\xi-\epsilon_i/\sqrt{N},\xi]$ gives
	\begin{equation}\label{eq:hockey.stick.taylor}
		f_{ij}(\mu_i,\mu_j,\epsilon_i,\epsilon_j) = f_{ij}(\xi, \mu_j,\epsilon_i,\epsilon_j) + \frac{\partial f_{ij}(\check{\mu}_i,\mu_j,\epsilon_i,\epsilon_j)}{\partial \mu_i} (\mu_i-\xi)
	\end{equation}
	for $\check{\mu}_i \in (\mu_i,\xi)$.
	Combining this with Assumption~\ref{assm:hockeystick}, the integrand of~\eqref{eq:cov.part3-2} is upper/lower-bounded by $\pm|\mu_j-\xi|\left\{|\mu_i-\xi|q_{0,M,N}(\mu_j,\epsilon_i,\epsilon_j) + q_{1,M,N}(\mu_j,\epsilon_i,\epsilon_j) (\mu_i-\xi)^2\right\}$, which becomes $\pm|\mu_j-\xi|\left\{\tfrac{\epsilon_i^2}{2N}q_{0,M,N}(\mu_j,\epsilon_i,\epsilon_j) + \tfrac{|\epsilon_i|^3}{3N^{3/2}}q_{1,M,N}(\mu_j,\epsilon_i,\epsilon_j)\right\}$ after integrating over $\mu_i\in[\xi-\epsilon_i/\sqrt{N},\xi]$. From Assumption~\ref{assm:hockeystick}, integrating these bounds with respect to $\mu_j\in[\xi,\infty), \epsilon_i\in(-\infty, \infty)$, and $\epsilon_j\in(-\infty,\infty)$ results in $\eqref{eq:cov.part3-2}=\cO(N^{-1})$.
	
	To show~\eqref{eq:cov.part4}=$\cO(N^{-1})$ we first partition the integration ranges for $\mu_i$ and $\mu_j$ as: (i) $\mu_i \in [\xi-\epsilon_i/\sqrt{N}, \xi], \mu_j \in [\xi-\epsilon_j/\sqrt{N}, \xi]$, (ii) $\mu_i \in [\xi, \infty), \mu_j \in [\xi, \infty)$, (iii) $\mu_i \in [\xi-\epsilon_i/\sqrt{N}, \xi], \mu_j \in [\xi, \infty)$, and (iv) $\mu_i \in[\xi, \infty), \mu_j \in [\xi-\epsilon_j/\sqrt{N}, \xi]$.
	
	\noindent \textbf{Part (i)} Plugging in the Taylor series expansion in~\eqref{eq:jointpdf.taylor} for $f_{ij}(\mu_i,\mu_j,\epsilon_i,\epsilon_j)$ and computing the two inner integrals of~\eqref{eq:cov.part4} for Part (i), we have the lower \& upper bounds,
	$\frac{2\epsilon_i^2\epsilon_j^2}{N^2}f_{ij}(\xi,\xi,\epsilon_i,\epsilon_j) \mp \frac{2|\epsilon_i^2\epsilon_j^3 + \epsilon_i^3\epsilon_j^2|}{N^{5/2}} p_{1,M,N}(\epsilon_i,\epsilon_j)$,
	which yields $\cO(N^{-2})$ when integrated with respect to $\epsilon_i$ and $\epsilon_j$.
	
	\noindent \textbf{Part (ii)} We can change orders of integrals because the ranges for $\mu_i$ and $\mu_j$ no longer depend on $\epsilon_i$ and $\epsilon_j$. Thus, Part (ii) can be rewritten as
	\small
	\begin{align*}
		\int_{\xi}^\infty \int_{\xi}^\infty \left\{ (\mu_i-\xi)\E\left[\left.\frac{\epsilon_j}{\sqrt{N}}\right|\mu_i,\mu_j\right] + (\mu_j-\xi)\E\left[\left.\frac{\epsilon_i}{\sqrt{N}}\right|\mu_i,\mu_j\right]+\E\left[\left.\frac{\epsilon_i\epsilon_j}{{N}}\right|\mu_i,\mu_j\right]
		\right\} \phi(\mu_i)\phi(\mu_j) d\mu_i d\mu_j.
	\end{align*}
	\normalsize
	Because $\E[\frac{\epsilon_j}{\sqrt{N}}|\mu_i,\mu_j]=\cO(N^{-1})$ and $\E[\frac{\epsilon_i\epsilon_j}{N}|\mu_i,\mu_j]=\cO(N^{-1})$ for all $\mu_i$ and $\mu_j$, Part (ii) $=\cO(N^{-1})$.
	
	\noindent \textbf{Part (iii) and (iv)} Because Part (iii) and (iv) are symmetric, it suffices to bound the former.
	Substituting $f_{ij}(\mu_i,\mu_j,\epsilon_i,\epsilon_j)$ with~\eqref{eq:hockey.stick.taylor} and integrating with respect to $\mu_i \in [\xi-\epsilon_i/\sqrt{N},\xi]$ yields the following upper \& lower bounds
	\begin{equation}\label{eq:part4}
		\left\{ \frac{\epsilon_i^2(\mu_j-\xi)}{N} + \frac{3\epsilon_j\epsilon_i^2}{2N^{3/2}} \right\}f_{ij}(\xi,\mu_j,\epsilon_i,\epsilon_j) \pm \left\{\frac{|\epsilon_i^3(\mu_j-\xi)|}{N^{3/2}} + \frac{|\epsilon_j\epsilon_i^3|}{N^2}\right\} q_{1,M,N}(\mu_j,\epsilon_i,\epsilon_j).
	\end{equation}
	Integrating~\eqref{eq:part4} over  {$\mu_j\in [\xi,\infty)$,} $\epsilon_i\in(-\infty,\infty)$~and~$\epsilon_j \in(-\infty,\infty)$, Part~(iv)~$=\cO(N^{-1})$.
	
	Finally, combining Parts (i)--(iv),~\eqref{eq:cov.part2},  {\eqref{eq:cov.part3}, and~\eqref{eq:cov.part3-2}}, $\Cov[(\tmustar_i-\xi) I(\tmustar_i>\xi),(\tmustar_j-\xi) I(\tmustar_j>\xi)]=\cO(N^{-1})$. Therefore, $\Var[\widetilde{\zeta}]=\cO(M^{-1}) + \cO(N^{-1})$.
\end{proof}

\section{Proof of Theorem~\ref{thm:smooth.func}} \label{app:proof.smooth}

\begin{proof}
	From the definition, $\E[\widetilde{\zeta}] = \sum_{i=1}^M\E[\zeta(\tmustar_i)]/M$.
	From the Taylor series expansion, we have
	$\zeta(\tmustar_i) = \zeta\left(\mu_i + \frac{\epsilon_i}{\sqrt{N}}\right) = \zeta(\mu_i) + \zeta^\prime(\mu_i) \frac{\epsilon_i}{\sqrt{N}} + \frac{\zeta^{\prime\prime}(\Check{\mu}_i)}{2} \frac{\epsilon_i^2}{N}$,
	where {$\Check{\mu}_i\in (\tmustar_i,\mu_i)$}. Therefore, $\E[\zeta(\tmustar_i)]-\E[\zeta(\mu_i)]=\E\left[\zeta^\prime(\mu_i)\E\left[\left.\frac{\epsilon_i}{\sqrt{N}}\right|\mu_i \right]\right] + \E[{\zeta^{\prime\prime}(\Check{\mu}_i)}\frac{\epsilon_i^2}{2N}]$.
	Recall that $\E\left[\left.\frac{\epsilon_i}{\sqrt{N}}\right|\mu_i \right] = \cO(N^{-1})$ for all $\mu_i$. Because $\zeta^\prime(\mu_i)$ does not depend on $N$ and $\E[\zeta^\prime(\mu_i)\epsilon_i]$ is bounded by Assumption~\ref{assm:continuous.risk.measure}, $\E\left[\zeta^\prime(\mu_i)\E\left[\left.\frac{\epsilon_i}{\sqrt{N}}\right|\mu_i \right]\right]=\cO(N^{-1})$. Since $\zeta^{\prime\prime}$ and $\E[\epsilon_i^2]$ are bounded, $\E[{\zeta^{\prime\prime}(\Check{\mu}_i)}\frac{\epsilon_i^2}{2N}] = \cO(N^{-1})$. Therefore, $\E[\zeta(\tmustar_i)]-\E[\zeta(\mu_i)]=\cO(N^{-1})$.
	For the variance,
	\begin{equation}
		\label{eq:var.of.continuous.zeta.est}
		\Var[\widetilde{\zeta}] = \tfrac{1}{M^2} \sum\nolimits_{i=1}^M \Var[\zeta(\tmustar_i)] + \tfrac{1}{M^2} \sum\nolimits_{1\leq i\neq j\leq M} \Cov[\zeta(\tmustar_i), \zeta(\tmustar_j)].
	\end{equation}
	Clearly, the first sum of~\eqref{eq:var.of.continuous.zeta.est} is $\cO(M^{-1})$. The covariance term of~\eqref{eq:var.of.continuous.zeta.est} can be written as
	$\Cov[\zeta(\tmustar_i), \zeta(\tmustar_j)] = \E[\zeta(\tmustar_i)\zeta(\tmustar_j)]-\E[\zeta(\tmustar_i)]\E[\zeta(\tmustar_j)]$.
	Because $\E[\zeta(\tmustar_i)]=\E[\zeta(\mu_i)]+\cO(N^{-1})$ as shown above, $\E[\zeta(\tmustar_i)]\E[\zeta(\tmustar_j)] = \E[\zeta(\mu_i)]\E[\zeta(\mu_j)] + \cO(N^{-1})$.
	Moreover,
	\begin{align*}
		\zeta(\tmustar_i)\zeta(\tmustar_j) = & \zeta(\mu_i)\zeta(\mu_j) + \left\{ \zeta(\mu_i)\zeta^\prime(\mu_j)\epsilon_i + \zeta(\mu_j)\zeta^\prime(\mu_i)\epsilon_j \right\} N^{-1/2}                                                                                             \\
		& + \left\{\zeta^\prime(\mu_i)\zeta^\prime(\mu_j)\epsilon_i\epsilon_j+\tfrac{1}{2}{\zeta(\mu_i)\zeta^{\prime\prime}(\Check\mu_j)\epsilon_j^2}+\tfrac{1}{2}{\zeta(\mu_j)\zeta^{\prime\prime}(\Check\mu_i)\epsilon_i^2}\right\}N^{-1} + R,
	\end{align*}
	where $R$ contains $\cO(N^{-3/2})$ terms.
	Note that
	$\E[\zeta(\mu_i)\zeta^\prime(\mu_j)\frac{\epsilon_i}{\sqrt{N}}] = \E[\zeta(\mu_i)\zeta^\prime(\mu_j)\E[\frac{\epsilon_i}{\sqrt{N}}|\mu_i]]=\cO(N^{-1})$.
	Under Assumption~\ref{assm:continuous.risk.measure}, one can verify that the coefficient of $N^{-1}$ is bounded in mean and $\E[R] = \cO(N^{-3/2})$.
	Therefore, from~\eqref{eq:var.of.continuous.zeta.est}, $\Var[\widetilde{\zeta}] = \cO(M^{-1}) + \cO(N^{-1})$.
\end{proof}

\section{Examples for construction of $\mathbf{\Theta}_{M^*}$ in Theorem~\ref{thm:linear.growth.obj.val}}\label{app:examples}

\noindent\textbf{Example 1:} Consider an exponential distribution, whose canonical form pdf is $h(x;\theta) = -\theta\exp(\theta x), x>0$ with natural parameter $\theta$ in $\bar{\Theta} = (-\infty, 0)$ and $A(\theta) = - \ln (-\theta)$.
% The log-partition function is $A(\theta) = -\ln(-\eta)$ with the natural parameter, $\eta = -\theta$.
Then, by Lemma~\ref{lem:exp.EW2}, we have $\E_{\theta_j}[W_{ij}^2] = \frac{\theta_i^2}{\theta_j(2\theta_i-\theta_j)}$, if $2\theta_i-\theta_j<0$, and $\E_{\theta_j}[W_{ij}^2] = \infty$, otherwise.
For any fixed $\theta_j<0$, $\mathcal{N}(\theta_j) = \{\theta_i<0|\E_{\theta_j}[W_{ij}^2]<2\}$ can be found by solving
$
\E_{\theta_j}[W_{ij}^2] = \frac{\theta_i^2}{\theta_j(2\theta_i-\theta_j)} = 2 \Rightarrow \theta_i^2 - 4\theta_i\theta_j + 2\theta_j^2 = 0 \Rightarrow \theta_i = \left(2 \pm \sqrt{2}\right)\theta_j.
$
This means that $\E_{\theta_j}[W_{ij}^2] < 2$ for all $\theta_i \in \left((2+\sqrt{2})\theta_j,(2-\sqrt{2})\theta_j\right)$.
%For ease of construction, we consider half-open intervals $\mathcal{N}(\theta_j) \equiv \left[\theta_j,(2+\sqrt{2})\theta_j\right)$.
%To cover the bounded intervals $[a,b]$ by a collection of such half-open neighborhoods, we define sampling parameters $\theta_{j_k} = (2+\sqrt{2})^{k-1}a$ for some $k$'s.
Suppose we have $\Theta = [a, b]$, for $-\infty<a<b<0$. We can construct $\Theta_{M^*}$ as follows:
let $\theta_{j} = (2+\sqrt{2})^{j-1}b, j=1,\ldots,M^*,$ where $M^*$ is the smallest positive integer satisfying $(2+\sqrt{2})^{M^*}b < a \leq (2+\sqrt{2})^{M^*-1}b$, i.e., $M^* = \lceil \frac{\ln (-a) - \ln (-b)}{\ln(2+\sqrt{2})} + 1\rceil$, and $\mathcal{N}(\theta_j) = \left(\left(2+\sqrt{2}\right)^{j}b,(2-\sqrt{2}){\left(2+\sqrt{2}\right)}^{j-1}b\right)$, for $j=1,\ldots,M^*$.
Then, by construction, we have $\bigcup_{j=1}^{M^*}\mathcal{N}(\theta_{j}) \supset [a,b]$, as desired.

\noindent\textbf{Example 2:} For a normal distribution with fixed variance $\sigma >0$, the natural parameter is $\theta = \frac{\mu}{\sigma}$ in $\bar{\Theta}=(-\infty,\infty)$ and $A(\theta) = \frac{\theta^2}{2}$.
%normal distribution density $h(x;\theta) = \frac{1}{\sqrt{2\pi \sigma^2}}\exp\left(-\frac{(x-\theta)^2}{2\sigma^2}\right)$,
% and mean parameter $\theta$ t andat lies within $\bTheta = [a,b]$, for $-\infty<a<b<\infty$. The log-partition function is given as $A(\theta) = \frac{\eta^2}{2}$ with the natural parameter, $\eta = \frac{\theta}{\sigma}$.
By Lemma~\ref{lem:exp.EW2}, we have $\E_{\theta_j}[W_{ij}^2] = \exp\left((\theta_i-\theta_j)^2\right)$ for any $\theta_i,\theta_j \in \real$. To find $\mathcal{N}(\theta_j)$ for fixed $\theta_j \in \real$, we first solve
$
\E_{\theta_j}[W_{ij}^2] = \exp\left((\theta_i-\theta_j)^2\right) = 2 \Rightarrow (\theta_i-\theta_j)^2 = \ln 2 \Rightarrow \theta_i = \theta_j \pm \ln 2.
$
Thus, $\E_{\theta_j}[W_{ij}^2] < 2$ for all $\theta_i \in \left(\theta_j - \ln 2,\theta_j + \ln 2\right)$.
%For ease of construction, we consider half-open intervals $\mathcal{N}(\theta_j) \equiv \left[\theta_j,\theta_j + \sigma^2 \ln 2\right)$.
%To cover the bounded intervals $[a,b]$ by a collection of such half-open neighborhoods, we define sampling parameters
Suppose $\bTheta = [a,b]$, for $-\infty<a<b<\infty$. We can construct $\Theta_{M^*}$ as follows:
let $\theta_{j} = a + (j-1) (\ln 2)$ for $j=1,\ldots,M^*$, where $M^*$ is the smallest positive integer satisfying $a + (\ln 2)(M^*-1) \leq b < a + (\ln 2)M^*$, i.e., $M^* = \lceil \frac{ b - a}{ \ln 2} + 1\rceil$, and $\mathcal{N}(\theta_j) = \left(a+ (j-2) \ln 2,a+ j \ln 2\right)$.
Then, we have $\bigcup_{j=1}^{M^*}\mathcal{N}(\theta_{j}) \supset [a,b]$.

\section{Simplified likelihood ratio calculation for the multi-asset ERM example in Section~\ref{subsec:MultiAssetERM}}\label{app:simplifiedLR}
In this example, the discounted portfolio payoff depends on the whole stock price path between $\tau$ and $T$ or even on additional random variables such as minimum stock prices between two time points to calculate barrier option payoffs.
Below, we show that the LR of the whole inner sample path between two outer scenarios can be simplified to a ratio of two one-step transition densities between time $\tau$ and $\tau+h$ thanks to the Markovian nature of the Black-Scholes model.

Let $\bm{S}_\tau = (S_0, S_h, \ldots, S_{\tau-h}, S_{\tau})$ be an outer stock scenario and $\bm{S}_T|\bm{S}_{\tau} = (S_{\tau+h}, S_{\tau+2h},\ldots, S_{T-h}, S_T)$ be the simulated inner price path given the scenario.
Let $h(\bm{S}_T|\bm{S}_\tau)$ be the transition density function for two consecutive stock price path segments, which is determined by the asset model.
For a Markov asset model such as the Black-Scholes model (i.e., geometric Brownian motion), we have $h(\bm{S}_T|\bm{S}_\tau) = h(\bm{S}_T|S_\tau)$.
Suppose one wishes to reuse an inner price path $\bm{S}_T$ sampled from scenario $\bm{S}_{\tau}^{j}$ to estimate the conditional mean for scenario $\bm{S}_{\tau}^{i}$, the appropriate LR is given by
\begin{align*}
	W_{ij}(\bm{S}_T) & = \frac{h(\bm{S}_T|\bm{S}_\tau^{i})}{h(\bm{S}_T|\bm{S}_\tau^{j})} = \frac{h(S_{\tau+h}, S_{\tau+2h},\ldots, S_{T-h}, S_T|S_\tau^{i})}{h(S_{\tau+h}, S_{\tau+2h},\ldots, S_{T-h}, S_T|S_\tau^{j})},                                               & \quad\mbox{by Markov property,} \\
	& = \frac{h(S_{\tau+h}|S_\tau^{i})h(S_{\tau+2h}|S_{\tau+h},S_\tau^{i})\cdots h(S_{T}|S_{T-h}, \ldots, S_{\tau+h},S_\tau^{i})}{h(S_{\tau+h}|S_\tau^{j})h(S_{\tau+2h}|S_{\tau+h},S_\tau^{j})\cdots h(S_{T}|S_{T-h}, \ldots, S_{\tau+h},S_\tau^{j})}, & \quad\mbox{by Bayes rule,}      \\
	& = \frac{h(S_{\tau+h}|S_\tau^{i})h(S_{\tau+2h}|S_{\tau+h})\cdots h(S_{T}|S_{T-h})}{h(S_{\tau+h}|S_\tau^{j})h(S_{\tau+2h}|S_{\tau+h})\cdots h(S_{T}|S_{T-h})},                                                                                     & \quad\mbox{by Markov property,} \\
	& = \frac{h(S_{\tau+h}|S_\tau^{i})}{h(S_{\tau+h}|S_\tau^{j})},                                                                                                                                                                                     & \quad\mbox{by cancellation.}
\end{align*}

Let $M_t = \max_{t< t' \leq t+h} S_{t'}$ be the maximum stock price between $t$ and $t+h$.
Conditional on two consecutive stock price $S_t$ and $S_{t+h}$, $M_t$ can be simulated from a Brownian bridge method~\citep[Chapter 6.4]{glassermanbook}.
In other words, conditioning on $S_t$ and $S_{t+h}$, the distribution of $M_t$ is independent of any other random variables.
Let $\bm{M}_T = (M_{\tau+h}, M_{\tau+2h},\ldots,M_{T-h})$ be the maximum stock prices between consecutive inner simulation stock prices.
Suppose one wishes to reuse an inner price path $(\bm{S}_T, \bm{M}_T)$ sampled from scenario $\bm{S}_{\tau}^{j}$ to estimate the conditional mean for scenario $\bm{S}_{\tau}^{i}$, the appropriate LR is given by
\begin{align*}
	W_{ij}(\bm{S}_T,\bm{M}_T) & = \frac{h(\bm{S}_T|\bm{S}_\tau^{i})h(\bm{M}_T|\bm{S}_T,\bm{S}_\tau^{i})}{h(\bm{S}_T|\bm{S}_\tau^{j})h(\bm{M}_T|\bm{S}_T,\bm{S}_\tau^{j})}, & \quad \mbox{by Bayes rule,}                                           \\
	& = \frac{h(S_{\tau+h}|S_\tau^{i})}{h(S_{\tau+h}|S_\tau^{j})}\frac{h(\bm{M}_T|\bm{S}_T)}{h(\bm{M}_T|\bm{S}_T)},                              & \quad\mbox{from $W_{ij}(\bm{S}_T)$ and the Brownian bridge property,} \\
	& = \frac{h(S_{\tau+h}|S_\tau^{i})}{h(S_{\tau+h}|S_\tau^{j})},                                                                               & \quad\mbox{by cancellation.}
\end{align*}

\end{document}